\documentclass[11pt,english]{article}
\usepackage{lmodern}
\usepackage[T1]{fontenc}
\usepackage[latin9]{inputenc}

\usepackage{geometry}
\usepackage{pdflscape} 
\usepackage{rotating}  

\usepackage{mathrsfs}
\usepackage{amsmath}
\usepackage{amssymb}
\usepackage{amsthm}
\usepackage{physics}
\usepackage{bm}
\usepackage{latexsym}
\usepackage{esint}
\usepackage{authblk}
\usepackage{graphicx}
\usepackage{subfig}
\usepackage{graphics}
\graphicspath{{./figures/}}

\usepackage{mathtools}
\DeclarePairedDelimiter{\ceil}{\lceil}{\rceil}

\usepackage{xfp}
\makeatletter
\define@key{Gin}{sqrtofarea}{%
    \def\Gin@req@sizes{%
      \edef\Gin@scalex{\fpeval{#1/sqrt(\Gin@nat@height*\Gin@nat@width)}}%
      \let\Gin@scaley\Gin@exclamation
      \Gin@req@height\Gin@scalex\Gin@nat@height
      \Gin@req@width\Gin@scalex\Gin@nat@width
      }%
  \@tempswatrue}
\makeatother

\usepackage{afterpage}

\usepackage{titlesec}
\titlespacing*{\section} {0pt}{3ex plus 1ex minus .2ex}{1.5ex plus .2ex}
\titlespacing*{\subsection} {0pt}{2.5ex plus 1ex minus .2ex}{1.25ex plus .2ex}
\titlespacing*{\subsubsection}{0pt}{2.25ex plus 1ex minus .2ex}{1ex plus .2ex}
\titlespacing*{\paragraph} {0pt}{2.5ex plus 1ex minus .2ex}{1em}

\usepackage[labelfont=bf,tableposition=top,figureposition=bottom]{caption}

\DeclareCaptionLabelSeparator*{spaced}{\\[2ex]}
\usepackage{booktabs}
\usepackage{multirow}

\usepackage{rotating}

\usepackage{enumitem}
\usepackage{pdfpages}
\usepackage{comment}

\usepackage{hyperref}
\usepackage{url}
\usepackage{color}

\usepackage{float}
\usepackage{placeins}
\usepackage{lscape}
\usepackage[english]{babel}

\usepackage[lined,boxed,linesnumbered,commentsnumbered,ruled,vlined]{algorithm2e}
\SetAlFnt{\smaller}
\usepackage{floatpag}

\usepackage{lineno}
\usepackage[toc,page]{appendix}
\usepackage{titlesec}
\usepackage{pdfpages}
\usepackage{eso-pic}
\usepackage{everyshi}

\usepackage[noabbrev,capitalise]{cleveref}

\theoremstyle{plain}
\newtheorem{thm}{Theorem}
\theoremstyle{plain}
  \newtheorem{lem}[thm]{Lemma}  
\theoremstyle{plain}
\newtheorem{prop}[thm]{Proposition}                
\newtheorem{cor}[thm]{Corollary}

\theoremstyle{definition}
\newtheorem{Def}[thm]{Definition}

\theoremstyle{remark}
      
\newtheorem{rem}{Remark}

\theoremstyle{observ}
\newtheorem{observ}[thm]{Observation}

\def\minimize{\mathop{\rm minimize}}
\def\maximize{\mathop{\rm maximize}}
\def\argmin{\mathop{\rm arg\,min}}
\def\Area{\mathop{\rm \textbf{Area}}}
\def\AR{\mathop{\rm AR}}
\def\Vol{\mathop{\rm \textbf{Vol}}}
\def\diam{\mathop{\rm diam}}

\def\dom{\mathop{\rm \textbf{dom}}}

\newcommand{\st}{\textnormal{s.t.}}

\def\Conv{\mathop{\rm \textbf{Conv}}}

\def\dist{\mathop{\rm \textbf{dist}}}
\def\proj{\mathop{\rm \textbf{proj}}}
\def\relint{\mathop{\rm \textbf{relint}}}
\def\interior{\mathop{\rm \textbf{int}}}
\def\refl{\mathop{\rm refl}}
\def\rot{\mathop{\rm rot}}
\def\dwidth{\mathop{\rm dwidth}}

\newcommand{\etal}{{et al.}\xspace}

\providecommand{\keywords}[1]
{
  \small	
  \textbf{\textit{Keywords:}} #1
}

\usepackage{mathpazo}

\usepackage{setspace}

\title{Largest Inscribed Rectangles in Geometric Convex Sets}

\author{Mehdi Behroozi}

\affil{\footnotesize Department of Mechanical and Industrial Engineering, Northeastern University\\m.behroozi@neu.edu}

\date{}

\begin{document}
\maketitle

\begin{abstract}
This paper considers the problem of finding maximum volume (axis-aligned) inscribed boxes in a compact convex set, defined by a finite number of convex inequalities, and presents optimization and geometric approaches for solving them. Several optimization models are developed that can be easily generalized to find other inscribed geometric shapes such as triangles, rhombi, and squares. To find the largest axis-aligned inscribed rectangles in the higher dimensions, an interior-point method algorithm is presented and analyzed. For 2-dimensional space, a parametrized optimization approach is developed to find the largest (axis-aligned) inscribed rectangles in convex sets. The optimization approach provides a uniform framework for solving a wide variety of relevant problems. Finally, two computational geometric $(1-\varepsilon)$--approximation algorithms with sub-linear running times are presented that improve the previous results. 

\end{abstract}


\section{Introduction}
In the context of computational geometry and geometric optimization, working with some geometric shapes, in the practical sense, is usually easier than others. For example, compare working with a regular polygon (equiangular and equilateral) versus a non-regular polygon, a simple polygon (not self-intersecting) vis-{\`a}-vis a self-intersecting polygon, a monotone polygon compared to a non-monotone polygon, or a convex polygon versus a non-convex polygon. Similarly, in many applications of geometric optimization, it is common to approximate the value of an objective function over a convex polygonal region with its value over a simpler approximating shape. Two well-studied categories of such approximations are inner or outer approximations which refer to the cases where the simpler approximated shape is inscribed in the convex polygon or it encloses it. Approximating a convex polygon with L{\"o}wner--John ellipsoids \cite{john1948extremum,john2014extremum,henk2012lowner} or inner and outer boxes of P{\'o}lya and Szeg{\"o} \cite{polya1951isoperimetric} are the prime examples of such approximations. This paper studies the problem of approximating a convex set $C$, not necessarily a polygon, by its maximum volume inscribed rectangle $R$. Figure \ref{fig:LIR-image} illustrates an instance in two-dimensional space (2D), where $C$ is a convex polygon and the $R$ is the desired rectangle. Practical applications of this kind of approximation arise in the computer vision, apparel industry, footwear manufacturing, aluminum container production, steel/aluminum foil cutting, glass sheet cutting, sail manufacturing, carpet cutting, upholstery production, and many other industries. For example, in the apparel industry, the problem is to lay out small polygonal apparel pattern pieces in the unused parts of a rectangular sheet of cloth, called ``marker", after laying out the larger pattern pieces to minimize waste \cite{milenkovic1991automatic,milenkovic1992placement}. They compute the largest axis-aligned rectangle inside each trim piece (unused part) to work with a nicer geometric shape (see section 7 in \cite{daniels1993finding}). In terms of application, the considered problem in this paper is in the same spirit  
of some closely-related categories of geometric optimization problems such as packing, covering, and tiling --- generally focused on minimizing waste. These related problems in these categories include: cutting stock; knapsack; bin packing; guillotine; disk covering; polygon covering; kissing number; strip packing; square packing; squaring the square; squaring the plane; and, in three-dimensional space (3D), cubing the cube and tetrahedron packing. 
\begin{figure}[t]
\captionsetup{farskip=0pt}
 \begin{centering}
  \includegraphics[width=0.45\textwidth]{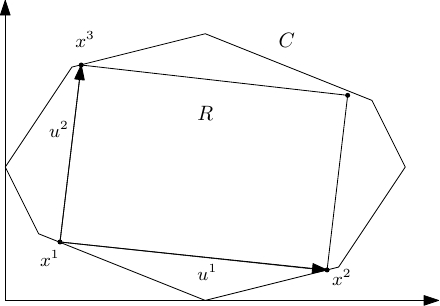}
  \caption{\protect
  An instance of the maximum volume inscribed box inside a convex set in 2D where the set $C$ is a convex polygon. The largest area rectangle $R$ can be determined either by three vertices $x^1,x^2,x^3$ or by the vertex $x^1$ and two vectors $u^1 = x^2 - x^1$ and $u^2 = x^3 - x^1$. 
  }
  \label{fig:LIR-image}
  \end{centering}
\end{figure}

\subsection*{Contributions and Organization of the Paper}
This paper considers the problem of finding the largest  
boxes inside a convex set defined by a finite number of convex inequalities. Throughout this paper, for simplicity and due to the frequency of use, the term ``rectangle'' is used for any box $R\in \mathbb{R}^d$. 
We are interested in finding 
the \emph{maximum volume/area inscribed rectangle} (MVIR / MAIR), the \emph{maximum volume/area axis-aligned inscribed rectangle} (MVAIR / MAAIR), and the \emph{maximum area axis-aligned inscribed rectangle in a fixed direction} (MAAIR-$\mbox{F}_{\mbox{dir}}$). The same acronyms are used for the problem of finding them. The main contributions, listed in presentation order, are:
\begin{enumerate}
\item \textbf{Optimal properties of the MAIR in convex polygons} are discussed and proved (\cref{subseq:MAIR_optProperties_Polygon}). These properties are stronger and better formulated than the current results and yet with simpler proofs. 

\item \textbf{Optimal properties of the MAIR in centrally symmetric and axially symmetric convex sets} are discussed and proved (\cref{subseq:MAIR_optProperties_CentrallySymmetric} and \cref{subseq:MAIR_optProperties_AxiallySymmetric}). To the best of our knowledge this is the first attempt in analyzing the properties of the MAIR in a convex set that is not necessarily a polygon. 

\item \textbf{Optimization models for the  
MVIR, and MVAIR problems} are developed (\cref{sec:OptimizationModels}). To our knowledge, this is the first comprehensive optimization approach to these problems in higher dimensions, which can bring new insights and open a new stream of research in this area. This approach is easily generalizable to other inscribed shapes such as triangles, rhombi, and squares.

\item \textbf{An interior-point method algorithm} is used to solve the MVAIR problem (\cref{subsec:SolvingMVAIR}). Full convergence analysis and computational complexity results are provided. The running time of our optimization-based $(1-\varepsilon)$--approximation algorithm is $\mathcal{O}((d^3+d^2 n) \sqrt{n}\log \frac{n}{\varepsilon})$. To our knowledge, this is the first such algorithm and analysis that is presented for this problem.

\item \textbf{A parametrized optimization approach} is used to solve the MAAIR-$\mbox{F}_{\mbox{dir}}$, MAAIR, and MAIR problems (\cref{subsec:SolvingMVIR}).  To our knowledge, this is the first parametrized optimization approach to solve these problems. Our algorithm for solving the MAAIR in any given direction, i.e., finding the largest rectangles aligned to any rotated axes without rotating them, in any convex set can find 
a $(1-\varepsilon)$--approximation solution in   
$\mathcal{O}(n\sqrt{n}\log \frac{n}{\varepsilon})$ time. Our $(1-\varepsilon)$--approximation algorithm for solving the MAIR runs in 
$\mathcal{O}(\varepsilon^{-1} n \sqrt{n}\log \frac{n}{\varepsilon})$ time.

\item \textbf{Two upper bounds on the aspect ratio of the MAIR} in a convex set are derived (\cref{subsubsec:Approx_MAIR}) that can be computed efficiently. We prove that the aspect ratio of the MAIR in a convex set is bounded from above by a multiple of aspect ratio of the convex set or a multiple of the aspect ratio of a minimum area enclosing rectangle in a particular direction. The bounds are quite loose as we just needed an upper bound in our proofs. However, to the best of our knowledge, this is the first such result and we expect that the tightness of these bounds could be improved.

\item \textbf{Two computational geometric algorithms} are developed for the 2-dimensional space (\cref{subsubsec:Family-Approx-MAIR}). When the convex set $C$ is a polygon defined by $n$ linear inequalities, our first geometric algorithm can find the MAIR in $\mathcal{O}(\varepsilon^{-1}\log n)$ time. Our second geometric algorithm finds the MAIR in any compact convex set in $\mathcal{O}(\varepsilon^{-1/2}T_C+\varepsilon^{-1}\log\varepsilon^{-1/2})$ time, where $T_C$ is the time needed to perform two different queries on $C$ due to \cite{ahn2006inscribing} and will be discussed in Section \ref{subsubsec:Family-Approx-MAIR}. For MAIR in a convex polygon this is $\mathcal{O}(\varepsilon^{-1/2}\log n+\varepsilon^{-1}\log\varepsilon^{-1/2})$. Both of these algorithms improve the previous best results by Cabello \etal \cite{cabello2016finding}, who present a deterministic $(1-\varepsilon)$--approximation algorithm with running time $\mathcal{O}(\varepsilon^{-3/2}+\varepsilon^{-1/2}T_C)$ for a convex set and $\mathcal{O}(\varepsilon^{-3/2}+\varepsilon^{-1/2}\log n)$ for a convex polygon.

\end{enumerate}

\section{Background}
\label{sec:background}
\subsection{Notational Conventions}
\label{subseq:NotationalConventions}
Throughout this paper, the following notational conventions are adopted: consider a compact convex set $C\in \mathbb{R}^d$, where a compact set is defined as a closed and bounded set in $\mathbb{R}^d$. The boundary of $C$ is shown with $\partial C$. The convex hull of a set of points $p^1,...,p^n$ is shown by $\Conv(p^1,...,p^n)$. The relative interior of a set $C$ is shown with $\relint(C)$. The line segment between points $A$ and $B$ is shown with $\overline{AB}$. The distance between two points $p^1$ and $p^2$ is shown by $\dist(p^1,p^2)$, while the distance between point $p$ and a set $S$ is defined as distance between $p$ and its projection on $S$, i.e., $\dist(p,\proj(p,S))$. The diameter of $C \in \mathbb{R}^d$ is denoted by $\diam(C)$ and its volume by $\Vol(C)$. The area of $C \in \mathbb{R}^2$ is shown by $\Area(C)$. For simplicity of notation in proofs $|C|$ is also used to represent the volume (area) of $C$ and similarly, $|\overline{AB}|$ is used to show the length of the line segment $\overline{AB}$.  In 2D, a rectangle $R$ with four corners at points $A,B,C,$ and $D$ is identified by $\square ABCD$ and a triangle with three corners at points $A,B,$ and $C$ is identified by $\triangle ABC$. 
The \emph{aspect ratio} $\AR$ of a rectangle $R \in \mathbb{R}^d$ with sides $s_1,...,s_d$ is defined as the ratio of the length of its longest side to the length of its shortest side, i.e., $\AR(R) = \max_i s_i / \min_i s_i$. In 2D, this is $\AR(R)=\max\left\{height/width, \: width/height\right\}$, 
i.e., $AR \geq 1$ and the equality holds for a square. We also define the aspect ratio of a non-rectangular convex set $C \subset \mathbb{R}^d$ as $\AR_{cvx} (C) = (\diam(C))^2 / \Vol(C)$. Finally, we denote by $\refl(\cdot,\ell)$ the reflected image of a point or a set under reflection at line $\ell$ and by $\rot(\cdot,o)$ the rotated image of a point or a set with respect to the center $o$.

\subsection{Related work}
\label{sec:related-work}
\subsubsection{Geometric Shape Approximation}
\label{subsubsec:GeoShapeApx}
The practice of approximating one geometric shape with another geometric shape is not restricted to approximating convex polygons with inscribed or circumscribed rectangles. DePano \etal \cite{depano1987finding} presented $\mathcal{O}(n^2)$ time algorithms for finding inscribed equilateral triangles and squares of maximum area inside convex polygons and an $\mathcal{O}(n^3)$ time algorithm for finding the largest inscribed equilateral triangle inside a general polygon, where $n$ is the number of vertices of the input polygon. Alt \etal \cite{alt1990approximation} presented several polynomial time algorithms for the problem of approximating convex polygons with rectangles, circles, and polygons with fewer edges, where the approximate shape is not restricted to be inscribed or circumscribing but the area of its symmetric difference with the polygon or the Hausdorff distance between their boundaries must be minimized. 
Zhu \cite{zhu1997approximating} expanded the results of \cite{alt1990approximation} to three dimensional convex polyhedrons. Jin proposed an $\mathcal{O}(n^2)$ time algorithm \cite{jin2017finding-parallelogram} and an $\mathcal{O}(n \log^2 n)$ time algorithm \cite{jin2018maximal-parallelogram} for finding all locally maximal area parallelograms inside a convex polygon. Recently, Keikha \etal \cite{keikha2017maximum} showed that a long-lasting linear-time algorithm for finding the maximal triangle inside a convex polygon proposed in 1979 by Dobkin and Snyder \cite{dobkin1979general} was, in fact, incorrect and then provided an $\mathcal{O}\left(n\log n\right)$ time algorithm for this problem. However, there exist other linear-time algorithms for this problem proposed by Chandran and Mount \cite{chandran1992parallel}, Kallus \cite{kallus2017linear}, and Jin \cite{jin2017maximal-triangle}. 
Shape approximation has not been limited to convex polygons either. For example, Chaudhuri \etal \cite{chaudhuri2003largest} developed an $\mathcal{O}(n^3)$ time algorithm for the problem of finding the largest empty rectangle among a point set. Marzeh \etal \cite{marzeh2019algorithm} developed a heuristic for finding the largest axis-aligned inscribed rectangle in a general polygon. 
Molano \etal \cite{molano2021finding} proposed an $\mathcal{O}(n^3)$ approximation algorithm for finding the largest volume parallelepiped of arbitrary orientation inscribed in a solid in $\mathbb{R}^3$.

\subsubsection{Approximation with Rectangles}
\label{subsubsec:ApxRectangles}
This paper is primarily focused on approximation of convex sets with ``rectangles''. There is a rich literature for this problem when the convex set is further restricted to be a polygon. The history of the problem of approximating polygons with inscribed (circumscribed) rectangles goes back to the famous problem posed in 1951 by P{\'o}lya and Szeg{\"o}  \cite[p.~110]{polya1951isoperimetric}. They showed that there exist homothetic rectangles $R_1$ and $R_2$ for a planar convex region $C$ with the homothety ratio 3 such that $R_1\subset C \subset R_2$, i.e., $R_2$ is a dilation of $R_1$ by a scaling factor 3. They also conjectured that the homothetic ratio is no more than 2. Radziszewski \cite{radziszewski1952probleme} presented a lower bound on the area of the largest inscribed rectangle $R$ inside a convex polygon $C$, as $\Area(R)\geq \frac{1}{2}\Area(C)$. 
Hadwiger \cite{hadwiger1955volumschatzung} showed that for a convex body $C \subset \mathbb{R}^d$, there exist side-parallel 
$d$-dimensional rectangles $R_1$ and $R_2$ such that $R_1\subset C \subset R_2$ and $\frac{1}{d!}\Vol(R_2) \leq \Vol(C) \leq d^d \Vol(R_1)$. 
Separately, Kosinski \cite{kosinski1957proof} proved that for a convex body $C \subset \mathbb{R}^d$, there exists a $d$-dimensional rectangle $R$ such that $C \subset R$ and  $\Vol(R) \leq d! \Vol(C)$.  
Gr{\"u}nbaum \cite[pp.~258-259]{grunbaum1963measures} showed that for a convex set $C \subset \mathbb{R}^2$ there exist parallelograms $L_1$ and $L_2$ such that $L_1 \subset C \subset L_2$ for which $L_1$ and $L_2$ are homothetic with the homothety ratio 2. Lassak \cite{lassak1993approximation} first proved P\'{o}lya and Szeg\"{o}'s conjecture and improved all the above results in 2D by showing that for any convex set $C \subset \mathbb{R}^2$ there are homothetic rectangles $R_1$ and $R_2$ for which $R_1$ is inscribed in $C$ and $R_2$ is circumscribed about $C$ with a positive homothety ratio of at most 2 and $\frac{1}{2}\Area(R_2)\leq \Area(C)\leq 2\Area(R_1)$. Schwarzkopf {\em \etal}~\cite{schwarzkopf1998approximation} obtained the same homothety ratio while presenting a more transparent proof. For a finer inner and outer approximation in 2D, Brinkhuis \cite{brinkhuis2016inner} showed that there exists a quadrangle $Q$ that its sides support $C$ at the vertices of a rectangle $R_1$ and at least three of its vertices lie on the boundary of a rectangle $R_2$ that is a dilation of $R_1$ with ratio 2.

\subsubsection{MAAIR, MAIR, and Higher Dimensions}
\label{subsubsec:MAAIR-MAIR-MVAIR-MVIR}
Amenta \cite{amenta1994bounded} proposed a convex programming model with exponential ($2^{d}n$) number of constraints to find the  MVAIR inside the intersection of a family of $n$ convex sets in $d$-dimensional space and proposed existing randomized algorithms to solve it in \emph{expected} $\mathcal{O}(n)$ time when the $n$ convex sets are either halfspaces (where the intersection forms a polytope) or of constant complexity.  In 2D, Daniels \etal \cite{daniels1993finding,daniels1997finding} proposed an $\mathcal{O}(n\alpha(n)\log^2 n)$ time algorithm for finding the MAAIR inside an $n$-vertex horizontally (vertically) convex polygon, where $\alpha(n)$ is the slowly growing inverse of Ackermann function \cite{ackermann1928hilbertschen}. For orthogonally convex polygons the algorithm performs in $\mathcal{O}(n\alpha(n))$ time. Fischer and H{\"o}ffgen \cite{fischer1994computing} developed an exact $\mathcal{O}(\log^2 n)$ time algorithm to compute the MAAIR in a convex $n$-gon. Alt \etal \cite{alt1995computing} developed an exact $\mathcal{O}(\log n)$ time algorithm for the same problem. This is the best known exact result for the MAAIR problem in a convex polygon.

For MAIR, Hall-Holt \etal \cite{hall2006finding} developed a polynomial-time approximation scheme (PTAS), which for any fixed $\varepsilon > 0$ computes the $(1-\varepsilon)$--approximation to the optimal solution of the maximum area $c$-\textit{fat} rectangle, i.e. a rectangle with aspect ratio bounded from above by $c$, inside a ``simple'' polygon in $\mathcal{O}\left(n\right)$ time.   
Knauer \etal \cite{knauer2012largest} showed that the fatness condition is unnecessary for approximating the optimal MAIR when the input polygon is convex. They developed a randomized $\mathcal{O}(\varepsilon^{-1}\log n)$ time $(1-\varepsilon)$--approximation algorithm that works with probability $t$ for any constant $t<1$ and a deterministic $\mathcal{O}(\varepsilon^{-2}\log n)$ time $(1-\varepsilon)$--approximation algorithm. Their algorithm uses Alt \etal's exact algorithm \cite{alt1995computing} for finding MAAIR as a subroutine. It appears that the analysis of the running time of this algorithm misses the fact that for using Alt \etal's algorithm \cite{alt1995computing} to find the MAAIR aligned to the $\varepsilon$-direction of the MAIR, one needs to rotate the axes or the polygon to that direction, which takes $\mathcal{O}(n)$ time (to make this a fair comparison we consider this part as pre-processing in one of our algorithms). This algorithm was the first $(1-\varepsilon)$--approximation algorithm for finding the MAIR inside a convex polygon. They have also sketched a straightforward exact algorithm that works in $\mathcal{O}(n^4)$ time. Cabello \etal \cite{cabello2016finding} improved these results by presenting an $\mathcal{O}(n^3)$ exact algorithm as well as a $(1-\varepsilon)$--approximation algorithm for this problem. Their approximation algorithm works for any convex set in running time $\mathcal{O}(\varepsilon^{-3/2}+\varepsilon^{-1/2}T_C)$, where $T_C$ is the time needed to perform two different queries on $C$ due to \cite{ahn2006inscribing} that will be discussed in Section \ref{subsubsec:Family-Approx-MAIR}. For a convex polygon, whose vertices are given as a sorted array or as a binary search tree, those queries can be done in $\mathcal{O}\left(\log n\right)$ time. Recently, Choi \etal \cite{choi2021maximum} presented an algorithm that computes the MAIR in a general polygon in $\mathcal{O}(n^3\log n)$ time. A variant of their algorithm finds the MAIR in a convex polygon in $\mathcal{O}(n^3)$ time.

\section{Optimal Properties of the MAIR}
\label{seq:MAIR_optProperties}

\subsection{Optimal Properties of the MAIR in a Convex Polygon}
\label{subseq:MAIR_optProperties_Polygon}
To understand the optimal inscribed rectangles better and before diving into the optimization models and algorithms in higher dimensions, we begin with the geometric properties of the traditional largest inscribed rectangles (MAIR) in a 2-dimensional convex polygon  $C\subset \mathbb{R}^2$.  

DePano \etal \cite{depano1987finding} proved that a maximum area equilateral triangle inscribed in $C$ must have at least one corner coincident with a vertex of $C$. They also proved that a maximum area square inscribed in $C$ either has at least one corner coincident with a vertex of $C$ (i.e., a vertex corner), or all four corners lie on the interior of edges of $C$. Schwarzkopf \etal \cite{schwarzkopf1998approximation} showed that the maximum area rectangle inscribed in $C$ has two diagonal vertices that lie on the boundary of $C$. Knauer \etal \cite{knauer2012largest} mentioned without proof that the largest inscribed rectangle inside $C$ is either a square with two opposite corners coincident with two vertices of $C$ or has at least three non-vertex corners on the boundary of $C$. The latter statement seems to be a misstatement or incorrect. Schlipf \cite{schlipf2014stabbing} proved the former statement of \cite{knauer2012largest} and also proved that a rectangle with three non-vertex corners on the boundary of $C$ cannot be optimal. We strengthen these results and provide a much simpler proof.

\begin{observ}
\label{observ:diagonal}
Consider a family of axis-aligned rectangles in $\mathbb{R}^2$ with a diagonal of fixed size $\ell$ that makes angle $\theta$ with the $x$-axis with $0 \leq \theta \leq \pi/4$. Among all such rectangles generated by this diagonal, the rectangle with $\theta=\pi/4$ has the largest area. Similarly, for rectangles with  $\pi/4 \leq \theta \leq \pi/2$ the area is maximized when $\theta=\pi/4$. In other words, the square has the maximum area among the rectangles of the same diagonal size.  
\end{observ}

\begin{Def} A corner of an inscribed rectangle inside $C$ is one of the following types: \\
(a) A \textit{vertex-corner} which coincides with a vertex of $C$. \\
(b) An \textit{edge-corner} which lies on a non-vertex point of 
an edge (boundary) of $C$. \\
(c) An \textit{interior-corner} which lies strictly inside $C$.
\end{Def}

\begin{thm}
\label{thm:optPropertiesPolygon}
The MAIR inside a convex polygon $C$ must satisfy at least one of the following conditions: 
\begin{description}
\item [Case 1:] It has no interior-corner (i.e., the four corners are on $\partial C$, the boundary of $C$).
\item [Case 2:] It has one interior-corner and at least one vertex-corner adjacent to the interior-corner. 
\item [Case 3:] It has two diagonal interior-corners, two diagonal vertex-corners, and the MAIR is a square. 
\end{description}
\end{thm}
\begin{proof}
Consider an inscribed rectangle $R$ with vertices $a,b,c$ and $d$ in a c.c.w. order and dimensions $w\times h$. We prove by contradiction that if none of the conditions hold for $R$, it cannot be the MAIR and that in each of these three cases the specified conditions must hold. The cases in the theorem are organized based on the number of interior-corners of the rectangle. So we present the proof in the same order. 

If $R$ has  four or three interior-corners, we can easily expand both its length and width. This rules out all cases not specified in the theorem.  
We prove Case 1 by providing a counterexample to the opposite statement that no MAIR could have four corners on the boundary of $C$. The counterexample is simply provided by considering $C$ to be a triangle.

Suppose $R$ has one interior-corner (Case 2), no vertex-corner, and three edge-corners $a,b$ and $c$ lying on edges $e_a, e_b$ and $e_c$. Let $l_a$ be the perpendicular line to $e_a$ at $a$. Similarly define $l_b$ and $l_c$ for $b$ and $c$. Let $p_1$ be the intersection of $l_a$ and $l_b$ and $p_2$ be the intersection of $l_b$ and $l_c$. A small rotation of $C$  in either directions around $p_1$ (or $p_2$), will put $a$ and $b$ (or $b$ and $c$) in the interior of $C$. At least one of these rotations (c.w. or c.c.w) will do the same for $c$ (or $a$), while maintaining the fourth corner in the interior of $C$. Hence, $R$ cannot be the MAIR. This proves that if the MAIR has one interior point it should have at least one vertex point as this rotation argument would not hold for a vertex-corner. It remains to prove that this vertex-corner should be adjacent to the interior-corner. Assume they are diagonal to each other. Without loss of generality (w.l.o.g.) let $a$ the lower left corner of $R$ be the vertex-corner and $c$ be the interior-corner; see Figure \ref{fig:oneInteriorCornerProof}. If either $b$ or $d$ is a vertex corner the case is proved. So let $b$ and $d$ be both edge-corners. Let $e_a^1$ and $e_a^2$ denote the edges of $C$ at the vertex $a$ such that $e_a^1$ lies below $R$ and $e_a^2$ lies to the left of $R$. Also, let $e_b$ and $e_d$ be the edges of $C$ tangent to the corners $b$ and $d$, respectively. Let $K_a$ be a convex cone of directions (vectors) created by the directions of $e_a^1$ and $e_a^2$. By convexity of $C$ and perpendicularity of the edges of $R$, the direction vectors of $e_b$ and $e_d$ must lie inside $K_a$. If the angle between $e_d$ and $e_a^1$ is greater than (or equal to) the angle between $e_b$ and $e_a^1$, then a small enough translation of $R$ in the direction of $e_b$, will keep $b$ on $e_b$, will keep $c$ an interior-corner, will make $d$ an interior-corner (or slide it along the edge $e_d$), and will make $a$ an interior-corner. This means $R$ cannot be the MAIR. It is crucial to note that the translation could be done with respect to any direction that lies between the direction of $e_b$ and $e_d$ in the cone $K_a$. If the angle between $e_d$ and $e_a^1$ is smaller than the angle between $e_b$ and $e_a^1$, there will be no feasible direction for translation. However, in this case $R$ can be rotated around $d$ in a c.c.w direction, putting both $a$ and $b$ strictly inside $C$, while keeping $c$ as an interior-corner. This rotation is possible since $c$ is an interior-corner. This completes the proof of Case 2.

\begin{figure}[t]
\begin{centering}
\subfloat[\label{fig:Translation-feasible}]{ \begin{centering}\includegraphics[width=0.3\columnwidth]{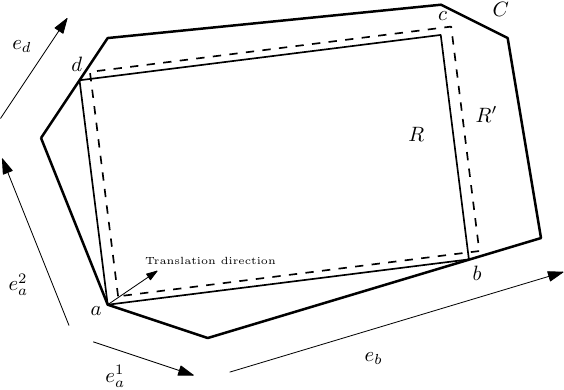}\par\end{centering} }
\qquad{} \qquad{}
\subfloat[\label{fig:Cone_feasible}]{ \begin{centering}\includegraphics[width=0.18\columnwidth]{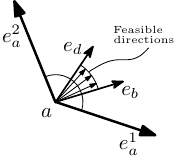}\par\end{centering} }
\par\end{centering}

\begin{centering}
\vspace{-5pt}
\subfloat[\label{fig:Translation-infeasible}]{ \begin{centering}\includegraphics[width=0.3\columnwidth]{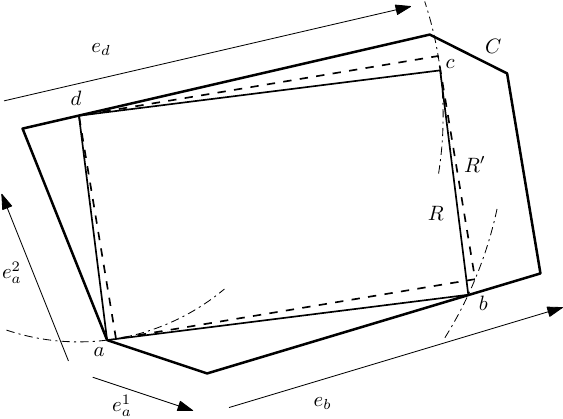}\par\end{centering} }
\qquad{} \qquad{}
\subfloat[\label{fig:Cone_infeasible}]{ \begin{centering}\includegraphics[width=0.18\columnwidth]{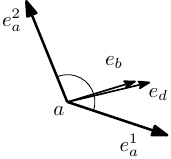}\par\end{centering} }
\par\end{centering}

\caption{\label{fig:oneInteriorCornerProof} The illustration of the proof of the adjacency part in the second case (two vertex-corners) in Theorem \ref{thm:optPropertiesPolygon}. Rectangle $R$ with one interior-corner which is diagonal to the vertex-corner. Figure (\ref{fig:Translation-feasible}) shows the case when feasible directions exists as shown in Figure (\ref{fig:Cone_feasible}) and Figure (\ref{fig:Translation-infeasible}) shows the case where there is no feasible direction  for translation of $R$ as seen in Figure (\ref{fig:Cone_infeasible}) but $R$ can be rotated around corner $d$. The dashed rectangle shows the translated or rotated rectangle $R'$ that has more than one interior-point. \vspace{-5pt}}
\end{figure}

\begin{figure}[t]
  \centering
  \includegraphics[width=0.35\textwidth]{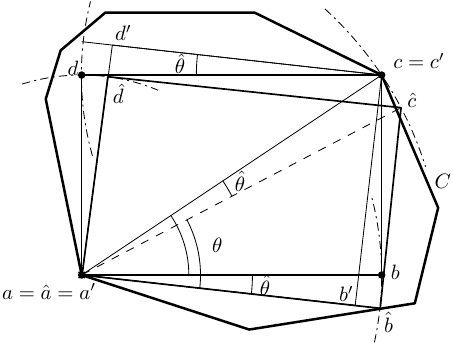}
  \caption{The illustration of the proof of the third case (two vertex-corners) in Theorem \ref{thm:optPropertiesPolygon}. By constructing the angle $\hat{\theta}$ and then the new rectangle $R'=a'b'c'd'$, we can show that a non-square rectangle $R$ with only two vertex-corners cannot be optimal. \vspace{-5pt}}
  \label{fig:2vertex-corners-proof6}
\end{figure}

Assume $R$ has two interior-corners (Case 3) that are adjacent to each other, then we can expand $R$ in the direction perpendicular to the edge connecting these two corners. Now suppose $R$ has two diagonal interior-corners, say $b$ and $d$, and two edge-corners, say $a$ and $c$, touching edges $e_a$ and $e_c$. Rotating $C$ slightly around $a$ (or $c$) either c.w. or c.c.w. will put $c$ (or $a$) in the interior and maintain $b$ and $d$ strictly inside the polygon and hence we can enlarge it. The direction of rotation (c.w or c.c.w.) is toward increasing the obtuse angle between the diagonal $\overline{ac}$ and $e_c$ (or $\overline{ac}$ and $e_a$). If it is a right angle then both directions work. Now suppose $R$ has two diagonal interior-corners and there is just one vertex-corner, a slight rotation of $C$ around the vertex-corner will put the fourth corner (an edge-corner) and keep the two interior-corners strictly inside $C$, making it expandable. This proves that if $R$ has two interior corners they should be diagonal and the other two vertices should be diagonal vertex-corners

Now, suppose $R$ has two diagonal interior-corners, say $b$ and $d$, and two vertex-corners, say $a$ and $c$, but it is not a square, the proof is a bit more complicated. Let $\theta= \angle bac$; see Figure \ref{fig:2vertex-corners-proof6}. Let $\theta_1$ be the maximum angle for c.w. rotation of $R$ around $a$ such that $b$ stays in $C$. Similarly $\theta_2$ be the maximum angle for c.w. rotation of $R$ around $c$ such that $d$ stays in $C$. Since $R$ is not a square, we have either $w>h$ or $w < h$. Without loss of generality assume $w>h$ and thus $\theta=\arctan (h/w) <\pi/4$. Choose $\hat{\theta}>0$ such that $\hat{\theta}<\min\{(\frac{\pi}{4}-\theta),\theta_1,\theta_2\}$. Notice that if $R$ is a square we cannot find such $\hat{\theta}$. Now without loss of generality assume $\theta_1<\theta_2$. Rotate $R$ in a c.w. direction around $a$ as much as angle $\hat{\theta}$. Let $\hat{R}=\square \hat{a}\hat{b}\hat{c}\hat{d}$ denote the rotated rectangle. Then clearly $a=\hat{a}$, $\hat{b}$ lies either strictly inside or on an edge of $C$, $\hat{c}$ is possibly outside $C$, and $\hat{d}$ is inside $C$. Rotate the segment $\overline{cd}$ in a c.w. direction as much as angle $\hat{\theta}$. The whole rotated segment, call it $r_{cd}$, will stay inside $C$. Extend $\overline{\hat{a}\hat{d}}$ to touch $r_{cd}$. Call the intersection point $d'$ and let $a'=a$. It is easy to see that $\overline{a'd'}$ is perpendicular to $r_{cd}$. Draw a line from $c$ parallel to $\overline{a'd'}$ to touch $\overline{\hat{a}\hat{b}}$ at $b'$ and let $c'=c$. It can also be observed that $\overline{b'c'}$ is perpendicular to $\overline{a'b'}$. Then the new rectangle $R'=\square a'b'c'd'$ is inscribed inside $C$, has two diagonal vertex-corners, has the same diagonal as $R$, and the $\angle{b'a'c'}$ is greater than $\theta$. Therefore, we have $\Area(R')>\Area(R)$ by Observation \ref{observ:diagonal}. This proves Case 3.

Finally, none of the three conditions in the theorem is redundant since for each one of them we can easily construct a polygon that gives us a MAIR satisfying that condition.
\end{proof}

The following corollaries are direct results of Theorem \ref{thm:optPropertiesPolygon}.
\begin{cor}
\label{cor:diagonalCorners}
The MAIR has at least two diagonal corners on $\partial C$. Unless both of these corners are vertex-corners, at least one other corner has to lie on $\partial C$. 
\end{cor}

\begin{cor}
\label{cor:adjacentCorners}
Each interior corner, if any exists, has two adjacent corners on $\partial C$.
\end{cor}

\subsection{Optimal Properties of the MAIR in a Centrally Symmetric Convex Set}
\label{subseq:MAIR_optProperties_CentrallySymmetric}
In this section we first define a centrally symmetric convex set and then present a theorem that summarizes a basic property of the MAIR in such sets.
\begin{Def}
A convex set $C$ is centrally symmetric if there exists a point $o\in C$ such that $C$ is centrally symmetric to itself with respect to $o$; the point $o$ is called the center of symmetry for $C$.
\end{Def}
\begin{thm}
\label{thm:optPropertiesCentrallySymmetricConvex}
Let $C$ be a centrally symmetric convex compact set with respect to the center $o$ and let $R_{opt}$ be the MAIR inside $C$. Then, the center of $R_{opt}$, i.e., the intersection of its diagonals, must lie at $o$.
\end{thm}
\begin{proof}
\begin{figure}[h]
  \centering
  \includegraphics[width=0.55\textwidth]{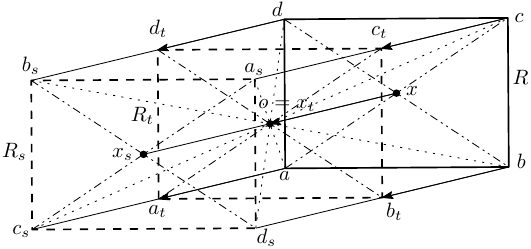}
 \protect\caption{An illustration of the proof of the Theorem \ref{thm:optPropertiesCentrallySymmetricConvex}. Rectangle $R_s$ is the symmetric counterpart of $R$ with respect to the center $o$ and $R_t = R + \protect\overrightarrow{xo}$ is the translation of $R$ in direction $\protect\overrightarrow{xo}$ and is centered at point $x_t=o$. Rectangle $R_t \subset \Conv(R,R_s) \subset C$ has the same size as $R$. \vspace{-5pt}}
  \label{fig:Symmetric-Centrally}
\end{figure}
Consider $R = \square abcd \subset C$ with its center at a point $x \neq o$. We first prove that the direction $\overrightarrow{xo}$ is a feasible translation direction for $R$ and for any such rectangle there exists a rectangle with the same size centered at $o$. Since $C$ is centrally symmetric with respect to $x$, let $R_s=\square a_s b_s c_s d_s \subset C$ be the symmetric counterpart of $R$ with respect to the center $o$, as shown in Figure \ref{fig:Symmetric-Centrally}. We must have $\overline{ac} \parallel \overline{a_s c_s}$ and $\overline{bd} \parallel \overline{b_s d_s}$ with $|\overline{ac}| = |\overline{a_s c_s}| = |\overline{bd}| = |\overline{b_s d_s}|$. Therefore, we have $\overline {x x_s} \parallel \overline{a c_s} \parallel \overline{c a_s} \parallel \overline{b d_s} \parallel \overline{d b_s}$,  and $|\overline {x x_s}| = |\overline{a c_s}|  = |\overline{c a_s}| = |\overline{b d_s}|  = |\overline{d b_s}|$. Let $R_t$ be the rectangle that has its center at $x_t=o=(x+x_s)/2$ and its corners $a_t,b_t,c_t,d_t$ at the midpoint of segments $\overline{a c_s}, \overline{b d_s}, \overline{c a_s}, \overline{d b_s}$, respectively. Note that $R_t \subset \Conv(R,R_s) \subset C$ and $|R_t| = |R|$ as $R_t$ is simply a translation of $R$ in direction $\overrightarrow{xo}$ and its corners are sliding along line segments that are entirely inside $C$ due to the convexity of $C$. 

Now, if $\overline{ox}$ is aligned with an edge of $R$ then $R' = \Conv(R,R_s)$ is a rectangle and $|R'| > |R|$, since $x \neq o$. If $\overline{ox}$ is not aligned with an edge of $R$, then $\Conv(R,R_s)\supset R$ is an irregular hexagon in which $R$ has one interior-corner and three vertex-corners. Assume w.l.o.g. that $a$ is the interior corner.  Translating $R$ slightly in the direction $\overrightarrow{xo}$ will keep $a$ as interior-corner, makes $c$ (its diagonal opposite) an interior corner, and makes $b$ and $d$ edge-corners. Therefore, by Theorem \ref{thm:optPropertiesPolygon}, $R$ cannot be the MAIR. 
\end{proof}

\subsection{Optimal Properties of the MAIR in an Axially Symmetric Convex Set}
\label{subseq:MAIR_optProperties_AxiallySymmetric}
Another major category of symmetric sets are axially symmetric sets. We first define an axially symmetric convex sets and then present a theorem that summarizes some of the properties of the MAIR in such sets.
\begin{Def}
A convex set $C$ is axially symmetric if there exists a line $\ell$ with $\ell\cap C\neq\emptyset$ such that $C$ is symmetric to itself  relative to $\ell$; the line $\ell$ is said to be the axis of symmetry for $C$.
\end{Def}
\begin{thm}
\label{thm:optPropertiesAxialSymmetricConvex}
Let $C$ be an axially symmetric convex compact set, $\ell$ be its line of axial symmetry with $\textrm{SymAxis}(C) = \ell \cap C$, and $R_{opt}$ be the MAIR inside $C$.  
Then, $R_{opt}$ must satisfy  
the following conditions: 
\begin{enumerate}
\item We must have $\Lambda=\ell \cap R_{opt} \neq \emptyset$, and the set $\Lambda$ is not a singleton or an edge of $R_{opt}$.
\item Unless $R_{opt}$ has a corner on one end point of $\textrm{SymAxis}(C)$, at least one corner of $R_{opt}$ must lie on $\partial C$ in each side of $\ell$.
\item If $R_{opt}$ is a square it cannot have three corners strictly on one side of $\ell$.
\item If $R_{opt}$ has two corners (a diagonal) on $\textrm{SymAxis}(C)$, it is either a square or a rectangle that makes an angle $\pi/6 \leq \alpha < \pi/4$ with $\ell$.
\end{enumerate}
\end{thm}
\begin{proof}
We first prove the intersection condition. Consider rectangle $R \subset C$ and assume $R=R_{opt}$ and that $\Lambda = \ell \cap R$ is an empty set, a singleton, or an edge of $R$. Under any of these three conditions, we must have $R$ completely on one side of $\ell$. Assume, w.l.o.g., that $\ell$ is aligned with the $x$-axis and consider rectangle $R = \square abcd  \subset C$, in c.c.w. order with $a$ being the lower left corner, lies above $\ell$ as shown in Figure \ref{fig:AxialSymmetricCase}. Let the points $a',b',c',d'$ be the reflection of corners $a,b,c,d$ with respect to $\ell$. In other words, $R'_{l} = \square a'b'c'd' =\refl(R,\ell) \subset C$ is the reflection of $R$ at $\ell$ and has the same size as $R$. Also let $p,q$ be the left and right end of $\textrm{SymAxis}(C)$ and define the polygon $P= \Conv(p,q,a,b,c,d,a',b',c',d')$. Note that $P \subset C$ and $|\textrm{SymAxis}(C)|=|\overline{pq}| \geq |\overline{bd}|=|\overline{ac}|$. If $R$ has four corners on $\partial C$ then $R$ must be aligned to $\ell$, due to the convexity of $C$, $P$ must be a rectangle, and $R$ can be extended in the direction orthogonal to $\ell$ due to the axial symmetry of $C$, contradicting the assumption. If only two corners of $R$ touch $\partial C$ and they are adjacent corners, then $R$ can be easily enlarged by extension. Also, if these two corners are diagonal $R$ can be enlarged by first rotating it around one of these two corners that has a larger distance to $\ell$ and then extending it. Hence, $R$ must have three of its corners on $\partial C$  and must have an angle with $\ell$, making either $a$ or $b$ an interior-corner. Let $a$ be the interior-corner (w.l.o.g.). Due to the axial symmetry and convexity of $C$, the vectors $\overrightarrow{dp}$ and $\overrightarrow{bq}$ must be either parallel or diverging. Any direction within the convex cone defined by these two vectors would be a translation direction that could put at least one more corner of $R$ in the interior of $C$. Therefore, $R\neq R_{opt}$. Note that in this case the other two corners would become either interior-corners or edge-corners on $\partial P$, which also shows (by Theorem \ref{thm:optPropertiesPolygon}) that $R$ cannot be the MAIR.  
If the set $\Lambda$ is a singleton, it must be an interior-corner of $R$, e.g., $a$, and the same translation argument applies. If $\Lambda$ is equal to an edge of $R$ then  $R$ is aligned to $\ell$ and it can be extended in the direction orthogonal to $\ell$. This proves that line $\ell$ crosses $\partial R_{opt}$ in two points, i.e., $\ell \cap \interior R_{opt} \neq \emptyset$ (1st condition). 

\begin{figure}[t]
  \centering
  \includegraphics[width=0.35\textwidth]{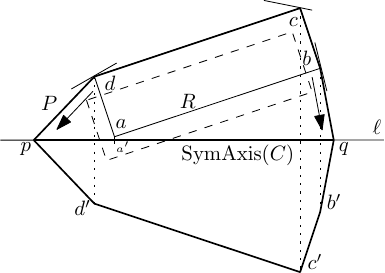}
  \caption{An illustration of the proof of the first condition in Theorem \ref{thm:optPropertiesAxialSymmetricConvex}. The line $\ell$, is the axis of symmetry of the convex set $C$ with $\textrm{SymAxis}(C) = \ell \cap C$. Rectangle $R$ lies on one side of $\ell$ and touches the boundary of $C$ at three of its corners for which a tangent line of $C$ at that point is drawn. The polygon $P$ is the convex hull of $R$, $\textrm{SymAxis}(C)$, and the reflected image of $R$ with respect to $\ell$. The dashed rectangle shows $R$ when translated in a feasible direction.}
  \label{fig:AxialSymmetricCase}
\end{figure}

For the second condition, by the intersection condition, we know that at least one corner of $R_{opt}$ lies on each side of $\ell$. If $R_{opt}$ has no corner on the endpoints of $\textrm{SymAxis}(C)$, we must have at least one corner of $R_{opt}$ on $\partial C$ in each side of $\ell$, otherwise the translation and extension are possible and $R_{opt}$ could be enlarged.

Now we prove that if $R_{opt}$ is a square, it cannot have three corners ``strictly'' on one side of $\ell$. Assume $\ell$ crosses the boundary of  square $R \subset C$ in two points and that $R$ has three corners strictly on one side of $\ell$ and the last corner is strictly on the other side of $\ell$. Then $R$ must have an angle $\alpha$ with $\ell$. Without loss of generality assume $a$ is the corner of $R$ below $\ell$ and let $P=\Conv(p,q,R,\refl(R,\ell))$. If $R$ has three corners on one side of $\ell$ and its diagonal $\overline{bd}$ is parallel to the line $\ell$, then $a$ cannot be on the $\partial C$, contradicting the second condition; see Figure (\ref{fig:SymmetricCase7-SquareDiagonalAligned}). In this case, a small enough translation of $R$ in the orthogonal direction to $\ell$ would make $c$ an interior-corner while keeping $a$ in the interior of the convex polygon $P\subset C$. The corners $b$ and $d$ would be either interior-corners or edge-corners with respect to $P$. Hence, by Theorem \ref{thm:optPropertiesPolygon}, $R$ cannot be the optimal rectangle inside $P$ and therefore cannot be the MAIR in $C$.

Now assume $R$ is a square that has three corners on one side of $\ell$ and makes angle $0< \alpha < \pi/4$ with $\ell$, i.e., its diagonal $\overline{bd}$ is not parallel to $\ell$; see Figure (\ref{fig:SymmetricCase7-SquareDiagonalNotAligned}). Let $a'$ and $c'$ be the points that vertical lines (perpendicular to $\ell$) going through $a$ and $c$ intersect the edges $\overline{cd}$ and $\overline{ab}$, respectively. Clearly, the quadrangle $aa'cc'$ is a parallelogram. Let $\ell_1,\ell_2$ and $\ell_3$ be lines parallel to $\ell = \ell_0$ such that $\ell_1$ bisects $\overline{aa'}$ at the point $e$, $\ell_2$ goes through the center point $f$ of the parallelogram, and $\ell_3$ bisects $\overline{cc'}$ at the point $g$. The line $\ell_i,\; i=0,1,2,3$ partitions the segment $\overline{aa'}$ into two pieces with lengths $d_1^i$ in the bottom and $d_2^i$ at the top and the segment $\overline{cc'}$ into two pieces with lengths $d_3^i$ on the top and $d_4^i$ in the bottom. Obviously, $d_1^i+d_2^i=d_3^i+d_4^i,\; \forall i$. For $\ell_1$, we have $d_1^1=d_2^1$ and $d_3^1 > d_4^1$. For $\ell_3$, we have $d_1^3 > d_2^3$ and $d_3^3 = d_4^3$. In the case of $\ell_2$, we have $d_1^2=d_3^2 > d_2^2=d_4^2$. Note that since the diagonal $\overline{bd}$ is not parallel to $\ell$, the lines $\ell_1,\ell_2$ and $\ell_3$ are distinguished, otherwise the three lines would coincide with $\overline{bd}$. Also, by Thales's basic proportionality theorem,  
the bisector of edges $\overline{ad}$ and $\overline{bc}$ goes through the points $e,f$ and $g$. Since $R$ is a square we can show that all three lines will cross $\overline{ad}$ and $\overline{bc}$ and thus have two corners of $R$ on each side. For $\ell_1$ not crossing $\overline{bc}$, we must have $d_1^1 < |\overline{ab}| \sin \alpha$, while we have the opposite as $d_1^1=d_2^1=|\overline{aa'}|/2 = |\overline{ad}|/(2\cos \alpha) = |\overline{ab}|/(2\cos \alpha) > |\overline{ab}| \sin \alpha$, as we have $2\sin\alpha \cos \alpha = \sin 2\alpha <1$ for $0< \alpha < \pi/4$. The proof for $\ell_3$ is symmetric. 
By the second condition we must have $a \in \partial C$ (otherwise, $R$ could be shifted in the direction orthogonal to $\ell$, making $c$ an interior-corner). Therefore, we must have $d_1^0 \geq d_2^0$ and $d_3^0\geq d_4^0$ due to the convexity and axial symmetry of $C$. Hence, $\ell$ must lie in the slab defined by $\ell_1$ and $\ell_3$ contradicting the assumption of having three corners on one side. This concludes the proof of the third condition.

\begin{figure}[t]
\begin{centering}
\subfloat[\label{fig:SymmetricCase7-SquareDiagonalAligned}]{ \begin{centering}\includegraphics[width=0.35\columnwidth]{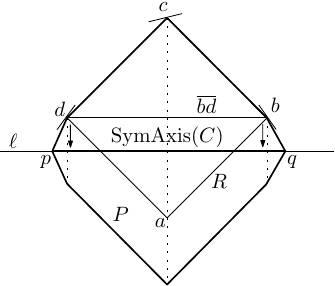}\par\end{centering} }
\quad{}
\subfloat[\label{fig:SymmetricCase7-SquareDiagonalNotAligned}]{ \begin{centering}\includegraphics[width=0.35\columnwidth]{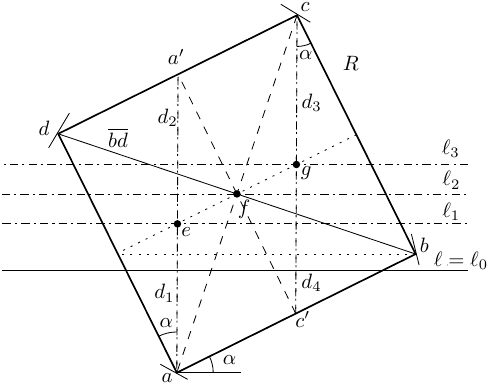}\par\end{centering} }
\par\end{centering}

\caption{\label{fig:SymmetricCase_Square} An illustration of the proof of the third condition in Theorem \ref{thm:optPropertiesAxialSymmetricConvex}. Square $R$ in Figure (\ref{fig:SymmetricCase7-SquareDiagonalAligned}) with three corner on one side of axis of symmetry $\ell$ and a diagonal parallel to $\ell$ cannot be the MAIR as it can be translated in the direction orthogonal to $\ell$ leaving $R$ with two interior-corners $a,c$ and two edge- or interior-corners $b,d$. Figure (\ref{fig:SymmetricCase7-SquareDiagonalNotAligned}) shows the case where the diagonal of $R$ is not aligned with $\ell$. Since $R$ is a square, $\ell$ must lie within the slab defined by $\ell_1$ and $\ell_3$, the bisectors of line segments $\overline{aa'}$ and $\overline{cc'}$. }
\end{figure}

We prove  
 the fourth condition by first showing the existence of such a square by construction and then proving that the same does not hold for a rectangle with $0<\alpha<\pi/6$. Consider a square $R=\square abcd$ with its diagonal $\overline{bd}$ aligned to the $x$-axis. Let $C$ be the rhombus obtained by stretching $R$ slightly via extending the diagonal $\overline{bd}$ for some $\delta > 0$.  For sufficiently small $\delta$, the square $R$ would be the MAIR inside $C$ with its diagonal on the $\textrm{SymAxis}(C)$ and one vertex-corner on the boundary of the polygon $C$ on each side of the symmetry axis.

\begin{figure}[t]
  \centering
  \includegraphics[width=0.35\textwidth]{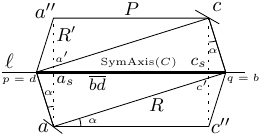}
  \caption{An illustration of the proof of fourth condition in Theorem \ref{thm:optPropertiesAxialSymmetricConvex}. Rectangle $R$ has the $\textrm{SymAxis}(C)$ as  its diagonal and  has one vertex on the boundary of $C$ on each side of the symmetry axis $\ell$. The rectangle $R'=\square aa''cc''$ has a larger area than $R$. \vspace{-5pt}}
  \label{fig:SymmetricCase7-RectangleDiagonalAligned}
\end{figure}
We now prove that a rectangle $R=\square abcd$ that has its diagonal on the $\textrm{SymAxis}(C)$ and makes angle $0<\alpha<\pi/6$ with $\ell$ cannot be the MAIR. Assume the most restrictive case where the diagonal $\overline{bd}$ is equal to the $\textrm{SymAxis}(C)$ and $a,c \in \partial C$; see Figure \ref{fig:SymmetricCase7-RectangleDiagonalAligned}.  Let $a_s$, $a'$, and $a''$ be the points that a vertical line going through the corner $a$ crosses the $\textrm{SymAxis}(C)$, the edge $cd$, and $\partial C$ on the opposite side of $\ell$, respectively. Similarly define the points $c_s$, $c'$, and $c''$. Due to the symmetry of $C$ with respect to  $\ell$ and the fact that $a,c \in \partial C$ we must have $\dist(a'',a_s) = \dist(a,a_s) \geq \dist(a_s,a')$ and $\dist(c'',c_s) = \dist(c,c_s) \geq \dist(c_s,c')$.  We prove that the rectangle $R'=\square aa''cc''$ satisfies $|R'| > |R|$, when $0<\alpha<\pi/6$. As $R'$ is aligned with $\ell$, we have 
\[
|R'| = |\overline{ac''}|\times |\overline{aa''}| = (|\overline{ab}|-|\overline{ad}|\tan \alpha)\cos \alpha \times 2|\overline{ad}| \cos \alpha = 2|\overline{ab}|\times|\overline{ad}| \cos^2 \alpha - 2 |\overline{ad}|^2\sin \alpha\cos \alpha 
\]
Therefore,
\[
\frac{|R'|}{|R|} = \frac{2|\overline{ab}|\times|\overline{ad}| \cos^2 \alpha - 2 |\overline{ab}|\times |\overline{ad}| \tan \alpha \sin \alpha\cos \alpha}{|\overline{ab}|\times|\overline{ad}|} = 2 (\cos^2 \alpha - \sin^2 \alpha) = 2 \cos 2 \alpha \geq 1 \, , \; \mbox{for } 0<\alpha<\pi/6
\] 
This concludes the proof of the fourth condition.
\end{proof}

\begin{cor}
\label{cor:TwoAxis_Centrally}
If a compact convex set $C \in\mathbb{R}^2$ has two perpendicular axes of symmetry $\ell_1$ and $\ell_2$ that cross each other at point $o$, then it is centrally symmetric with respect to the center $o$. 
\end{cor}
\begin{proof}
This is because a rotation with angle $\pi$ of any point $p \in C$ is achieved by two consecutive reflections with respect to lines $\ell_1$ and $\ell_2$, i.e., $\rot(p,\pi) = \refl ((\refl(p, \ell_1), \ell_2)$ and we will have $\rot(p,\pi) \in C$. This can be also seen by the fact that $(p+p')/2 = (p + (-p+2o))/2 = o$.
\end{proof}

\begin{cor}
\label{cor:TwoAxis_Centrally_Diam}
In a compact convex set $C \in\mathbb{R}^2$ that is centrally symmetric with respect to a center $o$, the $\diam(C)$ goes through $o$. Furthermore, if $C$ has two perpendicular axes of symmetry, the $\diam(C)$ may not be necessarily aligned to these axes. 
\end{cor}
\begin{proof}
Trivial.
\end{proof}

\begin{cor}
\label{thm:optPropertiesTwoOrthogonalAxialSymmetricConvex}
Let $C$ be an axially symmetric compact convex set with two perpendicular axes of symmetry $\ell_1$ and $\ell_2$ and let $R_{opt}$ be the MAIR inside $C$. Also, assume the quadrants created by the intersection of $\ell_1$ and $\ell_2$ are numbered from 1 to 4 in a c.c.w order starting from top right quadrant and let $\textrm{SymAxis}_1(C) = \ell_1 \cap C$ and $\textrm{SymAxis}_2(C) = \ell_2 \cap C$. Then, $R_{opt}$ must satisfy the following conditions:
\begin{enumerate}
\item The center of $R_{opt}$, i.e., the intersection of its diagonals, must lie at the intersection of $\ell_1$ and $\ell_2$.
\item Unless $R_{opt}$ has a corner on one end point of $\textrm{SymAxis}_1(C)$ or $\textrm{SymAxis}_2(C)$, at least one corner of $R_{opt}$ must lie on $\partial C$ in each side of $\ell_1$ and $\ell_2$.
\item  If $R_{opt}$ is a square, it either has its diagonals on $\textrm{SymAxis}_1(C)$ and $\textrm{SymAxis}_2(C)$ or has one corner strictly in each quadrant. It must also have at least two diagonally opposite corners on $\partial C$. 
\item If $R_{opt}$ has two corners (a diagonal) on $\textrm{SymAxis}_i(C),\; i=1,2$, it is either a square or a rectangle that makes an angle $\pi/6 \leq \alpha < \pi/4$ with $\ell_i$.
\end{enumerate}
\end{cor}
\begin{proof}
The first condition is proved by Corollary \ref{cor:TwoAxis_Centrally} and Theorem \ref{thm:optPropertiesCentrallySymmetricConvex}. The second condition is a direct application of the 2nd condition in Theorem \ref{thm:optPropertiesAxialSymmetricConvex} to both $\ell_1$ and $\ell_2$. The third condition follows from applying 2nd, 3rd, and 4th conditions in Theorem \ref{thm:optPropertiesAxialSymmetricConvex} to both $\ell_1$ and $\ell_2$. The fourth condition also follows from the 4th conditions in Theorem \ref{thm:optPropertiesAxialSymmetricConvex}.
\end{proof}

\section{Optimization Models for MVIR and MVAIR}
\label{sec:OptimizationModels}
Consider a compact convex set $C\in \mathbb{R}^d$ mathematically expressible in a finite number of convex inequalities. For example, one could consider $C$ to be a polytope 
defined by $C=\{x \in \mathbb{R}^d \;|\;Px\leq b \, , \; P\in \mathbb{R}^{n\times d}\, , \; b\in \mathbb{R}^n \}$ 
or intersection of $d$-ellipsoids defined as $C=\{x \in \mathbb{R}^d \;|\; x^T A_i x +2b_i^T x + c_i \leq 0\, , \; A_i \succ 0\,, \; i=1,...,n\}$, where $A\succ0$ denotes a positive definite matrix. The goal is to formulate the problems of finding the MVIR and the MVAIR inside $C$ as optimization problems. 

\subsection{MVIR as an Optimization Problem}
\label{subsec:MVIR-optModel}
To make the derivation of the optimization model for the MVIR problem more clear, first, let's consider the $d$-dimensional maximum volume inscribed parallelotope in $C$. Let $x^1,x^2,..., x^{d+1}$ be a set of $d+1$ affinely independent vertices of the parallelotope and put $U=[u^1\;u^2\;...\; u^d]$, where $u^i=x^{i+1}-x^1,\; i=1,...,d$. Note that columns of $U$ are linearly independent and form a basis for $\mathbb{R}^d$. See Figure \ref{fig:LIR-image} for a special case in 2D, where $C$ is a compact convex polygon and the desired parallelogram is a rectangle. 

The following definitions are required for constructing the optimization model.

\begin{Def}
A vector $x=(x_1,...,x_d)$ in $\mathbb{R}^d$ is lexicographically positive if its first non-zero coordinate is positive. 
\end{Def}

\begin{Def}
Two points $x,y\in \mathbb{R}^d$ are in lexicographic order if $y-x$ is lexicographically positive. 
\end{Def}

Let's label the $2^d$ vertices of the parallelotope in a lexicographic order with binary vectors so that $q^1=(0,...,0)^T\equiv x^1,\; q^2=(0,0,...,1)^T, ..., q^{(2^d)}=(1,...,1)^T\equiv x^1+\sum_{i=1}^d u^i$. Using this labeling, the $k$th vertex of the parallelotope can be shown by $x^1+ Uq^k$. Thus, the problem of finding the maximum volume inscribed parallelotope (MVIP) in $C$ can be formulated as the following optimization model:
\begin{eqnarray}
\label{mod:max-Parallelotope}
   \maximize_{x^1,\cdots,x^{d+1}} &
    \left| \det \left( \begin{array}{ccccc} x^1 & x^2 & \cdots & x^d & x^{d+1} \\ 1 & 1 & \cdots & 1 & 1 \end{array} \right) \right|  & \quad \st  \\
  & & \nonumber \\
    &  u^i \; = \; x^{i+1}-x^1  \; \;, & \quad i=1,...,d \nonumber \\
    &  x^1+\sum_{i=1}^d q_i^k u^i \; \in \;  C \; \;, &  \quad k=1,...,2^d \nonumber
\end{eqnarray}
where $|\cdot|$ is the absolute value and $\det(\cdot)$ denotes the determinant of a matrix. The objective function calculates the volume of the parallelotope. The first constrain is an auxiliary constraint to define $u^i$ vectors, while the second constraint ensures that all vertices of the parallelotope are inside the convex set $C$. This is, in general, a non-convex optimization problem with an exponential number of constraints and for general $d$ difficult to solve.

Problem (\ref{mod:max-Parallelotope}) can be solved by solving two optimization problems for the positive and negative values of the objective function after removing the absolute value and then choosing the one with the best outcome. The absolute value can also be removed using the epigraph form and rewriting the problem as
\begin{eqnarray}
   & \maximize_{z,x^1,\cdots,x^{d+1}} \qquad \qquad z   \qquad\qquad\qquad\qquad\qquad& \quad \st  \nonumber \\
    & & \nonumber \\
    & z + (-1)^j \det \left( \begin{array}{ccccc} x^1 & x^2 & \cdots & x^d & x^{d+1} \\ 1 & 1 & \cdots & 1 & 1 \end{array} \right) \; \leq \; 0 \; \; , & \quad j=1,2\nonumber \\
    &  u^i \; = \; x^{i+1}-x^1  \; \;, & \quad i=1,...,d \nonumber \\
    &  x^1+\sum_{i=1}^d q_i^k u^i \; \in \;  C \; \;, &  \quad k=1,...,2^d \nonumber
\end{eqnarray}
To keep the model closer to the concept and easier to follow, the model (\ref{mod:max-Parallelotope}) is used for further analysis. 

Now we can add various shape constraints to find different types of largest inscribed bodies. For example, if we require $u^i=u^j,\; \forall i\neq j$, we will obtain the largest inscribed rhombus and if we further add perpendicularity constraints we get the largest inscribed square. Also, we can modify the model by just defining $d+1$ vertices to obtain the largest triangles that we skip here for brevity. Since we are interested in a parallelotope that is a box (or hypercube), additional orthogonality constraints should be imposed. Hence, the problem of finding MVIR in a convex set $C\in \mathbb{R}^d$  can be formulated as
\begin{eqnarray}
\label{mod:MVIR-Set}
   \maximize_{x^1,\cdots,x^{d+1}} &
    \left| \det \left( \begin{array}{ccccc} x^1 & x^2 & \cdots & x^d & x^{d+1} \\ 1 & 1 & \cdots & 1 & 1 \end{array} \right) \right|  & \quad \st  \\
    & & \nonumber \\
     &  u^i  \; = \; x^{i+1}-x^1  \; \;, & \quad  i=1,...,d \nonumber \\
     &  (u^i)^T u^j \; = \; 0 \;\; , & \quad   1\le i < j \le d \nonumber \\
     & x^1+\sum_{i=1}^d q_i^k u^i \; \in \;  C \; \;, &  \quad k=1,...,2^d \nonumber
\end{eqnarray}

In the special case, where the convex set $C$ is a polytope defined by $C=\{x \in \mathbb{R}^d \;|\;Px\leq b\}$ with $P\in \mathbb{R}^{n\times d}$ and $b\in \mathbb{R}^n$, we have
\begin{eqnarray}
\label{mod:MVIR-Polytope}
    \maximize_{x^1,\cdots,x^{d+1}} &
    \left| \det \left( \begin{array}{ccccc} x^1 & x^2 & \cdots & x^d & x^{d+1} \\ 1 & 1 & \cdots & 1 & 1 \end{array} \right) \right|  & \quad \st  \\
        & &  \nonumber \\
     &  u^i  \; = \; x^{i+1}-x^1  \; \;, & \quad  i=1,...,d \nonumber \\
     &  (u^i)^T u^j \; = \; 0 \;\; , & \quad   1\le i < j \le d \nonumber \\
     & v^k \; = \; x^1+\sum_{i=1}^d q_i^k u^i \; \;, &  \quad k=1,...,2^d \nonumber \\
     & Pv^k \; \leq \; b \; \; , & \quad k=1,...,2^d \nonumber
\end{eqnarray}
where $v^k$ vectors are the vertices of the MVIR. 

Both Problems (\ref{mod:MVIR-Set}) and (\ref{mod:MVIR-Polytope}) are non-convex optimization problems with an exponential number of constraints and difficult to solve.

\subsection{MVAIR as an Optimization Problem}
\label{subsec:MVAIR-optModel}
In order to find the MVAIR in a compact convex set $C\in \mathbb{R}^d$, additional constraints are needed to impose axis-alignment. In fact, the MVAIR is a box $R=\left\{x\in \mathbb{R}^d \; | \; x^l \leq x \leq x^u \right\}$ of maximum volume inscribed in $C$, where $x^u$ and $x^l $ are some upper and lower bounds in $ \mathbb{R}^d$, respectively. Hence, it is enough to ensure that for each vertex $v^k$ of $R$, we have $x^l \leq v^k \leq x^u\,,\; k=1,...,2^d$.  

Therefore, we obtain 
\begin{eqnarray}
\label{mod:MVAIR-Set}
   \maximize_{x^1,\cdots,x^{d+1},x^l,x^u} &
    \left| \det \left( \begin{array}{ccccc} x^1 & x^2 & \cdots & x^d & x^{d+1} \\ 1 & 1 & \cdots & 1 & 1 \end{array} \right) \right|  & \quad \st  \\
    & & \nonumber \\
     &  u^i  \; = \; x^{i+1}-x^1  \; \;, & \quad  i=1,...,d \nonumber \\
     &  (u^i)^T u^j \; = \; 0 \;\; , & \quad   1\le i < j \le d \nonumber \\
     & v^k \; = \; x^1+\sum_{i=1}^d q_i^k u^i \; \;, &  \quad k=1,...,2^d \nonumber \\
     & v^k \; \in \;  C \; \;, &  \quad k=1,...,2^d \nonumber \\
     & x^l  \; \leq \; v^k \; \leq \;  x^u \;\,, & \quad k=1,...,2^d \nonumber
\end{eqnarray}
At optimality, $x^l$ and $x^u$ will coincide with two of the opposing extreme points (vertices) of $R$.
For the special case, where $C$ is a convex polytope, as defined in Problem (\ref{mod:MVIR-Polytope}), we have
\begin{eqnarray}
\label{mod:MVAIR-Polytope}
   \maximize_{x^1,\cdots,x^{d+1},x^l,x^u} &
    \left| \det \left( \begin{array}{ccccc} x^1 & x^2 & \cdots & x^d & x^{d+1} \\ 1 & 1 & \cdots & 1 & 1 \end{array} \right) \right|  & \quad \st  \\
    & & \nonumber \\
     &  u^i  \; = \; x^{i+1}-x^1  \; \;, & \quad  i=1,...,d \nonumber \\
     &  (u^i)^T u^j \; = \; 0 \;\; , & \quad   1\le i < j \le d \nonumber \\
     & v^k \; = \; x^1+\sum_{i=1}^d q_i^k u^i \; \;, &  \quad k=1,...,2^d \nonumber \\
     & Pv^k \; \leq \; b \; \; , & \quad k=1,...,2^d \nonumber \\
     & x^l  \; \leq \; v^k \; \leq \;  x^u \;\,, & \quad k=1,...,2^d \nonumber
\end{eqnarray}

Both Problems (\ref{mod:MVAIR-Set}) and (\ref{mod:MVAIR-Polytope}) are also non-convex optimization problems with exponential number of constraints and difficult to solve. However, there is a hidden convexity here and if we take advantage of the structure of the MVAIR problem, i.e., axis-alignment, we are able to model it in a more efficient way.  
Note that to have $R\subseteq C$, model   (\ref{mod:MVAIR-Polytope}) enforces all of its $2^d$ vertices to be inside $C$, hence giving an exponential number of constraints. For this case, 
we can use a more efficient formulation inspired by a problem in \cite{boyd2004convex}. 
This new formulation skips the quadratic number, $\mathcal{O}(d^2)$, of the nonlinear perpendicularity constraints $(u^i)^T u^j  =  0$ and avoids the exponential number, $\mathcal{O}(n2^d)$, of linear constraints $Pv^k \leq  b$. Instead, it deals with $\mathcal{O}(n)$ linear inequality constraints.

Since we have $x^l_j\leq x_j\leq x^u_j$, an upper bound for the left-hand side of each inequality $\sum_j p_{ij}x_j \leq b_i$ is obtained by $\sum_j (p_{ij}^+x^u_j-p_{ij}^-x^l_j)$, where $p_{ij}^+=\max\{p_{ij},0\}$  and $p_{ij}^-=\max\{-p_{ij},0\}$. This means we have $R\subseteq C$ if and only if $ \sum_{j=1}^d (p_{ij}^+x^u_j-p_{ij}^-x^l_j) \; \leq \; b_i \:,\; i=1,...,n$. Hence, the alternative formulation for this problem is as follows: \vspace{-10pt}
\begin{eqnarray*}
     \maximize_{x^l ,x^u} \quad  \prod_{j=1}^d (x^u_j-x^l_j)   & & \st   \\   
     \sum_{j=1}^d (p_{ij}^+x^u_j-p_{ij}^-x^l_j) & \leq & b_i \; , \quad i=1,...,n \\
      x^l_j & \leq & x^u_j\; , \quad j=1,...,d
\end{eqnarray*}
where the objective function calculates the volume of $R$ as the product of the length of $d$ linearly independent vectors corresponding to $d+1$ affinely independent vertices of $R$.  This reproduces the solution space of the Problem \ref{mod:MVAIR-Polytope} in a much simpler format and significantly reduces the number of constraints. This is made possible by taking advantage of the geometry of the axis-aligned rectangle. Note that Polytope $C$ is the intersection of $n$ halfspaces. Due to the axis-alignment of $R$, for each halfspace $C_i=\{x \in \mathbb{R}^d\:|\;p_i^T x \leq b_i\}$ it is sufficient to have $p_i^T v_i^c \leq b_i$, where $v_i^c$ is the closest vertex of $R$ to the hyperplane $H_i=\{x \in \mathbb{R}^d\:|\;p_i^T x = b_i\}$, i.e., $v_i^c = \argmin_{x\in X(R)} \dist(x,\proj(x,H_i))$, where $X(R)$ is the set of extreme points of $R$. Note that $v_i^c$ in not necessarily unique but the value $p_i^T v_i^c$, which is the same for all such vertices, can be bounded from above by $\sum_j (p_{ij}^+x^u_j-p_{ij}^-x^l_j)$. Therefore, instead of enumerating over $2^d$ vertices to be satisfied by $n$ halfspaces, we are effectively checking only one vertex for each halfspace.  Also, perpendicularity constraints are not needed in this formulation as they are implied by the objective function and the last constraint.

Therefore, the problem of finding the MVAIR in a polytope can be efficiently formulated as the following convex optimization problem \vspace{-5pt}
\begin{eqnarray}
\label{mod:MVAIR-Polytope-Efficient}
    \maximize_{x^l ,x^u} \quad f_0(x^u ,x^l) = \sum_{j=1}^d \log(x^u_j-x^l _j)  & & \st  \\
    \sum_{j=1}^d (p_{ij}^+x^u_j-p_{ij}^-x^l_j) & \leq & b_i \; , \quad i=1,...,n   \nonumber
\end{eqnarray}
with the implied constraint $x^u>x^l$, which  
can be solved effciently
and will be discussed in \cref{subsec:SolvingMVAIR}. Note that the number of constraints in this model is exactly the number of linear inequalities defining $C$.

Similarly, the Problem (\ref{mod:MVAIR-Set}) can be rewritten in a far more efficient way by incorporating the upper and lower bound points $x^u$ and $x^l$ in the convex inequalities defining $C$. However, the details of the analysis in this case depends on the definition of $C$. As an example, consider the convex set $C=\{x \in \mathbb{R}^d \;|\; x_1+\cdots+x_d\leq 1\,, \; x_1^2+\cdots + x_{d-1}^2 - x_d \leq 0\}$, i.e., the intersection of a halfspace and a $d$-dimensional paraboloid. To have $R\subseteq C$, we must have $v^k\in C\,, \; k=1,...,2^d$, which means $v^k_1+\cdots+v^k_d \leq 1$ and $\sum_{i=1}^{d-1} (v^k_i)^2 -v^k_d \leq 0$ for $k=1,...,2^d$. They all can be replaced by only ``two'' constraints $\sum_{i=1}^d x^u_i \leq 1$ and $\sum_{i=1}^{d-1} (x^{\max}_i)^2 - x^l_d \leq 0$, where $x^{\max}_i = \max\{|x^l_i|,|x^u_i|\},\; i=1,...,d$. These two constraints solely depend on the two points $x^l$ and $x^u$. Using this setting, there is also no need for the remaining constraints in model (\ref{mod:MVAIR-Set}). Note that the number of required constraints in this model is also exactly the number of inequalities defining $C$ (i.e., $\mathcal{O}(n)$) and depending on the structure of $C$, e.g., for asymmetric convex sets, we may also need an additional $\mathcal{O}(d)$ constraints for defining $x^u$ and $x^l$. The objective function would be the same as that of Problem (\ref{mod:MVAIR-Polytope-Efficient}). Therefore, for this convex set, we obtain
\begin{eqnarray}
\label{mod:MVAIR-Set-Efficeint}
    \maximize_{x^l ,x^u} \quad \sum_{j=1}^d \log(x^u_j-x^l _j)  & & \st  \nonumber \\
    \sum_{i=1}^d x^u_i & \leq & 1 \; ,   \nonumber \\
    \sum_{i=1}^{d-1} (x^{\max}_i)^2 - x^l_d & \leq & 0 \; , \nonumber \\
    - x^{\max}_i \; \leq \; x^u_i  & \leq &  x^{\max}_i \; , \quad i=1,...,d \nonumber \\
    - x^{\max}_i \; \leq \; x^l_i  & \leq &  x^{\max}_i \; , \quad i=1,...,d \nonumber 
\end{eqnarray}
which is a convex optimization problem and can be efficiently solved.

\section{Approximation Algorithms}
\label{sec:ApproxAlgos}
In this section we present exact and approximation algorithms for finding  
lasrgest inscribed rectangle in a compact convex set. 

\subsection{Solving the MVAIR}
\label{subsec:SolvingMVAIR}
The analysis of finding the MVAIR in a general compact convex set defined by a ``finite'' number of convex inequalities depends on having those specific inequalities as there are various possibilities. For this reason, here we introduce an algorithm for finding the MVAIR in a compact convex polytope and analyze its computational complexity. However, both the algorithm and the analysis apply to the general compact convex sets as well.

Having the MVAIR problem modeled as a convex optimization problem enables us to efficiently solve it via efficient convex programming algorithms such as interior-point methods. One of the most efficient interior-point methods for solving convex optimization problems such as MVAIR is the logarithmic barrier method. In addition to the efficiency, the choice of logarithmic barrier method here is also motivated by the fact that the objective function in this method (Eq. (\ref{mod:MVAIR-Polytope-Efficient-LogBarrier}) below) is a closed strictly convex \emph{self-concordant function} --- a class of functions for which the barrier method (with Newton minimization used as a subroutine) provides a rigorous worst-case bound on the number of iterations needed for finding their minimizer, which is useful for the analysis here. Moreover, the convergence analysis is independent of some of the common unknown parameters such as the upper bound on the condition number of the Hessian matrix and its Lipschitz constant, and is also affine invariant, thus insensitive to the choice of coordinates. The latter property is specifically useful for the MVAIR problem as it enables us to rotate or shift the input region in the coordinate system, without changing the worst-case analysis. This worst-case analysis for logarithmic barrier method, which is based on self-concordance properties, was first introduced by Nesterov in \cite{nesterov1988Polynomial_a,nesterov1988Polynomial_b} and was further developed by Nesterov and Nemirovski in a series of papers including \cite{nesterov1988convex,nesterov1989self,nesterov1991acceleration} and their seminal book \cite{nesterov1994interior}. Note that convex optimization problems can be solved via several efficient algorithms, some of which may provide better practical efficiency. The goal here is not to pinpoint the best algorithm but rather to provide a bound on the computational complexity of solving the MVAIR problem via the convex optimization models described in Section (\ref{subsec:MVAIR-optModel}).

Note that the Problem (\ref{mod:MVAIR-Polytope-Efficient}) can be rewritten as an unconstrained optimization problem
\begin{equation}
\label{mod:MVAIR-Polytope-Efficient-Barrier}
    \minimize_{x^l ,x^u} \quad -\sum_{j=1}^d \log(x^u_j-x^l _j)  + \sum_{i=1}^n I_{-} \left( \sum_{j=1}^d (p_{ij}^+x^u_j-p_{ij}^-x^l_j) - b_i  \right),
\end{equation}
where the indicator function is defined as 
\[
I_{-}(y) = \left\{\begin{matrix} 
0 & y \leq 0 \\
\infty & y > 0.
\end{matrix}\right.
\]
In Problem (\ref{mod:MVAIR-Polytope-Efficient-Barrier}) the constraints are implicitly incorporated in the objective function. 
The indicator barrier function is a non-smooth function. However, we can efficiently approximate it with a logarithmic barrier function, which is smooth. This approximation can be written as $\hat{I}_{-}(y) = -(1/t)\log(-y)$ with $\dom \hat{I}_{-} = -\mathbb{R}_{++}$, i.e., the set of strictly negative real numbers. Here $\tau>0$ is the barrier parameter that controls the accuracy of this approximation. As $\tau$ increases in each iteration with $\tau:=\mu \tau$, where $\mu>0$ is an increment parameter, the approximation becomes more accurate. Problem (\ref{mod:MVAIR-Polytope-Efficient-Barrier}) can now be approximated by 
\begin{equation}
\label{mod:MVAIR-Polytope-Efficient-LogBarrier}
 \minimize_{x^l ,x^u} \quad f(x^u,x^l)= - \tau \left( \sum_{j=1}^d \log(x^u_j-x^l_j) \right) - \left(\sum_{i=1}^{n} \log \left(b_i-\sum_{j=1}^d (p_{ij}^+x^u_j-p_{ij}^-x^l_j)\right) \right)\; ,
\end{equation} 
with $f:\;\mathbb{R}^{2d} \rightarrow \mathbb{R}$ and the feasible convex and compact domain  
$G=\{(x^u,x^l)\: \mid \: \sum_{j=1}^d (p_{ij}^+x^u_j-p_{ij}^-x^l_j) -b_i \leq 0, \; x^l - x^u \leq 0\}$.
For simplicity, let $x^{ul}$ to denote the solution pair $(x^u,x^l)$ and let $f_0(\cdot)$ to be the original objective function and $\phi(\cdot)$ to be the barrier term. Hence, we have 
$f(x^{ul}) = -\tau f_0(x^{ul}) -\phi(x^{ul})$. It is known that by sequentially updating $\tau$ with $\tau:=\mu \tau$ we converge to the optimal solution when $\tau \rightarrow \infty$, tracing a central path $\mathcal{P} = \{x^{ul^*}(\tau)\: : \: \tau \geq 0\}$ \cite{nesterov1994interior,nemirovski2008interior}.
The objective function $f$ is closed, smooth, continuously differentiable, and strictly convex. In addition, $f$ is a ``self-concordant'' function on $G$ for all real values of 
 $\tau\geq 0$. 
This is due to the invariance property of self-concordant functions under scaling 
and addition operations and the fact that each of the negative logarithm terms is a self-concordant function.

Let's begin the analysis with the pre-processing operations. The first observation is that we can fairly assume that the Problem (\ref{mod:MVAIR-Polytope-Efficient}) is strictly feasible as $C$ is a compact convex set with a non-empty interior. This means the Slater's condition holds. To find a strictly feasible solution as the starting point of the algorithm, which removes the necessity of an infeasible starting step and simplifies the analysis, we can first choose $d+1$ arbitrary but affinely independent points on the boundary of $C$. For a polytope these $d+1$ points are readily available by the given set of vertices of $C$, but for a general convex set $C$, it would depend on the structure of $C$. 
 Then, given $d+1$ points $p^1,...,p^{d+1}\in \partial C$, we have the simplex $S=\Conv(p^1,...,p^{d+1}) \subset C$. Let $H_1=\Conv(p^1,p^2,...,p^{d})$ and $H_2=\Conv(p^1,p^3,...,p^{d+1})$ 
 be two of the facets of $S$ that have at least $p^1$ in common. Let $y^1$ and $y^2$ be the centroid (or median) of the $d$ vertices forming $H_1$ and $H_2$, respectively. Each of these centroidal or median points could be found by taking the mean over the $d$ vertices forming the corresponding facet or by solving a $d$-dimensional 1-median problem with $d$ input points. Note that $y^i\in\relint(H_i),\; i=1,2$. We shall have $y^1\neq y^2$ since the points $p^1,...,p^{d+1}$ were affinely independent.  Finally, let the pair
\[
x^u_{init} = \left(\begin{matrix}
\max\{y^1_{1},y^2_{1}\} \\
\max\{y^1_{2},y^2_{2}\} \\
\vdots \\
\max\{y^1_{d},y^2_{d}\} 
\end{matrix}\right) \; , \qquad 
x^l_{init} = \left(\begin{matrix}
\min\{y^1_{1},y^2_{1}\} \\
\min\{y^1_{2},y^2_{2}\} \\
\vdots \\
\min\{y^1_{d},y^2_{d}\} 
\end{matrix}\right)\; ,
\]
be the initial strictly feasible solution. Note that $x^{ul}_{init} = (x^u_{init}, x^l_{init})$ could be a degenerate solution, i.e., we could have 
$\Vol(R_{init}) = 0$.

An alternative way for constructing $x^{ul}_{init}$ that works for all convex sets is as follows. Find  the minimum volume axis-aligned bounding box $B$ of the convex domain 
 $G=\{(x^u,x^l)\: \mid \: x^l \in C\,, \; x^u \in C\,, \; x^l - x^u \leq 0\}$ and then let $y^1$ and $y^2$ be the diagonally opposing vertices of $B$ with the minimum and the maximum coordinates, respectively. Set  $x^{ul}_{mid} = \frac{1}{2}(y^1 + y^2)$. We have $x^{ul}_{mid} \in G$, since $\Vol(G) \geq \Vol(B)/2$ due to the convexity of $G$. If $x^{ul}_{mid} \in \interior G$ then set $x^{ul}_{init} := x^{ul}_{mid}$. Otherwise, find the vector $\nu = e^k - \frac{(e^k)^T (y^2 - y^1)}{\norm{y^2 - y^1}^2} (y^2 - y^1)$, where the index $k$ corresponds to the smallest component of $(y^2 - y^1)$ and $e^k$ is the $k$th column of the $d\times d$ identity matrix. Then, find a sufficiently small $\delta_1 >0$ such that either $x^{ul}_{mid} + \delta_1 \nu$ or $x^{ul}_{mid} -\delta_1 \nu$ is in the interior of $G$, i.e., strictly feasible. To increase the depth of strict feasibility of the initial solution and thus its quality we can take one further step. Without loss of generality, assume the direction $-\nu$ gives the strictly feasible solution. Let $\delta_2^{max}$ be the maximum value of $\delta_2 > \delta_1$ such that $x^{ul}_{mid} -\delta_2 \nu \in G$. Let  $x^{ul}_{ray} = x^{ul}_{mid} -\delta_2^{max} \nu$ and set $x^{ul}_{init} := \frac{1}{2} (x^{ul}_{mid} + x^{ul}_{ray})$.

Algorithm \ref{alg:MVAIR_LogBarrier} describes our logarithmic barrier method. It has an outer iteration loop in which the barrier parameter is updated and an inner iteration loop (centering step) in which usually Newton's method, with a backtracking line search for finding a reasonable step size in each iteration, is used to reach the minimizer for any given $\tau$. Based on the analysis of the log-barrier method for self-concordant convex functions, the logarithmic barrier method spends $N^{(0)}$ iterations (Newton steps) in the first centering step to reach a point sufficiently close to the central path of Problem (\ref{mod:MVAIR-Polytope-Efficient-LogBarrier}). It then takes $N^{(CP)}$ iterations (Newton steps) in the path following step, during which the algorithm iteratively updates the parameter $\tau$ with $\tau:=\mu \tau$ and tries to follow the central path as $\tau \rightarrow \infty$, to get sufficiently close to the optimal solution. Note that the other end of this central path that is associated with $\tau \rightarrow + 0$ is the analytic center of $G$ with respect to the barrier function $\phi$. Therefore, the total number of iterations required to solve Problem (\ref{mod:MVAIR-Polytope-Efficient-LogBarrier}) using the logarithmic barrier method is $N= N^{(0)} + N^{(CP)}$.

\begin{algorithm}[tbp]
    \SetAlgoLined
    \BlankLine
    \KwIn{Given convex functions $f_i:\mathbb{R}^{2d}\rightarrow \mathbb{R},\; i=0,...,n$, convex function $\phi:\mathbb{R}^{2d}\rightarrow \mathbb{R}$ with $\phi(x,y)=\left(\sum_{i=1}^{n} \log \left(-f_i(x,y)\right) \right)$, initial barrier parameter $\tau:=\tau^{(0)} > 0$, increment parameter $\mu > 1$, tolerance $\varepsilon > 0$, and a strictly feasible initial solution $x^{ul}_{0}=(x^u_0,x^l_0)\in\mathbb{R}^{2d}$ within proximity $\kappa >0$ of 
    the central path $\mathcal{P}$.}
    \KwOut{A $(1-\varepsilon)$--approximation solution to the optimal solution.} 
    \BlankLine
    \While {$n/\tau \geq \varepsilon$}{
        \tcc{Centering Step}
        Compute $x^{ul^{*}}(\tau)$ by minimizing $f(x^{ul}) = -\tau f_0(x^{ul}) - \phi(x^{ul})$, starting at $x^{ul}(\tau)$\;
        \tcc{Path Following Step}
        Set $x^{ul}(\tau) := x^{ul^{*}}(\tau)$\;
        Set $\tau:=\mu \tau$\;
    }
    \Return{$x^{ul^{*}}(\tau)$ and $f_0(x^{ul^{*}}(\tau))$}\;
    \caption{
    	${\sf Logarithmic Barrier Algorithm}$;
        it solves a constrained convex optimization problem with sequence of unconstrained optimization problems, where the obtained minimizer of each iteration is used as the staring point of the next iteration. 
    }
    \label{alg:MVAIR_LogBarrier} 
\end{algorithm}

For the number of iterations in the initial centering step, we need to define the Minkowsky function of a convex domain.
\begin{Def}
The Minkowsky function of a convex domain $G\subset \mathbb{R}^{2d}$ with its pole at $x\in G$ is 
\[
\pi_x (y) = \inf\{ \eta \geq 0 \: \mid \: x+\eta^{-1}(y-x) \in G\}\,, \quad \forall y \in \mathbb{R}^{2d}.
\]
\end{Def}
Geometrically speaking, consider a ray $[x,y)$ and let $y'$ be the point this ray intersects $\partial G$. If $y'$ exists, then $\pi_x (y)$ is the length of the segment $[x,y]$ divided by the length of the segment $[x,y']$. If $G$ is unbounded and the ray $[x,y)$ is contained in $G$, then $\pi_x (y)=0$. In other words, $\pi_x (y)$ measures the distance between $x$ and $y$ relative to the distance of $x$ to the boundary of $G$ in the direction $y-x$. 

For the main path following scheme to work, we need to start from an initial point sufficiently close to the central path $\mathcal{P}$. We can first move from $x^{ul}_{init} \in \interior G$ to the beginning of the central path, i.e., the analytic center $x^{ul}_{ac} \in \mathcal{P}$, and from there, we can follow $\mathcal{P}$ to converge to the optimal solution. It is proven that, with tolerance $\kappa$, we can converge to the analytic center $x^{ul}_{ac}$ starting from the strictly feasible solution $x^{ul}_{init}$ in 
\[
N^{(0)} = \mathcal{O}(1)\: \sqrt{n} \: \log (\frac{n}{1-\pi_{x^{ul}_{ac}} (x^{ul}_{init})} )\; ,
\]
Newton iterations, where $\mathcal{O}(1)$ are constant factors depending solely on the path tolerance $\kappa$ and the penalty rate used in this initial centering step \cite{nemirovski2004lecture}. The tolerance (accuracy) in this step does not need to be very small. Also, note that $0\leq \pi_{x^{ul}_{ac}} (x^{ul}_{init}) < 1$, since $x^{ul}_{init} \in \interior G$. The smaller it is the better our initial solution $x^{ul}_{init}$, i.e., further away from the boundary and closer to the analytic center. Due to the construction of $x^{ul}_{init}$, we expect the ratio $\frac{1}{1-\pi_{x^{ul}_{ac}} (x^{ul}_{init})}>1$ to be fairly small for all instances of the problem. Since $N^{(0)}$ still depends on the unknown point $x^{ul}_{ac}$, we can bound it from above using a \emph{symmetry measure} proposed by Minkowski \cite{minkowski1911allegemeine}. 
\begin{Def}
The Minkowski symmetry measure of a convex set $G$ with respect to $x\in G$ is defined as 
\[
\textrm{sym}_{G}(x) = \max\{\eta \geq 0 \: : \: x+\eta (x-y) \in G, \; \forall y\in G\}
\]
\end{Def}
Note that we have $\textrm{sym}_{G}(x)=0,\; \forall x\in \partial G$. We also have $\textrm{sym}_{G}(x) \leq 1,\; \forall x\in G$, where the equality holds when $G$ is symmetric around $x$, which also means $G$ is symmetric. 

Let $x' = x^{ul}_{ac} + \frac{1}{\pi_{x^{ul}_{ac}} (x^{ul}_{init})}(x^{ul}_{init}-x^{ul}_{ac})$ be the point on the boundary of $G$ where the ray $[x^{ul}_{ac}, x^{ul}_{init})$ crosses the boundary. We have 
\begin{eqnarray*}
1-\pi_{x^{ul}_{ac}} (x^{ul}_{init}) & = & 1 - \frac{\dist(x^{ul}_{init},x^{ul}_{ac})}{\dist(x',x^{ul}_{ac})} 
\; = \; \frac{\dist(x',x^{ul}_{init})}{\dist(x',x^{ul}_{ac})} \\
& = & \max\{\tau \geq 0 \: : \: x^{ul}_{init}+\tau (x^{ul}_{init}-x^{ul}_{ac}) \in G\} \\
& \geq & \textrm{sym}_{G}(x^{ul}_{init}) > 0
\end{eqnarray*}
Therefore, we have 
\[
N^{(0)} = \mathcal{O}(1)\: \sqrt{n} \: \log (\frac{n}{1-\pi_{x^{ul}_{ac}} (x^{ul}_{init})} ) \leq \mathcal{O}(1)\: \sqrt{n} \: \log (\frac{n}{sym_{G}(x^{ul}_{init})} )  \,.
\] 
Note that the algorithm does not need to compute $\textrm{sym}_{G}(x^{ul}_{init})$. 

To start the main path following scheme, in Algorithm \ref{alg:MVAIR_LogBarrier}, we can let $x^{ul}_{0}$ to be the approximate solution $\widehat{x^{ul}_{ac}}$ that was achieved in the initialization phase for $x^{ul}_{ac}$ and set 
\begin{equation}
\label{eq:t0}
\tau^{(0)} = \max\{\tau \: \mid \: \lambda(\widehat{x^{ul}_{ac}}) \leq \kappa\}\; ,
\end{equation}
where $\tau^{(0)}$ is the initial choice of barrier parameter for the main path following step and $\lambda(x)$ is the Newton decrement defined as $\lambda(x)= \sqrt{\nabla f(x) \nabla^2 f(x)^{-1} \nabla f(x)}$ with $f$ being the objective function in (\ref{mod:MVAIR-Polytope-Efficient-LogBarrier}).
Let $\mu=1+1/\sqrt{n}$. This gives $\log \mu = \log (1+1/\sqrt{n}) \geq \log 2 / \sqrt{n}$ and $\mu-1-\log \mu \leq 1/(2n)$. For the number of iterations during the main path following phase, we have
\begin{eqnarray*}
N^{(CP)} 
& = & \ceil[\Bigg]{\frac{\log (\frac{n}{\tau^{(0)} \varepsilon})}{\log \mu}}   \left( \frac{f(x^{ul^{*}}(\tau))-f(x^{ul^{*}}(\mu \tau))}{\gamma} + \log_2 \log_2 (\frac{1}{\varepsilon}) \right) \\
& = & \ceil[\Bigg]{\frac{\log (\frac{n}{\tau^{(0)} \varepsilon})}{\log \mu}}   \left( \frac{n(\mu-1-\log \mu)}{\gamma} + c \right) \\
& \leq & \ceil[\bigg]{\frac{\sqrt{n}\log (\frac{n}{\tau^{(0)} \varepsilon})}{\log 2}} \left(\frac{1}{2\gamma} + c \right)  \\
& = & \ceil[\bigg]{\sqrt{n}\log_2 (\frac{n}{\tau^{(0)} \varepsilon})} \left(\frac{1}{2\gamma} + c \right)  \\
& \leq & \left(1+\sqrt{n}\log_2 (\frac{n}{\tau^{(0)} \varepsilon})\right) \left(\frac{1}{2\gamma} + c \right),
\end{eqnarray*}
where $n$ is the number of inequalities defining the polytope, $\gamma$ is a constant lower bound (given in Eq. (\ref{eq:gamma}) below) on the reduction amount in the objective function in each iteration during the damped Newton phase, $\varepsilon$ is the required accuracy for the optimal solution, $c=\log_2 \log_2 (1/\varepsilon)$, and $x^{ul^{*}}(\mu \tau)$ is the optimal solution of the centering step starting from $x^{ul^{*}}(\tau)$ after updating the barrier parameter in the outer loop. 
In the first equality, the first term is the number of iterations in the outer loop and the second term is the number of Newton steps in the inner loop per centering iteration. The ratio $n/\tau^{(0)}$ is the initial duality gap and $\varepsilon$ is, in fact, the final duality gap. The second equality is derived by extracting the dual function in the first fraction of the second term and then simplifying it using the duality gap when the barrier parameter is $\tau$. The third equality is derived by assuming a fixed value for $\mu$ as $\mu=1+1/\sqrt{n}$. 

The term $c=\log_2 \log_2 (1/\varepsilon)$ is an upper bound on the number of iterations during the quadratically convergent phase of Newton's method and has a very weak dependence on the inverse of $\varepsilon$ and can be effectively considered as constant; for $\varepsilon= 10^{-9}$, this is less than 5. Also, the parameter $\gamma$ depends, weakly, on the backtracking line search parameters $\alpha$ and $\beta$ and is equal to 
\begin{equation}
\label{eq:gamma}
\gamma = \frac{\alpha\beta(1-2\alpha)^2}{20-8\alpha},
\end{equation}
where $0<\alpha<0.5$ and $0<\beta<1$. For $\alpha=0.2$ and $\beta=0.9$, we have $\frac{1}{2\gamma} < 142$. Finally, note that $N^{(CP)}$ does not depend on the dimension $d$. 

Finally, the total number of Newton steps is 
\begin{eqnarray*}
N & = & N^{(0)} + N^{(CP)} \leq \mathcal{O}(1) \: \sqrt{n} \: \log_2 \left(\frac{n}{ (sym_{G}(x^{ul}_{init})) \tau^{(0)} \varepsilon }\right)\; , 
\end{eqnarray*}
where $\mathcal{O}(1)$ is an absolute constant. Following a more refined analysis, such as the original analysis of Nesterov and Nemirovski \cite{nesterov1994interior}, the bound could be tightened by finding smaller and more accurate constants. Furthermore, by carefully choosing the input parameters of the method, the condition in Eq. (\ref{eq:t0}) can be guaranteed and $\tau^{(0)}$ can be removed from the right hand side of the bound. Hence, the process is terminated with a $(1-\varepsilon)$-solution in $N = \mathcal{O}\left(\sqrt{n} \: \log (\frac{n}{\varepsilon})\right) $ iterations. Alternatively, if we can establish a strictly feasible dual solution, then the ratio $n/\tau^{(0)}$ can be replaced by the constant $\varepsilon_0$, the initial duality gap, leading to $\mathcal{O}\left(\sqrt{n} \: \log (\frac{\varepsilon_0}{\varepsilon})\right) $ running time. Note that this is a conservative upper bound and in practice the algorithm performs  
better and in many cases, depending on the structure of the problem, it works just in a few iterations independent of the size of the problem. In general, the observed average run time of path-following methods is $\mathcal{O}(\log n \log(\frac{\varepsilon_0}{\varepsilon}))$ \cite{bertsimas1997introduction}.

Each Newton step (inner iteration) is equivalent to solving a linear system of equations 
$H \Delta x^{ul} =-g$, where 
$H=\nabla^2 f(x^{ul})$ is the Hessian matrix, 
$g=\nabla f(x^{ul})$ is the gradient vector, and 
$\Delta x^{ul}$ is the Newton direction. This will cost $\mathcal{O}\left((2d)^3\right)$ arithmetic operations for solving the linear system 
plus the costs of computing (forming) $g$ and $H$. Computing $g$ and $H$ require at most $\left(2d\times(n+1)\right)=\mathcal{O}(dn)$ and $\left((2d)^2\times(n+1)\right)=\mathcal{O}(d^2 n)$ operations, respectively. Therefore, the computational complexity of  the log-barrier method for solving Problem (\ref{mod:MVAIR-Polytope-Efficient}) in each step is $\mathcal{O}(d^3+d^2n)$.  
Nevertheless, certain structures of the problem could be exploited to reduce this bound in practice.  

Therefore, the total computational complexity of solving Problem (\ref{mod:MVAIR-Polytope-Efficient}) using the log-barrier method summarized in Algorithm \ref{alg:MVAIR_LogBarrier} is $\mathcal{O}((d^3+d^2 n)\sqrt{n}\log \frac{n}{\varepsilon})$. For a polytope, since $C$ is bounded we have $n > d$ and thus $d^3$ is dominated by $d^2 n$ leading to $\mathcal{O}(d^2 n\sqrt{n}\log \frac{n}{\varepsilon})$ time. 

Since this analysis is not based on an assumption restricting it to polytopes, the result is the same for finding the MVAIR in general convex sets, 
which can be easily represented in a finite set of inequalities. For example, for $d$-dimensional ellipses we obtain $\mathcal{O}(d^3)$ running time since we just need one inequality to define an ellipsoidal convex set.

Finally, it should be mentioned that for a general convex set the centroid may not be easily computable or at least not as easy as that of polytope. In this case we may prefer or have to solve a 1-median problem. Since this problem could be formulated as a second-order cone program (SOCP), the preprocessing time for finding the initial strictly feasible solution takes at most $\mathcal{O}(d^3\sqrt{d}\log \frac{1}{\varepsilon_m})$ time using the primal-dual potential reduction algorithm of \cite{lobo1998applications}. Here, $\varepsilon_m$ is the accuracy for finding the exact median point, which could be considered a fairly large number (e.g., $10^{-1}$) as the exact median point is not needed for constructing $R_{init}$. So the pre-processing time the first way of finding $x^{ul}_{init}$ is essentially $\mathcal{O}(d^3\sqrt{d})$. 
The pre-processing time for the alternative way depends on the geometry of $C$ and the algorithm used for finding its minimum volume axis-aligned bounding box. Let $T_x (C)$ be the time that it takes to find an extreme point in a given direction in a convex set $C$. Then this box can be found in $\mathcal{O}(d T_x (C) )$. For example, for a convex polygon, this can be done in $\mathcal{O}(\log n)$ using binary search, if the $n$ vertices are given as an array in a c.w. or c.c.w. order. 
 It must be mentioned that this preprocessing for finding a strictly feasible starting point is not an essential part of the algorithm as the algorithm could have an infeasible starting point for the initial centering step as well. However, in that case, the analysis and the upper bound on the number of iterations $N$ are a bit different, although the bound would still grow with $\sqrt{n}\log \frac{n}{\varepsilon}$.

The following theorem summarizes the complexity analysis of solving MVAIR.  
\begin{thm}
\label{thm:MVAIR_Complexity}
The problem of finding the MVAIR in a compact convex set in $C\subset\mathbb{R}^d$ defined by $n$ convex inequalities can be solved to a $(1-\varepsilon)$--approximation solution by the logarithmic barrier algorithm in $\mathcal{O}((d^3+d^2 n) \sqrt{n}\log \frac{n}{\varepsilon})$ time. 
When $C$ is a polytope, this is $\mathcal{O}(d^2 n \sqrt{n}\log \frac{n}{\varepsilon})$. 
\end{thm}

\subsection{Solving the MVIR}
\label{subsec:SolvingMVIR}
As discussed in Section \ref{subsec:MVIR-optModel}, the Problem (\ref{mod:MVIR-Set}), i.e., the MVIR problem in the higher dimension, is a non-convex optimization problem with an exponential number of constraints and is difficult to solve even for the special case of Problem (\ref{mod:MVIR-Polytope}), where the convex set $C$ is a polytope. Solving the MVIR problem efficiently in higher dimensions would require further exploitation of the structure of the problem and the properties of the optimal solution that is considered as a future research direction for this study. In the rest of this section, we limit ourselves to the 2D version, which is the MAIR problem.

\subsubsection{A Parametric Approach for Finding the MAIR}
\label{subsubsec:Parametric-Approach-2D}
This section provides a parametrized optimization approach for the MAIR problem in a compact convex set $C \subset \mathbb{R}^2$. 
Consider the 2D version of the Problem (\ref{mod:MVIR-Set})
 and let $u=u^1,\;v=u^2,\;x=x^1,\;y=x^2,\;z=x^3$; see Figure \ref{fig:LIR-image} for an illustration for the special case when $C$ is a convex polygon. Then finding the MAIR in $C$ can be formulated as: 
\[
\begin{array}{lll}
(Q)   & \maximize & |\det ( u, v )| \\
        & \st  & u^T v = 0  \\
        &      & u=y-x,\, v= z-x \\
        &      & x,y,z,y+z-x \in C \subset \mathbb{R}^2.
\end{array}
\]

Clearly, the last two constraints in $(Q)$ can be rewritten as $(u,v) \in S \subseteq \mathbb{R}^4$, where $S=\{(u,v)\,|\,  u=y-x,\, v= z-x, \; \mbox{and}\; x,y,z,y+z-x \in C\}$ is a compact convex set.
We have 
\[
\begin{array}{lll}
(Q')     & \maximize & |u_1 v_2 - u_2 v_1| \\
         & \st  & u_1 v_1 + u_2 v_2 = 0 \\
         &      & (u,v) \in S \subset \mathbb{R}^4.
\end{array}
\]
Let $\theta$ to be the angle between the vector $u$ and the $x$-axis and let $t=\tan(\theta)$. We have $t=u_2/u_1$ and $-1/t=v_2/v_1$, and the above problem reduces to a parameterized model
\[
\begin{array}{lll}
(Q_t)    & \maximize & (1+t^2) |u_1 v_2| \\
         & \st  & u_2 - t u_1 =0 ,\, v_1 + t v_2 = 0,  \\
         &      & (u,v) \in S \subset \mathbb{R}^4.
\end{array}
\]
For any fixed $t$, $(Q_t)$ can be solved by sequentially solving four separate subproblems:
\[
\begin{array}{ll}
 \maximize & u_1 v_2 \\
 \st  & u_1 \ge 0,\, v_2 \ge 0, \\
      & u_2 - t u_1 =0 ,\, v_1 + t v_2 = 0,  \\
      & (u,v) \in S \subset \mathbb{R}^4;
\end{array}
\]
\[
\begin{array}{ll}
 \maximize & u_1 v_2 \\
 \st  & u_1 \le 0,\, v_2 \le 0, \\
      & u_2 - t u_1 =0 ,\, v_1 + t v_2 = 0,  \\
      & (u,v) \in S \subset \mathbb{R}^4;
\end{array}
\]
\[
\begin{array}{ll}
 \maximize & -u_1 v_2 \\
 \st  & u_1 \ge 0,\, v_2 \le 0, \\
      & u_2 - t u_1 =0 ,\, v_1 + t v_2 = 0,  \\
      & (u,v) \in S \subset \mathbb{R}^4;
\end{array}
\]
\[
\begin{array}{ll}
 \maximize & -u_1 v_2 \\
 \st  & u_1 \le 0,\, v_2 \ge 0, \\
      & u_2 - t u_1 =0 ,\, v_1 + t v_2 = 0,  \\
      & (u,v) \in S \subset \mathbb{R}^4.
\end{array}
\]

\[
\begin{array}{ll}
 \maximize & \log u_1 + \log v_2 \\
 \st  & u_2 - t u_1 =0 ,\, v_1 + t v_2 = 0,  \\
      & (u,v) \in S \subset \mathbb{R}^4;
\end{array}
\]
\[
\begin{array}{ll}
 \maximize & \log (-u_1) + \log (-v_2) \\
 \st  & u_2 - t u_1 =0 ,\, v_1 + t v_2 = 0,  \\
      & (u,v) \in S \subset \mathbb{R}^4;
\end{array}
\]
\[
\begin{array}{ll}
 \maximize & \log u_1 + \log (-v_2) \\
 \st  & u_2 - t u_1 =0 ,\, v_1 + t v_2 = 0,  \\
      & (u,v) \in S \subset \mathbb{R}^4;
\end{array}
\]
\[
\begin{array}{ll}
 \maximize & \log (-u_1) + \log v_2 \\
 \st  & u_2 - t u_1 =0 ,\, v_1 + t v_2 = 0,  \\
      & (u,v) \in S \subset \mathbb{R}^4 .
\end{array}
\]

Let the optimal value of $(Q_t)$ be $f(t)$. Finding the maximum area rectangle inside $S$ can be achieved by sequentially solving the parameterized problems and then identifying $t$ to maximize $f(t)$.

Observe that $x,y,z$ are chosen arbitrarily in $C$, and the role of $u=y-x$ and $v=z-x$ are symmetric. Therefore, one does not need to go through all four cases; it suffices to focus only on the first case, and the parameter $t$ can also be restricted to be nonnegative. In other words, we need only to consider
\begin{eqnarray}
\label{mod:MAIR_Parametric_GivenDirection}
    & \maximize & \log u_1 + \log v_2     \\
    & \st               & u_2 - t u_1 =0 ,\, v_1 + t v_2 = 0 \; ,  \nonumber \\
    &                   & (u,v) \in S \subset \mathbb{R}^4 \; , \nonumber
\end{eqnarray}

for any given $t\geq 0$. This is a convex optimization problem, for any fixed $t$, that has a \emph{unique} optimal solution (i.e., the MAIR with respect to the direction $t$), since the feasibility set is nonempty, compact, and convex, and the objective function is closed and strictly concave. Let $\psi(t)$ to denote the optimal value function of (\ref{mod:MAIR_Parametric_GivenDirection}). Then, $f(t) = (1+t^2) e^{\psi(t)}$, and the parametric search in $(Q_t)$ boils down to the one-dimensional optimization: $\maximize_{t\geq 0}  f(t)$.

To make the domain of $t$ bounded, we need one more step. 
Note that, $t\geq 0$ is equivalent to $\theta \in[0,\pi/2]$.

\begin{observ}
\label{observ:AngelsSet}
Any rectangle with $\pi/4<\theta < \pi/2$ has an identical counterpart with $\theta'=\theta-\pi/2$ and $\theta'\in(-\pi/4,0)$, which satisfies conditions $u(1)>0, \; v(2)>0$ and can be obtained by the linear transformation $z\rightarrow x',\; x\rightarrow y',\; y+z-x\rightarrow z',\; u'=-v,\; v'=u$. Between the identical rectangles, we have $|\theta-\theta'|=\pi/2$. Using the same transformation, we obtain that the rectangle with $\theta=\pi/2$ is identical to the rectangle with $\theta'=0$ and also the rectangles with $\theta=\pi/4$ and $\theta'=-\pi/4$ are identical. Therefore, it suffices to consider only the rectangles with $\theta\in[-\pi/4,\pi/4]$, or equivalently $t\in[-1,1]$.
\end{observ}
Thus, the problem of finding the MAIR in a convex set $C \subset \mathbb{R}^2$ is
\begin{eqnarray}
\label{mod:MAIR_Parametric}
\maximize_{t} \; f(t) = (1+t^2) e^{\psi(t)} & & \st \\
-1 \; \leq \; t \; \leq \; 1\; , & &  \nonumber
\end{eqnarray}
where $\psi(t)$ is the optimal value of (\ref{mod:MAIR_Parametric_GivenDirection}). This problem has an optimal solution by the following proposition. 

\begin{prop}
\label{prop:Weirestrass}
The function $f(t)$ attains its maximum.
\end{prop}
\begin{proof}
Since (\ref{mod:MAIR_Parametric_GivenDirection}) is a convex optimization problem for any given $t$, its set of maximizers is also convex. 
Moreover, the feasibility set of (\ref{mod:MAIR_Parametric_GivenDirection}) is compact, its objective function is closed, and Slater's condition holds. Therefore, the set of maximizers is nonempty, closed, and bounded as well; in fact, it is a singleton due to the strict concavity of the objective function. Adding the fact that the linear independence constraint qualification (LICQ) holds for the optimal solution $(u^*(\bar{t}),v^*(\bar{t}))$ for any given $\bar{t}$, 
 we obtain that the optimal value function $\psi(t)$ is upper semi-continuous on its domain $-1\leq t \leq 1$ \cite{fiacco1990sensitivity,still2018lectures}. 
Hence, the function $f(t)$ is upper semi-continuous on $-1\leq t \leq 1$, since $e^x$ is continuous on $\mathbb{R}$ 
and the function $(1+t^2)$ is continuous on $-1\leq t \leq 1$. 
Moreover, the domain of $t$ is a compact set, i.e., the interval $[-1,1]$. Therefore, by the  
Weierstrass extreme value theorem, $f(t)$ attains its supremum.
\end{proof}

However, finding this maximum could be difficult. This is because the function $f(t)$ is not explicitly formulated, since it depends on $\psi(t)$, and it has some undesirable properties. Even though it is an upper semi-continuous and univariate function, it is non-smooth, non-differentiable, non-unimodal, and non-concave and in some cases, it could be a very ill-behaved function 
(see Figures \ref{fig:6a} and \ref{fig:Random2a} in Section \ref{sec:EvaluationAndExamples} for an example). Therefore, it is very difficult to design an algorithm that guarantees to find the optimal solution. Figure \ref{fig:sensitivity} visualizes this difficulty by showing the sensitivity of the area of the inscribed rectangles to the direction $t$. Nevertheless, being restricted to a one-dimensional search enables us to develop a good approximation algorithm.

Such an algorithm will not only solve the MAIR problem but also provides, by fixing the direction, a much more general algorithm for the 2-dimensional case of the MVAIR, i.e., the MAAIR problem in a convex set $C \subset \mathbb{R}^2$. It is a more general algorithm than the interior point algorithm presented in Section \ref{subsec:SolvingMVAIR}, in the sense that it finds the largest rectangle aligned to not only the regular axes but also to any rotated axes in any direction without performing the rotation. 

\begin{figure}[t]
  \centering
  \includegraphics[width=0.4\textwidth]{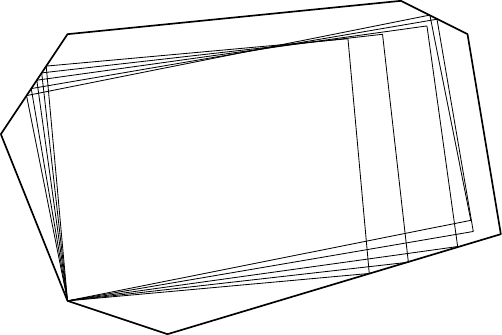}
  \caption{Sensitivity of the inscribed rectangles to the direction assuming a fixed vertex-corner. }
  \label{fig:sensitivity}
\end{figure}

\subsubsection{Finding the MAAIR in any Given Direction (MAAIR-$\mbox{F}_{\mbox{dir}}$)}
\label{subsubsec:Parametric-MAAIR}
In this section we provide an optimization approach for finding the largest inscribed rectangle in any given direction, i.e. aligned to any rotated axes, in a convex set. The following theorem states the main result of this section.
\begin{thm}
For a compact convex set, $C \subset \mathbb{R}^2$ defined by $n$ convex inequalities, a $(1-\varepsilon)$--approximation to the MAAIR regarding the perpendicular axes in any given direction can be found in 
 $\mathcal{O}(n \sqrt{n} \log \frac{n}{\varepsilon})$ time. 
\end{thm}
\begin{proof}
Let's assume, without loss of generality, that $C$ is a convex polygon. The problem (\ref{mod:MAIR_Parametric_GivenDirection}) can be expanded as
\begin{eqnarray}
\label{model-10}
    &\textrm{maximize} & \log u_1 + \log v_2    \\
    &\st   &   \nonumber \\
    & & \begin{array}{rcl}
            u_2 - t u_1 & = & 0  \nonumber \\
            v_1 + t v_2 & = & 0  \nonumber \\
            u-y+x & = & 0  \nonumber \\
            v-z+x & = & 0  \nonumber \\
            Px & \leq & b   \nonumber \\
            Py & \leq & b   \nonumber \\
            Pz & \leq & b   \nonumber \\
            P(y+z-x) & \leq & b \nonumber
    \end{array}
\end{eqnarray}
where $P\in \mathbb{R}^{n\times 2}$ and $b\in \mathbb{R}^n$ are the given characterizations of the convex polygon $C$. Putting $t=0$ gives the MAAIR with respect to the regular (non-rotated) axes.
 The assumption of $C$ being a polygon is not restrictive since problem (\ref{model-10}) can be defined similarly for any other closed and bounded convex set that can be defined with a ``finite'' number of convex inequalities. 

Suppose $t$ is given. We can rewrite (\ref{model-10}) in the matrix form. Define 
$s=(u^T v^T x^T y^T z^T)^{T}$, which is a $10\times 1$ vector. Then we have
\begin{eqnarray}
\label{model-11}
    &\textrm{minimize} & -\left(\log e_1^T s + \log e_4^T s \right)    \\
    &\st   &   \nonumber \\
    & & \begin{array}{rcl}
            As & = & 0  \nonumber \\
            \tilde{P}s & \leq & \tilde{b}  \nonumber
    \end{array}
\end{eqnarray}
where $e_i$ is the $i$th column of a 10-by-10 identity matrix,  
$\tilde{b}=(b_1,b_1,b_1,b_1,b_2,b_2,b_2,b_2,...,b_n,b_n,b_n,b_n)^{T}$, i.e. a $4n\times 1$ vector, $A$ is
\begin{equation*}
    A=\left[
    \begin{array}{cccccccccc}
        -t & 1 & 0 & 0 & 0 & 0 & 0 & 0 & 0 & 0 \\
        0 & 0 & 1 & t & 0 & 0 & 0 & 0 & 0 & 0 \\
        1 & 0 & 0 & 0 & 1 & 0 & -1 & 0 & 0 & 0 \\
        0 & 1 & 0 & 0 & 0 & 1 & 0 & -1 & 0 & 0 \\
        0 & 0 & 1 & 0 & 1 & 0 & 0 & 0 & -1 & 0 \\
        0 & 0 & 0 & 1 & 0 & 1 & 0 & 0 & 0 & -1
    \end{array}
    \right]_{(6\times 10)}
\end{equation*}
and $\tilde{P}$ is
\begin{equation*}
    \tilde{P}=\left[
    \begin{array}{cccccccccc}
        0 & 0 & 0 & 0 & p_{11} & p_{12} & 0 & 0 & 0 & 0 \\
        0 & 0 & 0 & 0 & 0 & 0 & p_{11} & p_{12} & 0 & 0 \\
        0 & 0 & 0 & 0 & 0 & 0 & 0 & 0 & p_{11} & p_{12} \\
        0 & 0 & 0 & 0 & -p_{11} & -p_{12} & p_{11} & p_{12} & p_{11} & p_{12} \\
        \cdots & & & & \cdots & & & & & \\
        \vdots & & & & \ddots & & & & & \\
        0 & 0 & 0 & 0 & p_{n1} & p_{n2} & 0 & 0 & 0 & 0 \\
        0 & 0 & 0 & 0 & 0 & 0 & p_{n1} & p_{n2} & 0 & 0 \\
        0 & 0 & 0 & 0 & 0 & 0 & 0 & 0 & p_{n1} & p_{n2} \\
        0 & 0 & 0 & 0 & -p_{n1} & -p_{n2} & p_{n1} & p_{n2} & p_{n1} & p_{n2}
    \end{array}
    \right]_{(4n\times 10)}
\end{equation*}

This gives us the log-barrier function as
\[
f(s)= -\tau \,\left(\log e_1^T s + \log e_4^T s \right) - \left(\sum_{i=1}^6 \log(-a_i^T s) + \sum_{i=1}^6 \log(a_i^T s)\right) - \left(\sum_{i=1}^{4n} \log(\tilde{b}_i-\tilde{p}_i^T s)\right)\,,
\]
where $a_i$ is the $i$th column of $A$ and $\tilde{p}_i$ is the $i$th column of $ \tilde{P}$.

The function $f:\mathbb{R}^{10} \rightarrow \mathbb{R}$ is convex with fixed dimension, hence this is a convex optimization problem with 
fixed dimension. Therefore, by Theorem \ref{thm:MVAIR_Complexity} and the analysis of Section \ref{subsec:SolvingMVAIR}, for any given direction $-1\leq t \leq 1$ including the traditional axis-aligned case ($t=0$), this can be solved to  
a $(1-\varepsilon)$--approximation solution by the logarithmic barrier algorithm in 
$\mathcal{O}(n \sqrt{n} \log \frac{n}{\varepsilon})$ time.
\end{proof}

\begin{rem}
This result is more general than the existing results such as the algorithm of Alt \etal \cite{alt1995computing}, as it deals with general convex sets, compared to the existing results that are limited to convex polygons. Moreover, it can find the MAAIR aligned to any rotated axes without the rotation preprocessing . For a convex polygon with $n$ vertices, one could find the rotated coordinates of all vertices in linear time and then perform one of the existing axis-aligned algorithms. However, for general sets, such change of coordinates could be computationally very expensive, and yet after such rotation, one would need an algorithm capable of finding MAAIR in general convex sets, for which to the best of our knowledge this paper proposes the first such algorithm.
\end{rem}
\begin{rem}
It is important to observe that this computational complexity depends only on $n$, the number of convex inequalities defining the set $C$. For example, for ellipses  
it will be $\mathcal{O}(1)$ as only one inequality defines an ellipsoidal convex set. This is in contrast with the computational geometric algorithms that usually approximate a convex set with a convex polygon defined by a large number of vertices. 
This means, although this algorithm may underperform over polygons it can easily outperform the polygonal approximation approaches for a wide range of convex sets.
\end{rem}

\subsubsection{An Approximation Algorithm for Finding the MAIR}
\label{subsubsec:Approx_MAIR} 
We can use this fast  
algorithm, that finds 
a $(1-\varepsilon)$--approximation to
the largest rectangle in any given direction inscribed in a compact convex set $C \subset \mathbb{R}^2$, as a subroutine to obtain an approximation algorithm for the case where we want to find the largest inscribed rectangle among all directions, i.e., the MAIR.

Consider a compact and convex set $C\in \mathbb{R}^2$. We seek to solve the Problem (\ref{mod:MAIR_Parametric}), which has a univariate objective function $f(t) = (1+t^2) e^{\psi(t)}$, where $\psi(t)$ is the optimal value of (\ref{mod:MAIR_Parametric_GivenDirection}).
For simplicity, in this section, we use the notation $|\cdot|$ as a measure of both area and length. Let $R_{opt}$ to be the optimal solution, i.e., the MAIR, and $R_{apx}$ to be the $(1-\varepsilon)$--approximation solution, i.e. $|R_{apx}|\geq (1-\varepsilon)|R_{opt}|$, that we seek to find. The basic intuition behind our algorithm is that the direction of a $(1-\varepsilon)$--approximation solution should be very close to the direction of the optimal rectangle. Suppose the optimal rectangle happens at direction angle $\theta^*$. We want to know how much the area of the rectangle changes if we change the direction slightly. So we want to find a lower bound for the area of an approximation rectangle with direction angle $\theta_{apx}\in[\theta^*-\alpha,\theta^*+\alpha]$, for some small $\alpha>0$.

\begin{lem}
\label{lem:aspect-ratio-UB} 
Let $C\subset \mathbb{R}^2$ be a compact convex set and $R_{opt} \subset C$ be the MAIR. The aspect ratio $\rho = \AR(R_{opt})$ is bounded from above. When $C$ is a polygon an upper bound that only depends on $C$ can be obtained in linear time and an upper bound that depends on a rectangular outer approximation of $C$, i.e., an enclosing rectangle, can be obtained in logarithmic time. 
\end{lem}
\begin{proof}
Let $w\geq h >0$ be the side lengths of $R_{opt}$. Also, let $R_{diam}$ be the minimum area rectangle enclosing $C$ such that its longer side is parallel to a $diam(C)$ with the same length. Let $h_{diam} \leq \diam(C)$ be the length of the side of $R_{diam}$ that is perpendicular to $diam(C)$, i.e., the width of $C$ when seen from the same direction angle of $diam(C)$. See Figure \ref{fig:AR-bound-proof} for an illustration. By the fact that $C \subset R_{diam}$ and the convexity of $C$ we have
\[
\frac{\diam(C)h_{diam}}{2} = \frac{|R_{diam}|}{2} \leq |C| \leq |R_{diam}| = \diam(C)h_{diam}
\]
Using Lassak's bound \cite{lassak1993approximation}, we have $wh = |R_{opt}| \geq |C| /2$. Therefore,
\[
\diam(C)h \geq wh =  |R_{opt}| \geq |C| /2 \geq \frac{|R_{diam}|}{4} = \frac{\diam(C)h_{diam}}{4} \geq  \frac{|C|}{4}\; , 
\]
which gives $h \geq h_{diam}/4 \geq \frac{|C|}{4\diam(C)}$. Hence
\begin{equation}
\label{eq:ARbound}
\rho = \AR(R_{opt}) = \frac{w}{h} \leq \frac{\diam(C)}{h_{diam}/4} \leq \frac{4 (\diam(C))^2}{|C|} = 4 \, \AR\,_{cvx} (C)\, ,
\end{equation}
where $\AR\,_{cvx} (\cdot)$ is the aspect ratio of a convex set as defined in Section \ref{subseq:NotationalConventions}.

If the geometry of $C$ is such that $|C|$ and $\diam(C)$ can be computed in constant time, then the upper bound is readily available. Otherwise, when $C$ is a polygon with $n$ vertices, we can find $|C|$ with $\mathcal{O}(n)$ effort by for example triangulation. Also, $\diam(C)$ can be found in $\mathcal{O}(n)$ using the idea of antipodal pairs and parallel support lines introduced by Shamos \cite{shamos1978computational} or the idea of rotating calipers introduced in \cite{toussaint1983solving}. 

For the general convex set $C$ we take one further step to bound the right hand side of (\ref{eq:ARbound}) with something else to make the computation of the upper bound simpler. Let $R_b$ be the minimum area axis-aligned bounding (enclosing) rectangle of $C$ as shown in Figure \ref{fig:AR-bound-proof}. This can be obtained in $\mathcal{O}(T_x (C) )$ time, where $T_x (\cdot)$ is the time that it takes to find an extreme point in a given direction in a convex set and would be determined according to the input. See \cite{ahn2006inscribing} for more discussion on this time. Let $p$ and $q$ be the points where $C$ touches the shorter sides of $R_b$. Let $R'$ denote the minimum area rectangle enclosing $C$ that has a side parallel to the line segment $\overline{pq}$. Let $w' \geq |\overline{pq}|$ be the length of this side and $h'$ the length of its other side. From Lemma 5 of Ahn \etal \cite{ahn2006inscribing}, we have $\diam(C) \leq \sqrt{2}\, w'$ and $|C| \geq |R'|/(2\sqrt{2})$. Thus we have
\begin{equation}
\label{eq:ARbound_loose}
\rho = \AR(R_{opt}) \leq \frac{4 (\diam(C))^2}{|C|} \leq \frac{16\sqrt{2} \,(w')^2}{|R'|} \leq 16 \sqrt{2}\, \AR(R')
\end{equation}
For a convex polygon, this weaker bound can be obtained in $\mathcal{O}(\log n)$, since $T_x (C)=\mathcal{O}(\log n)$ in this case. 
\begin{figure}[t]
  \centering
  \includegraphics[width=0.45\textwidth]{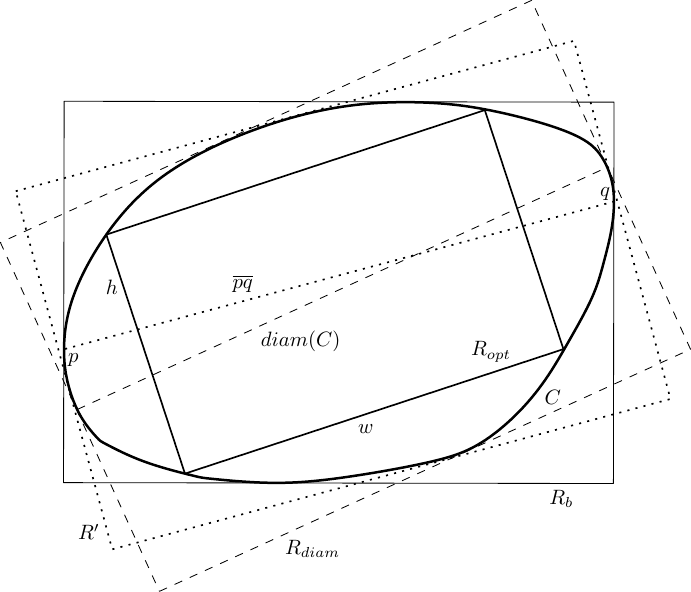}
  \caption{An illustration of the proof of Lemma \ref{lem:aspect-ratio-UB} for the upper bound on the aspect ratio of the MAIR in a convex set $C$. The dashed rectangle $R_{diam}$ is the minimum area rectangle enclosing $C$ induced by the diameter of $C$ and the dotted rectangle $R'$ is the minimum area rectangle enclosing $C$ induced by the line segment $\overline{pq}$ that connects the touching points of $C$ and the shorter sides of its smallest axis-aligned bounding box $R_b$.}
  \label{fig:AR-bound-proof}
\end{figure}
\end{proof}

\begin{rem}
The right hand side of (\ref{eq:ARbound}) is minimized when $C$ is a circle giving an upper bound of $16/\pi < 5.1$, while in that case we have $\AR(R_{opt})=1$. Note that for the purpose of our algorithm we just need to show that $\rho$ is bounded from above and the quality of this upper bound is not our goal here, but we expect that it could be improved.
\end{rem}

\begin{lem}
\label{lem:lower-bound}
Let the optimal rectangle $R_{opt}$ in convex set $C$ to have the aspect ratio $\rho$ and the direction angle $\theta^*$ with the $x$-axis. Also, let $\overline{\rho}$ be the upper bound on $\rho$ from Lemma \ref{lem:aspect-ratio-UB}. Then, $R_{apx}$, the largest inscribed rectangle with direction angle $\theta_{apx}\in[\theta^*-\alpha,\theta^*+\alpha]$, for some small $\alpha>0$, has area  $|R_{apx}|\geq(1-2\overline{\rho}\alpha)|R_{opt}|$.
\end{lem}
\begin{proof}
First note that it suffices to consider only the extreme points in $[\theta^*-\alpha,\theta^*+\alpha]$. This is because for $\theta^*$ we have $R_{apx}=R_{opt}$ and for any other non-extreme point $\theta'\in[\theta^*-\alpha,\theta^*+\alpha]$ there exists an $\alpha'$ with $0<\alpha'<\alpha$ such that $\theta'=\theta^*-\alpha'$ or $\theta'=\theta^*+\alpha'$, i.e. $\theta'$ is one of the extreme points of $[\theta^*-\alpha',\theta^*+\alpha']$. Hence, we will have $|R_{apx}|\geq(1-2\rho\alpha')|R_{opt}|\geq(1-2\rho\alpha)|R_{opt}|$.  
Also, it suffices to consider only one of the extreme cases, say $\theta^*+\alpha$, as the analysis for the other extreme case is symmetric. Now assume $R_{apx}$ has direction angle $\theta^*+\alpha$, where $R_{apx}$ is the largest inscribed rectangle for this direction and $\theta^*$ could be any angle in $[-\pi/4,\pi/4]$ which is the domain of $\theta$ from Observation \ref{observ:AngelsSet}. 

Consider Figure \ref{fig:lower-bound-proof} as an illustration for the proof.  Let rectangle $\square abcd$, in c.c.w. order, to represent $R_{opt}$ and assume, w.l.o.g., that $|\overline{ab}|\geq |\overline{bc}|$ 
and that $a$ is its lower left corner.
Draw a line from $a$ that makes the angle $\theta^*+\alpha$ with the $x$-axis. Let $g$ be the intersection of this line with $\overline{bc}$. Let $R'=\square efgh$ be the largest rectangle with direction $\theta^*+\alpha$ inscribed in $R_{opt}$ that has $g$ as a corner. Notte that corner $f$ lies on the line segment $\overline{ag}$ in the interior of $R_{opt}$. It is clear that $R' \subset R_{opt} \subset C$ and $|R_{apx}|\geq |R'|$, since $R_{apx}$ is the largest rectangle with  the direction angle $\theta^*+\alpha$ inscribed in $C \supseteq R_{opt}$.

Now, draw line segments $\overline{gi}$ and $\overline{ej}$ parallel to $\overline{ab}$ and let $R''$ to denote the rectangle $\square eigj$. For small enough $\alpha$ the line segments $\overline{ej}$ and $\overline{gi}$ do not cross $\overline{ag}$ and $\overline{eh}$, respectively. Let $k$ be the intersection of $\overline{ej}$ and $\overline{gh}$. Also, let $T_1,\,T_2$ be the right triangles $\triangle ehk$ and $\triangle gjk$, respectively. From the construction of $R'$ and $R''$ we have $\max\{|R'|-|R''|,|R''|-|R'|\}=2\max\{|T_1|-|T_2|,|T_2|-|T_1|\}$. 

To show which rectangle has a larger area, first note that we have $\alpha=\angle{bag}=\angle{iga}=\angle{cgh}=\angle{dhe}=\angle{keh}=\angle{fea}$, and $|\overline{ch}| \geq |\overline{jk}|$ and observe that $AR(R')\geq AR(R_{opt})$ and $|\overline{eh}|\geq |\overline{gh}|$.   
This gives 
\[
|\overline{cj}|=|\overline{de}|=|\overline{eh}|\sin \alpha \geq |\overline{gh}|\sin \alpha = |\overline{ch}| \; ,
\]
and
\[
|\overline{cj}|=|\overline{de}|=|\overline{eh}|\sin \alpha \leq |\overline{ag}|\sin \alpha = |\overline{bg}|\, .
\] 
Thus we have
$$ |T_2| = \frac{(|\overline{bc}|-|\overline{bg}|-|\overline{cj}|)\times |\overline{jk}|}{2} \leq \frac{(|\overline{bc}|-2|\overline{cj}|)\times|\overline{ch}|}{2}=\frac{|\overline{bc}|\times|\overline{ch}|-2|\overline{cj}|\times|\overline{ch}|}{2}, $$
and
\begin{eqnarray*}
 |T_1|  = 
  \frac{|\overline{ek}|\times |\overline{cj}|}{2}  & \geq &\frac{|\overline{eh}|\times |\overline{cj}|}{2} \\
 & = & \frac{(|\overline{cd}|-|\overline{ch}|)\times|\overline{cj}|}{2\cos \alpha} \\
 & \geq & \frac{|\overline{cd}|\times|\overline{cj}|-|\overline{cj}|\times|\overline{ch}|}{2} \\
  & \geq & \frac{(|\overline{bc}|-|\overline{cj}|)\times|\overline{ch}|-|\overline{cj}|\times|\overline{ch}|}{2} \\
 & = & \frac{|\overline{bc}|\times|\overline{ch}|-2|\overline{cj}|\times|\overline{ch}|}{2}.
 \end{eqnarray*}

Hence, $|T_1|\geq|T_2|$ and therefore $|R'|\geq |R''|$. Using 
$|\overline{ab}|\leq\rho |\overline{bc}|$, we obtain a lower bound on the area of $R''$ as
\begin{eqnarray*}
|R''|=|\overline{ej}|\times|\overline{gj}| & =& |\overline{ab}|\times(|\overline{bc}|-|\overline{bg}|-|\overline{cj}|) \\
& \geq & |\overline{ab}|\times(|\overline{bc}|-2|\overline{bg}|) \\
& = & |\overline{ab}|\times|\overline{bc}|-2|\overline{ab}|\times |\overline{bg}| \\
& = & |\overline{ab}|\times|\overline{bc}|-2|\overline{ab}|^2\tan\alpha \\
& \geq & (|\overline{ab}|\times|\overline{bc}|)-2\rho\tan\alpha(|\overline{ab}|\times|\overline{bc}|) \\
& \simeq & (1-2\rho\alpha)|R_{opt}|.
\end{eqnarray*}
Note that for sufficiently small $\alpha$ we have $\alpha\sim \tan\alpha$. Therefore, we obtain
\[
|R_{apx}|\geq|R'| \geq |R''|\geq (1-2\rho\alpha)|R_{opt}| \geq (1-2\overline{\rho}\alpha)|R_{opt}|\; ,
\]
which concludes the proof. 
\begin{figure}[t]
  \centering
  \includegraphics[width=0.5\textwidth]{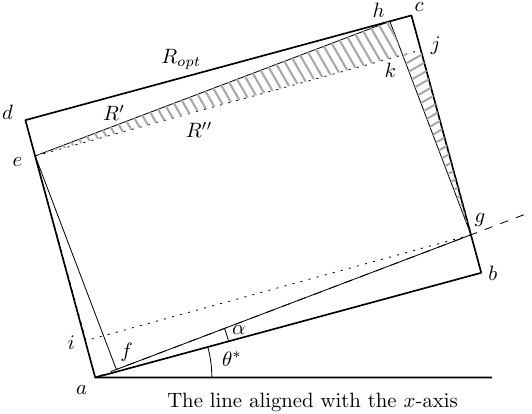}
  \caption{An illustration of the proof of the lower bound in Lemma \ref{lem:lower-bound}. The rectangle $R'$ is the largest rectangle with direction $\theta^*+\alpha$ inscribed in $R_{opt}$ and the rectangle $R''$ is induced by $R'$.}
  \label{fig:lower-bound-proof}
\end{figure}
\end{proof}

It remains to search for a direction angle that is close enough to $\theta^*$, i.e., that falls within the interval $[\theta^*-\alpha,\theta^*+\alpha]$.
\begin{thm}
\label{thm:main-theorem}
Given a compact convex set $C \subset \mathbb{R}^2$, a $(1-\varepsilon)$--approximation solution for the problem of finding the maximum area inscribed rectangle (MAIR) in $C$ can be obtained in 
$\mathcal{O}(\varepsilon^{-1} n \sqrt{n} \log \frac{n}{\varepsilon})$ time.
\end{thm}
\begin{proof}
For any given $\varepsilon>0$ take $\alpha>0$ small enough such that $\varepsilon/2=2\overline{\rho}\alpha$. We divide the interval $[-\frac{\pi}{4},\frac{\pi}{4}]$ into $\frac{\pi/2}{\alpha}=\frac{\pi/2}{\varepsilon/(4\overline{\rho})}=\frac{2\overline{\rho}\pi}{\varepsilon}$ equal pieces and sample one direction point from each piece. One could simply take these equi-distanced partition points as the sample.  
Finding $\overline{\rho}$ requires an $\mathcal{O}(T_x (C) )$ pre-processing time; for convex polygons this can be computed in $\mathcal{O}(\log n)$, as stated in Lemma \ref{lem:aspect-ratio-UB}.
Since the Problem (\ref{model-11}) is solvable for any fixed $t$, then $\psi(t)$ is defined for all $t\in[-1,1]$. We solve the subroutine of Section \ref{subsubsec:Parametric-MAAIR} for each sampled direction with precision $\varepsilon/2$, and choose the maximum value of $f(t)$ over all these samples of $t$. The result will have an area of at least $(1-\varepsilon/2)|R_{apx}| \geq (1-\varepsilon/2)^2|R_{opt}| \geq (1-\varepsilon)|R_{opt}|$, where the first $\geq$ is due to Lemma \ref{lem:lower-bound}. Considering the complexity of the subroutine, the computational complexity of this algorithm is $\mathcal{O}(\frac{2\overline{\rho}\pi}{\varepsilon} n \sqrt{n} \log \frac{n}{\varepsilon/2}) =\mathcal{O}(\varepsilon^{-1} n \sqrt{n} \log \frac{n}{\varepsilon})$.
\end{proof}

\subsubsection{A Family of Approximation Algorithms for Finding the MAIR in a Convex Set}
\label{subsubsec:Family-Approx-MAIR}
In this section we will improve the complexity of this algorithm by integrating other subroutines. For the special case where $C$ is a convex polygon, the parametrized optimization approach is not restricted to use the  
subroutine presented in Section \ref{subsubsec:Parametric-MAAIR} and any of the existing efficient algorithms from the literature that can find the MAAIR in a convex polygon could be used as a subroutine. The following theorem states the results for the case when the best-known algorithm for finding the MAAIR in convex polygons is used as the subroutine in the parametrized optimization algorithm. 
\begin{thm}
\label{thm:Approx-MAIR-Poly-Alt}
Given a compact convex polygon $C$, a $(1-\varepsilon)$--approximation solution for the problem of finding the maximum area inscribed rectangle (MAIR) in $C$ can be obtained in $\mathcal{O}(\varepsilon^{-1}\log n)$ time with an $\mathcal{O}(\varepsilon^{-1}n)$ pre-processing time.
\end{thm}
\begin{proof}
In the algorithm described in Theorem \ref{thm:main-theorem}, replace the  
subroutine of Section \ref{subsubsec:Parametric-MAAIR} with the algorithm from Alt \etal \cite{alt1995computing} that takes $\mathcal{O}(\log n)$ to find the MAAIR in a convex polygon for any of the $\mathcal{O}(\varepsilon^{-1})$ directions. By Lemma \ref{lem:aspect-ratio-UB}, finding $\overline{\rho}$ requires an $\mathcal{O}(\log n)$ pre-processing time. Also, the rotation of the axes and finding the new coordinates for $n$ vertices of $C$ takes an $\mathcal{O}(n)$ pre-processing time for each of the $\mathcal{O}(\varepsilon^{-1})$ directions.
\end{proof}

This result when combined with a polygonal approximation of a convex set such as the one proposed by Ahn \etal \cite{ahn2006inscribing} can provide a faster algorithm for the case of convex sets as well. The following definitions are due to \cite{ahn2006inscribing}.

\begin{Def}
Let $\mathcal{U}$ be the set of unit vectors in $\mathbb{R}^2$. For each $u\in \mathcal{U}$ and each compact convex set $C$ , the directional width of $C$ in direction $u$, denoted by $\dwidth(u,C)$, is the minimum width of a slab that contains C and is orthogonal to $u$, or in other words:
\[
\dwidth(u,C) = \max_{x\in C} \langle u, x \rangle - \min_{x\in C} \langle u, x \rangle,
\]
where $\langle \cdot,\cdot \rangle$ denotes the inner product.
\end{Def}

\begin{Def}
For a compact convex set $C$ and a parameter $\varepsilon \in (0, 1)$, the convex set $C_{\varepsilon} \subset C$ is an $\varepsilon$-kernel for $C$ if an only if
\[
 \forall u\in \mathcal{U}, \quad (1-\varepsilon)\dwidth(u,C) \leq \dwidth(u,C_{\varepsilon}).
 \]
\end{Def}

Ahn \etal \cite{ahn2006inscribing} present an $\mathcal{O}(\varepsilon^{-1/2}T_C)$ algorithm for computing an $\varepsilon$-kernel of $C$, i.e., a polygonal inner approximation of $C$,  with $\mathcal{O}(\varepsilon^{-1/2})$ vertices, where $T_C$ is the time needed to perform the following two queries on $C$:
\begin{itemize}
\item given a query line $\ell$, find $C\cap\ell$,
\item given a query direction $u$, find an extreme point in $C$ in direction $u$.
\end{itemize}
For instance, when $C$ is a convex $n$-gon given as an array of its vertices in counter-clockwise order, then these two types of queries can be answered in $T_C = \mathcal{O}(\log n)$ time by binary search.
Their only assumption is that the convex set $C$ is given in a data structure that allows these two queries in time $T_C$. 

\begin{cor}
\label{cor:Approx-MAIR-Set-Ahn}
Given a compact convex set $C$, a $(1-\varepsilon)$--approximation solution for the problem of finding the maximum area inscribed rectangle (MAIR) in $C$ can be obtained in $\mathcal{O}(\varepsilon^{-1/2}T_C+\varepsilon^{-1}\log\varepsilon^{-1/2})$ time.
\end{cor}

\begin{proof}
By Lemma 8 of \cite{cabello2016finding}, an $\varepsilon$-kernel for $C$ contains a rectangle $R_k$ with area at least $(1-32\varepsilon)|R_{opt}|$. We first compute an ($\varepsilon/64$)-kernel for $C$ with $\mathcal{O}((\varepsilon/64)^{-1/2})=\mathcal{O}(\varepsilon^{-1/2})$ vertices in time $\mathcal{O}((\varepsilon/64)^{-1/2}T_C)=\mathcal{O}(\varepsilon^{-1/2}T_C)$. This gives $|R_k|\geq (1-\varepsilon/2)|R_{opt}|$.
We then apply the result of Theorem \ref{thm:Approx-MAIR-Poly-Alt} with precision $\varepsilon/2$ to the polygon $C_{\varepsilon/64}$ in time $\mathcal{O}((\varepsilon/2)^{-1}\log (\varepsilon/2)^{-1/2})=\mathcal{O}(\varepsilon^{-1}\log \varepsilon^{-1/2})$. This gives $|R_{apx}|\geq (1-\varepsilon/2)|R_k|\geq (1-\varepsilon/2)^2|R_{opt}|\geq (1-\varepsilon)|R_{opt}|$.
\end{proof}

\begin{cor}
\label{cor:Approx-MAIR-Poly-Ahn}
Given a compact convex polygon $C$, a $(1-\varepsilon)$--approximation solution for the problem of finding the maximum area inscribed rectangle (MAIR) in $C$ can be obtained in $\mathcal{O}(\varepsilon^{-1/2}\log n+\varepsilon^{-1}\log\varepsilon^{-1/2})$ time.
\end{cor}

\begin{proof}
In this case we have $T_C=\mathcal{O}(\log n)$.
\end{proof}

\section{Evaluation and Illustrative Examples}
\label{sec:EvaluationAndExamples}
In this paper we considered both problems of finding MVAIR (MAAIR) and MVIR (MAIR) inside a convex set. Unlike most of the literature, that considers this set to be a polygon (a set of linear inequalities), this paper relaxes this constraint and allows the set to be any geometric convex body that is expressible in a finite number of convex inequalities. Such convex sets include polytopes (polygons), ellipsoids, and the intersection of convex sets such as the intersection of ellipses and halfspaces or the bounded intersection of the epigraph of convex parabolas with ellipses and halfspaces. The MVIR problem is formulated as a non-convex optimization problem, while a convex optimization problem is developed for the MVAIR problem, both models in higher dimensions. The models can also be easily generalized to the cases of finding other inscribed geometric shapes. The proposed algorithm finds a $(1-\varepsilon)$--approximation to the MVAIR in a convex set in $\mathcal{O}((d^3+d^2 n) \sqrt{n}\log \frac{n}{\varepsilon})$ time, where $d$ is the dimension, $n$ is the number of inequalities defining the convex set.  
For finding MAIR in a 2D convex set, a parametric approach is developed that helps us to find  
a $(1-\varepsilon)$--approximation to the
MAAIR in any given direction of axes in  
$\mathcal{O}(n \sqrt{n}\log \frac{n}{\varepsilon})$ time, which can be used as a subroutine to compute a $(1-\varepsilon)$--approximation to MAIR in  
$\mathcal{O}(\varepsilon^{-1} n \sqrt{n} \log \frac{n}{\varepsilon})$ time.
Using a geometric approach we improve this to $\mathcal{O}(\varepsilon^{-1/2}T_C+\varepsilon^{-1}\log\varepsilon^{-1/2})$. For the special case of convex polygons our geometric algorithms with time complexities of $\mathcal{O}(\varepsilon^{-1}\log n)$ and $\mathcal{O}(\varepsilon^{-1/2}\log n+\varepsilon^{-1}\log\varepsilon^{-1/2})$ work faster than Cabello \etal's algorithm \cite{cabello2016finding}, which is the previous best known result.

To our knowledge, except Amenta's model \cite{amenta1994bounded} for the MVAIR problem and the model of Cabello \etal \cite{cabello2016finding} for the MAIR problem where $C$ is restricted to be a convex polygon, no other optimization-based algorithm is published so far for the problems discussed in this paper. Furthermore, except for Amenta's model for MVAIR \cite{amenta1994bounded} and a brief discussion on MVIR in convex polytopes in Cabello \etal \cite{cabello2016finding}, we are not aware of any other model or algorithm for the higher dimensional problems. Our optimization model for the MVAIR problem is much more efficient than Amenta's model as it reduces the $\mathcal{O}(n2^d)$ number of constraints to $\mathcal{O}(n)$ constraints. Similar randomized algorithms, such as \cite{sharir1992combinatorial,matouvsek1996subexponential}, could be used to solve it in \emph{expected} linear (in $n$) time as suggested in \cite{amenta1994bounded}.

 While none of the existing algorithms can solve both general and axis-aligned problems, our parametrized optimization approach can do so in 2D. Moreover, except for Cabello \etal's algorithm \cite{cabello2016finding}, no other algorithm is published so far for either of these problems capable of dealing with a broader spectrum of geometric convex sets other than polygons. Also, while there has been no algorithm in the literature capable of finding MAAIR in a given direction without rotating axes, our optimization method computes the MAAIR in any given direction in both convex polygons and general convex sets.  
Our optimization algorithm is unique regarding the dependence of its performance on the number of inequalities defining the convex set; although for a convex polygon, this is $n$, for an ellipse, this is just one. For the general convex sets, it is difficult to compare its performance to geometric algorithms such as our algorithm presented in Corollary \ref{cor:Approx-MAIR-Set-Ahn} and that of Cabello \etal \cite{cabello2016finding}, due to the dependance of their performance to the unknown factor $T_C$ that depends on the geometry of $C$. For the special case of convex polygons our geometric algorithms outperform the optimization approach due to their simplicity. Nevertheless, the importance of the presented optimization approach is in the uniform framework that it provides for solving a variety of geometric shape approximation problems. For rectangular approximation, as illustrated in Tables \ref{tab:ComputationalComplexity2D} and \ref{tab:ComputationalComplexityHD}, this approach is capable of solving 8 out 10 existing problems. We conjecture that it could be also applied to the remaining two problems, i.e., MVIR problem for polytopes and general convex sets. Tables \ref{tab:ComputationalComplexity2D} and \ref{tab:ComputationalComplexityHD} 
summarize our results across all discussed problems and their comparison with the best existing algorithms. Our optimization approach is capable of solving a wide range of problems and our geometric algorithms, from Theorem \ref{thm:Approx-MAIR-Poly-Alt} and Corollaries \ref{cor:Approx-MAIR-Set-Ahn} and \ref{cor:Approx-MAIR-Poly-Ahn}, improve the existing results. Figure \ref{fig:ComplexityComparison} illustrates the comparison of the computational complexity of our geometric algorithms with that of Cabello \etal \cite{cabello2016finding} for the MAIR problem in convex polygons.

\begin{figure}[ht]
  \centering
  \includegraphics[width=0.45\textwidth]{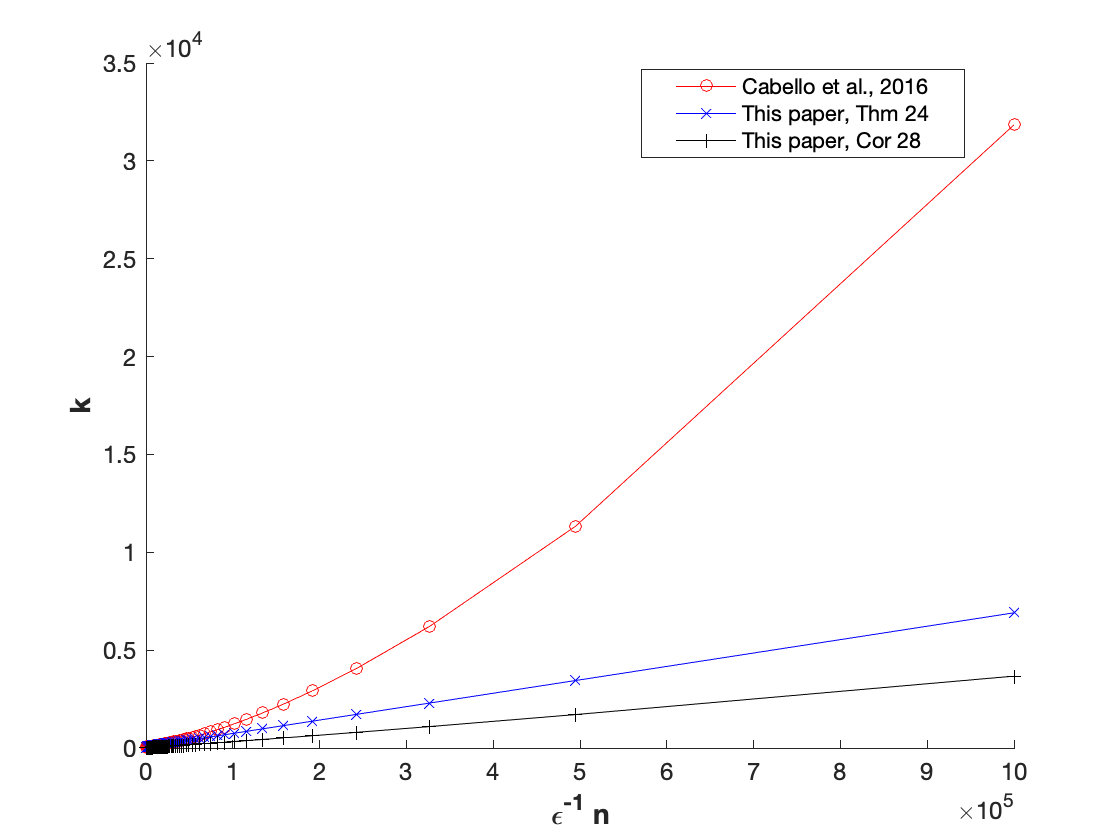}
  \vspace{3pt}
  \caption{A comparison of the computational complexity, for $n\in\{10,20,...,1000\}$ and $\varepsilon\in\{0.1,0.099,...,0.001\}$, between algorithms presented by Theorem \ref{thm:Approx-MAIR-Poly-Alt} and Corollary \ref{cor:Approx-MAIR-Poly-Ahn} and the algorithm of Cabello \etal \cite{cabello2016finding}. The value $k$ in the $y$-axis shows the order of magnitude of the number of operations required by each of the algorithms to solve the MAIR problem in a convex polygon.}
  \label{fig:ComplexityComparison}
\end{figure}

\begin{sidewaystable}
  \centering
\scriptsize
  \caption{Comparison of the computational complexity of the most competitive algorithms in $\mathbb{R}^2$}
    \begin{tabular}{ccccccccc}
    \toprule
    \multirow{3}[6]{*}{} & \multirow{3}[6]{*}{\bfseries Approach} & \multicolumn{7}{c}{\bfseries Computational Complexity} \\
  \cmidrule{3-9}         &       & \multicolumn{3}{c}{\bfseries 2D-Covex Polygon} & 	$\quad$		& \multicolumn{3}{c}{\bfseries 2D-Convex Set}  \\
     \cmidrule{3-9}
          &       & \textbf{MAAIR} & \textbf{MAAIR}-$\bm{\mbox{F}_{\mbox{dir}}}$ & \textbf{MAIR} &		 & \textbf{MAAIR} &  \textbf{MAAIR}-$\bm{\mbox{F}_{\mbox{dir}}}$ & \textbf{MAIR}   \\
        \cmidrule{3-5} \cmidrule{7-9} 
          Alt \etal \cite{alt1995computing}    & CG \footnote{\tiny \;Computational Geometric (CG) algorithm(s)}    &   $\mathcal{O}(\log n)$    &    $\mathcal{O}(n)$   &   --    &  	&   --   &    --   &   --  \\ 
        \cmidrule{3-5} \cmidrule{7-9}
         Knauer \etal \cite{knauer2012largest}  & CG    &   --    &   --    &  $\mathcal{O}(\varepsilon^{-2}\log n)$   &		&   --    &    --   &   --  \\ 
     \cmidrule{3-5} \cmidrule{7-9}
    \multirow{2}[4]{*}{ Cabello \etal \cite{cabello2016finding} } & \multirow{2}[4]{*}{CG-OPT \footnote{\tiny \;Computational Geometric (CG) and Optimization (OPT) algorithm(s) with more focus on CG}} & \multirow{2}[4]{*}{--} & \multirow{2}[4]{*}{--} &    $\mathcal{O}(n^3)$,  &		 & \multirow{2}[4]{*}{--} & \multirow{2}[4]{*}{--} &   \multirow{2}[4]{*}{$\mathcal{O}(\varepsilon^{-3/2}+\varepsilon^{-1/2}T_C)$}  \\
         &       &       &       &   $\mathcal{O}(\varepsilon^{-3/2}+\varepsilon^{-1/2}\log n)$  &		  &       &       &      \\ 
     \cmidrule{3-5} \cmidrule{7-9}
    \multirow{2}[4]{*}{ This work} & \multirow{2}[4]{*}{OPT-CG \footnote{\tiny \;Optimization (OPT) and Computational Geometric (CG) algorithm(s) with more focus on OPT}} & \multirow{2}[4]{*}{$\mathcal{O}(n\sqrt{n}\log \frac{n}{\varepsilon})$} & \multirow{2}[4]{*}{$\mathcal{O}(n\sqrt{n}\log \frac{n}{\varepsilon})$} &   $\mathcal{O}(\varepsilon^{-1}n\sqrt{n}\log \frac{n}{\varepsilon})$,  &	  & \multirow{2}[4]{*}{$\mathcal{O}(n\sqrt{n}\log \frac{n}{\varepsilon})$} & \multirow{2}[4]{*}{$\mathcal{O}(n\sqrt{n}\log \frac{n}{\varepsilon})$} &    $\mathcal{O}(\varepsilon^{-1}n\sqrt{n}\log \frac{n}{\varepsilon})$,   \\
	 &       &       &       &    $\mathcal{O}(\varepsilon^{-1}\log n)$   &    &		   &       &       \\
	  &       &       &       &    $\mathcal{O}(\varepsilon^{-1/2}\log n+\varepsilon^{-1}\log \varepsilon^{-1/2})$   &    &		   &       &  $\mathcal{O}(\varepsilon^{-1/2}T_C+\varepsilon^{-1}\log \varepsilon^{-1/2})$       \\
    \bottomrule
    \end{tabular}
  \label{tab:ComputationalComplexity2D}

\vspace{5\baselineskip}
  \centering
\scriptsize
  \caption{Comparison of the computational complexity of the most competitive algorithms in $\mathbb{R}^d$}
    \begin{tabular}{ccccccc}
    \toprule
    \multirow{3}[6]{*}{} & \multirow{3}[6]{*}{\bfseries Approach} & \multicolumn{5}{c}{\bfseries Computational Complexity} \\
  \cmidrule{3-7}         &       &  \multicolumn{2}{c}{\bfseries HD-Convex Polytope} &	$\qquad$	& \multicolumn{2}{c}{\bfseries HD-Convex Set} \\
 \cmidrule{3-7}
          &       &  \textbf{MVAIR} & \textbf{MVIR}  &		& \textbf{MVAIR} & \textbf{MVIR} \\
          \cmidrule{3-4} \cmidrule{6-7}
                Alt \etal \cite{alt1995computing}    & CG    &     --     &   --    &		&   --    & --  \\
           \cmidrule{3-4} \cmidrule{6-7}
           Knauer \etal \cite{knauer2012largest}     & CG    &    --    &   --    &		&   --    & --  \\
       \cmidrule{3-4} \cmidrule{6-7}
       Cabello \etal \cite{cabello2016finding}    & CG-OPT    &     --     &   --    &		&   --    & --  \\
       \cmidrule{3-4} \cmidrule{6-7}
      This work    & OPT-CG   &     $\mathcal{O}(d^{2}n\sqrt{n}\log \frac{n}{\varepsilon})$     &   --    &		&   $\mathcal{O}((d^3+d^{2}n)\sqrt{n}\log \frac{n}{\varepsilon})$    & --  \\
    \bottomrule
    \end{tabular}
  \label{tab:ComputationalComplexityHD}
\end{sidewaystable}

It should be also noted that our optimization approach is unique in dealing with general compact convex sets, such as polygons, ellipses, and the bounded intersection of convex sets, directly without approximating them with a polygonal $\varepsilon$-kernel first. To illustrate this capability and to provide further insight on the behavior of the objective function of the parametrized optimization model \ref{mod:MAIR_Parametric}, discussed in Section \ref{subsubsec:Parametric-Approach-2D}, we end this section by providing some examples of the optimization results for the MAAIR, MAAIR-$\mbox{F}_{\mbox{dir}}$, and MAIR problems in a variety of compact convex sets. Results for two given polygons and two random polygons are shown in Figure \ref{fig:polygon6&4} and Figure \ref{fig:polygonRandom1&2}, respectively. They also show the ill behavior of the objective function $f(t)$ of the Model \ref{mod:MAIR_Parametric}. Examples of axis-aligned rectangles for the regular case and the case with given directions (equivalent of rotated axes) are shown in Figure \ref{fig:LIAR}. Examples for ellipses and the intersection of some convex sets are shown in Figures \ref{fig:Ellipse} and \ref{fig:intersection}, respectively.

\begin{figure}[htb]
\begin{centering}
\subfloat[\label{fig:6a}]{ \begin{centering}\includegraphics[width=0.35\columnwidth]{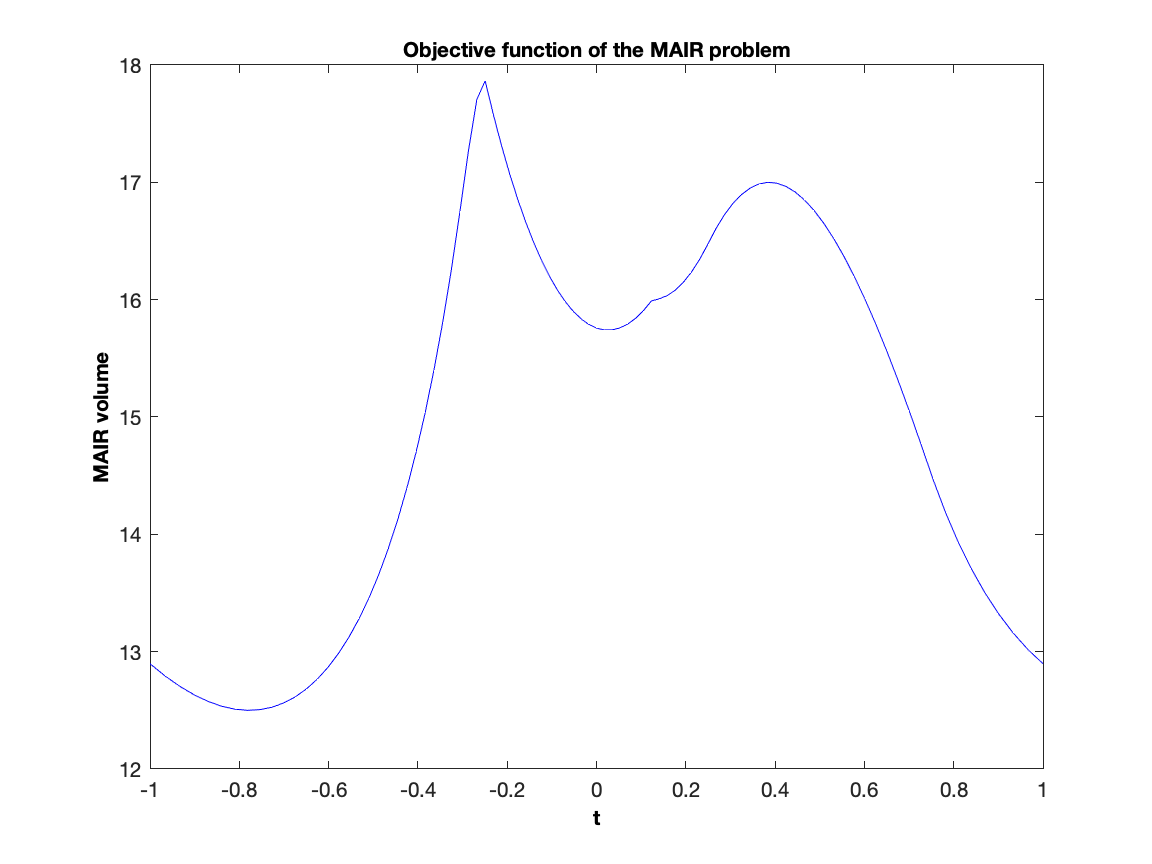}\par\end{centering} }
\subfloat[\label{fig:6b}]{ \begin{centering}\includegraphics[width=0.35\columnwidth]{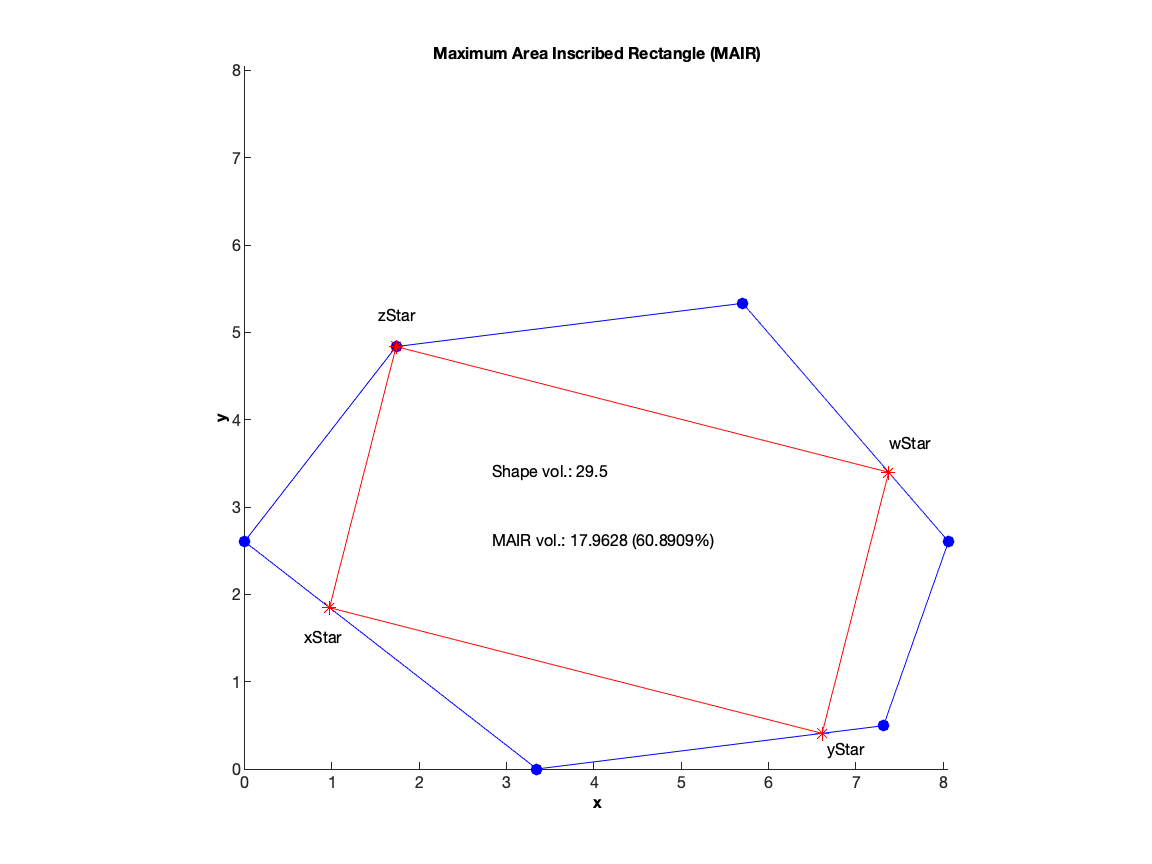}\par\end{centering} }
\par\end{centering}

\begin{centering}
\subfloat[\label{fig:4a}]{ \begin{centering}\includegraphics[width=0.35\columnwidth]{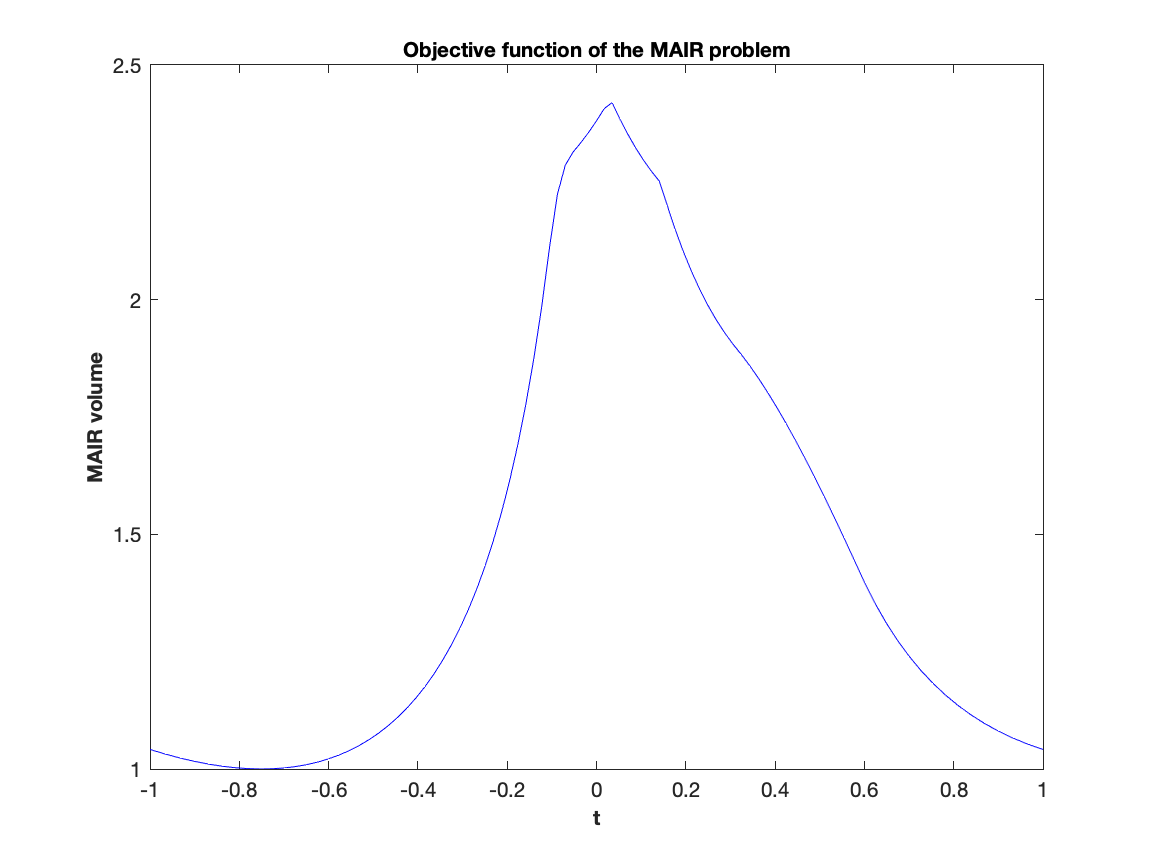}\par\end{centering} }
\subfloat[\label{fig:4b}]{ \begin{centering}\includegraphics[width=0.35\columnwidth]{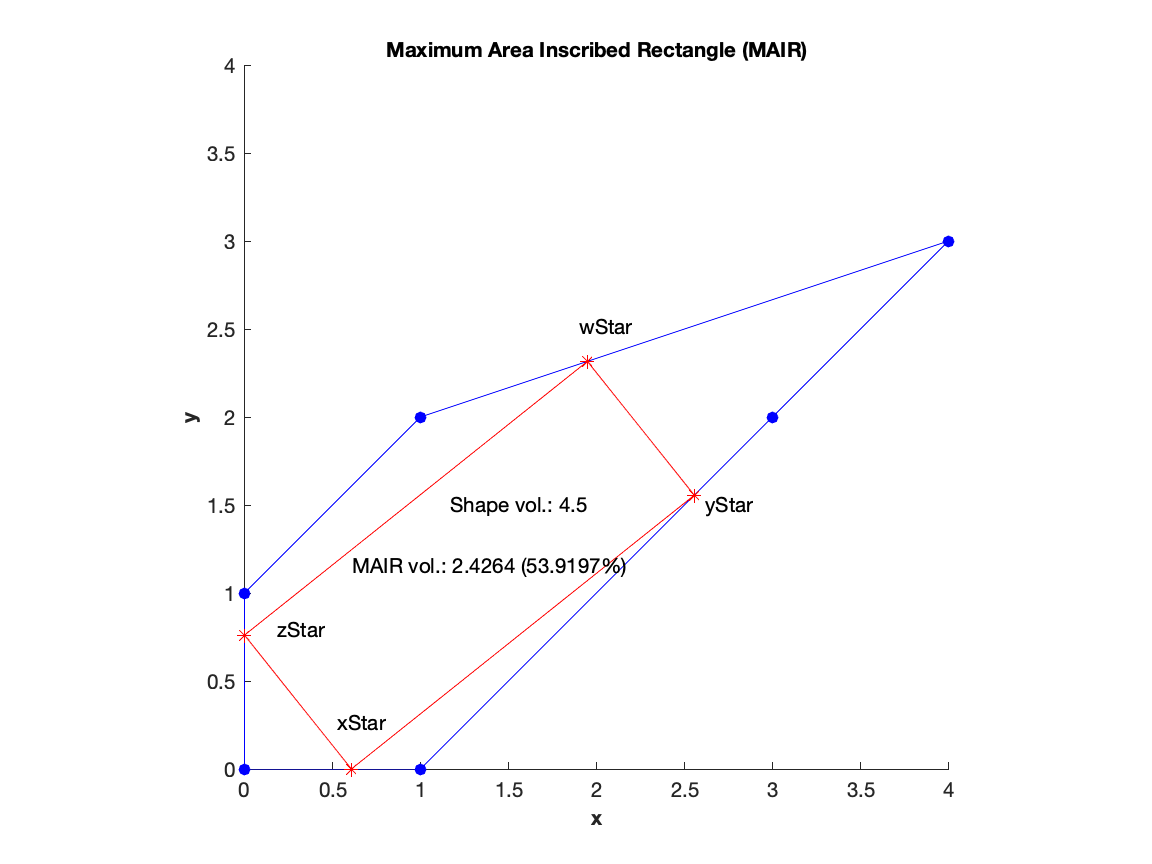}\par\end{centering} }
\par\end{centering}

\caption{\label{fig:polygon6&4}The largest inscribed rectangle in two given polygons. The objective functions are shown in (\ref{fig:6a}) and (\ref{fig:4a}). It can be seen that $f(-1)=f(1)$ but $f(t)$ is not necessarily symmetric or even unimodal over $-1\leq t \leq t$. The largest inscribed rectangles are obtained in (\ref{fig:6b}) and (\ref{fig:4b}), using the algorithm described in Section \ref{subsubsec:Approx_MAIR}. ``Shape vol.'' shows the area (volume in general) of the polygon and ``MAIR vol.'' show the area of the largest inscribed rectangle and the percentage of this area to the area of the polygon. Figure (\ref{fig:6b}) shows the MAIR with one vertex-corner and three edge-corners, while the MAIR in Figure (\ref{fig:4b}) has four edge-corners. }
\end{figure}

\begin{figure}[h]
\begin{centering}
\subfloat[\label{fig:Random1a}]{ \begin{centering}\includegraphics[width=0.35\columnwidth]{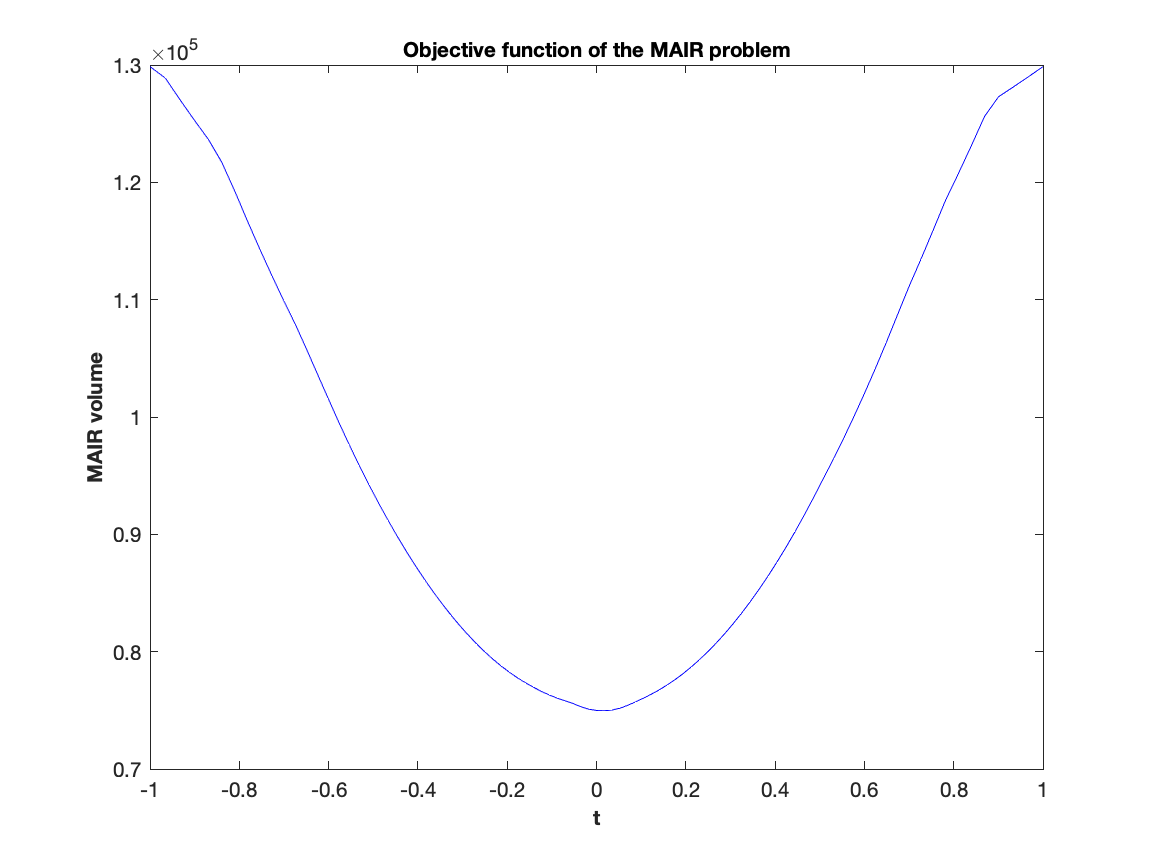}\par\end{centering} }
\subfloat[\label{fig:Random1b}]{ \begin{centering}\includegraphics[width=0.35\columnwidth]{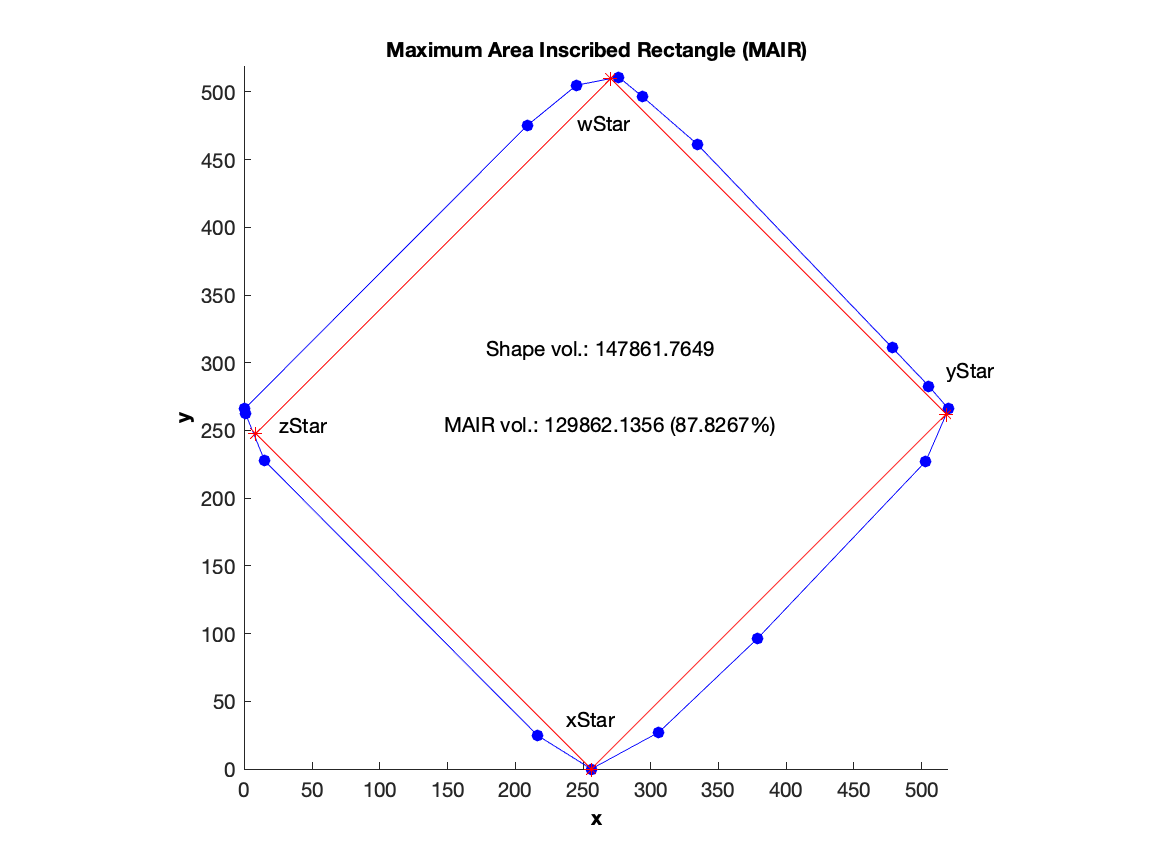}\par\end{centering} }
\par\end{centering}

\begin{centering}
\subfloat[\label{fig:Random2a}]{ \begin{centering}\includegraphics[width=0.35\columnwidth]{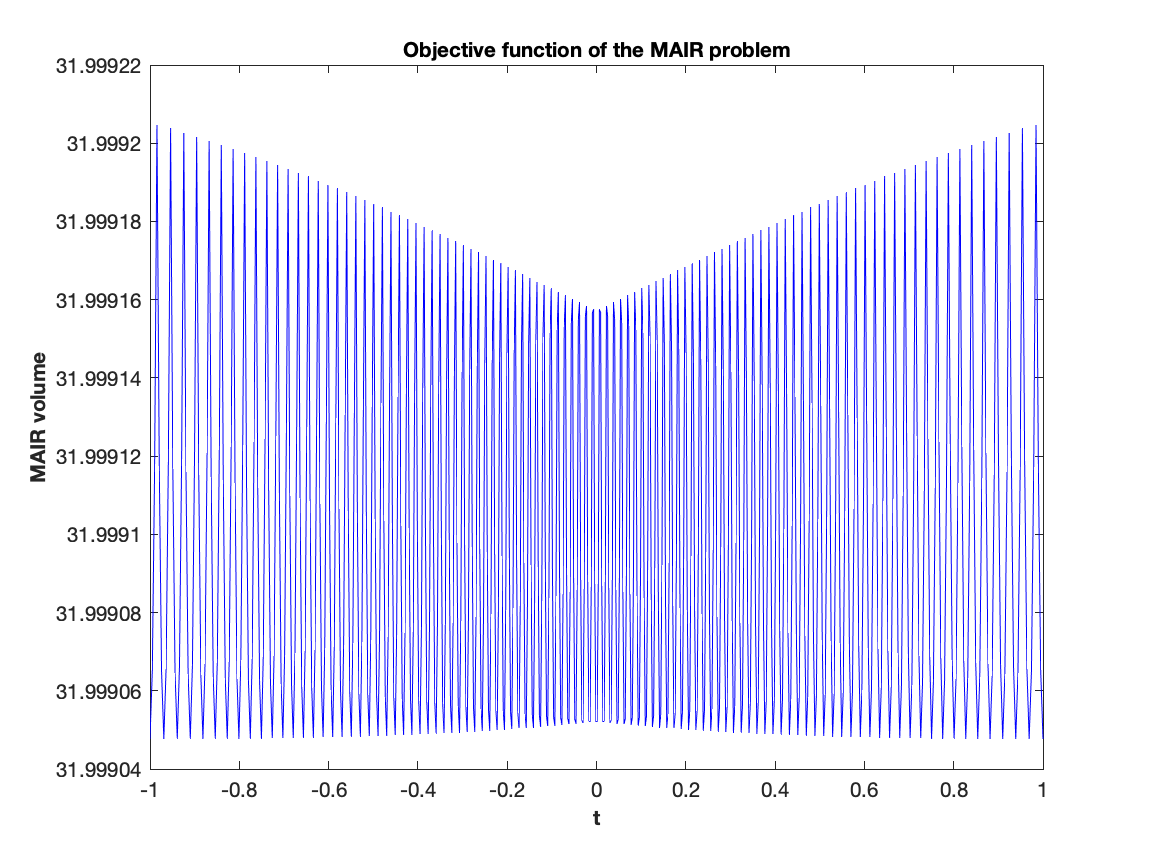}\par\end{centering} }
\subfloat[\label{fig:Random2b}]{ \begin{centering}\includegraphics[width=0.35\columnwidth]{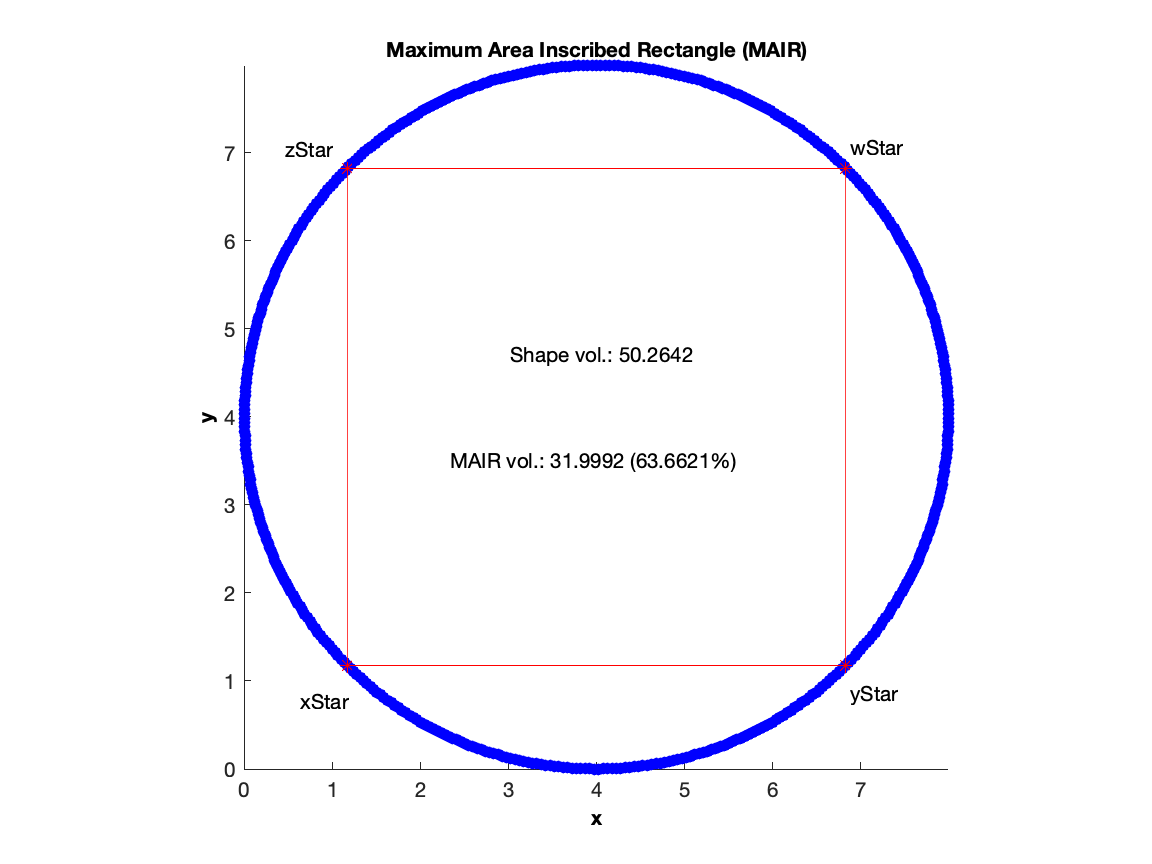}\par\end{centering} }
\par\end{centering}

\caption{\label{fig:polygonRandom1&2}The largest inscribed rectangle in a random polygon with 14 vertices and a regular 500-gon randomly generated on a circle. The objective functions are shown in (\ref{fig:Random1a}) and (\ref{fig:Random2a}), which shows the contrast between a well-behaved unimodal quasiconvex and almost symmetric function and an ill-behaved non-smooth function.  The largest inscribed rectangles are obtained in (\ref{fig:Random1b}) and (\ref{fig:Random2b}). It should be noted that the largest inscribed rectangle inside a circle is a square and the fraction of the area at optimality is $2/\pi \simeq 0.63662$.}
\end{figure}

\begin{figure}[h]
\begin{centering}
\subfloat[\label{fig:LIAR1a}]{ \begin{centering}\includegraphics[width=0.35\columnwidth]{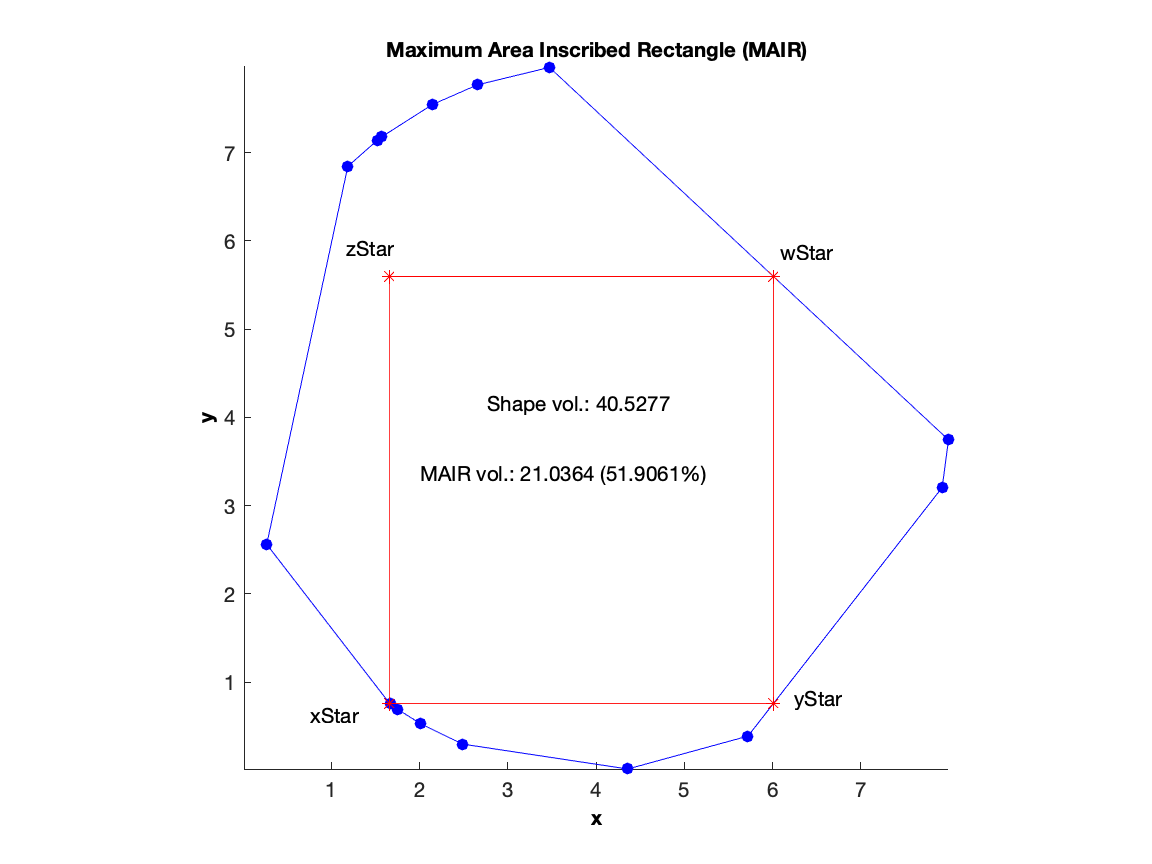}\par\end{centering} }
\subfloat[\label{fig:LIAR1b}]{ \begin{centering}\includegraphics[width=0.35\columnwidth]{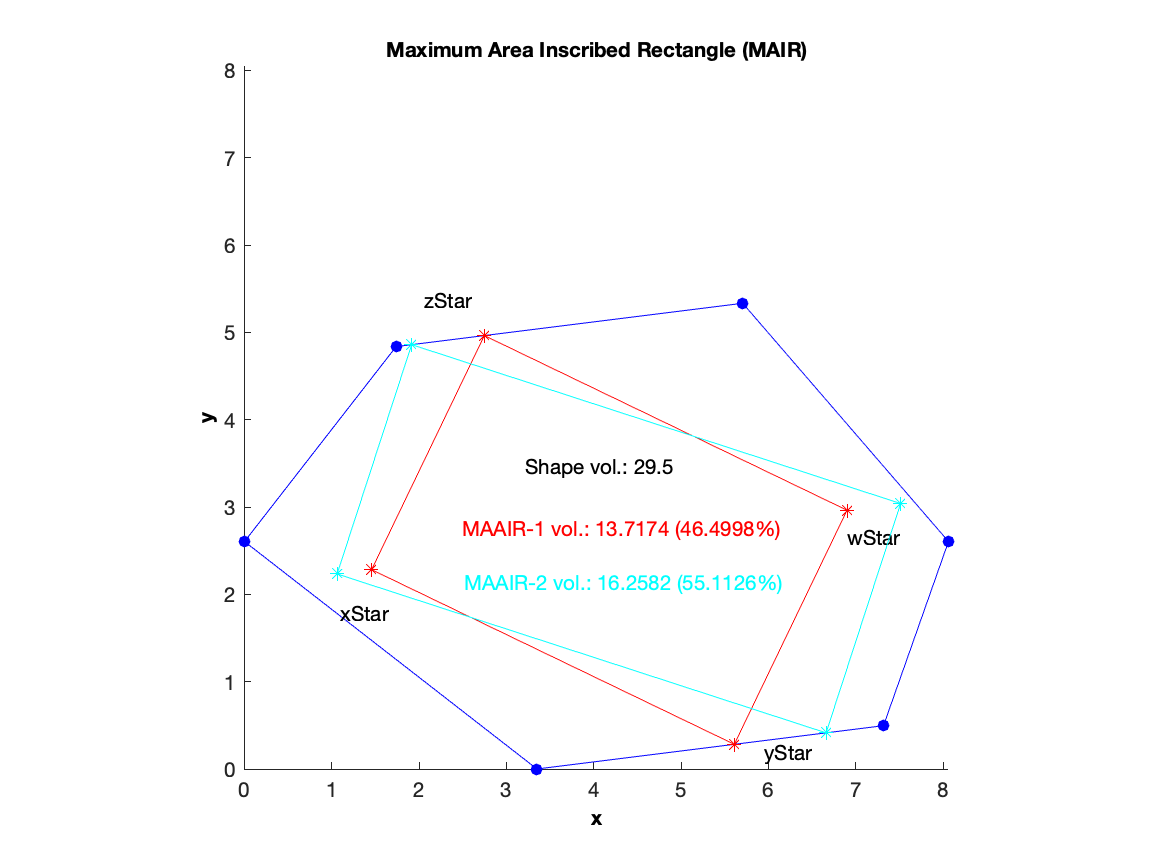}\par\end{centering} }

\par\end{centering}

\caption{\label{fig:LIAR}The maximum area axis-aligned inscribe rectangles for the regular axes in a random 15-gon generated on a circle in (\ref{fig:LIAR1a}) and for two given directions in a given polygon in (\ref{fig:LIAR1b}). Note in Figure (\ref{fig:LIAR1b}) that the conditions of Theorem \ref{thm:optPropertiesPolygon} may not hold for the maximum area rectangles for given directions as they may not be optimal regarding all directions.}
\end{figure}

\begin{figure}[h]
\begin{centering}
\subfloat[\label{fig:Ellipse-a}]{ \begin{centering}\includegraphics[width=0.35\columnwidth]{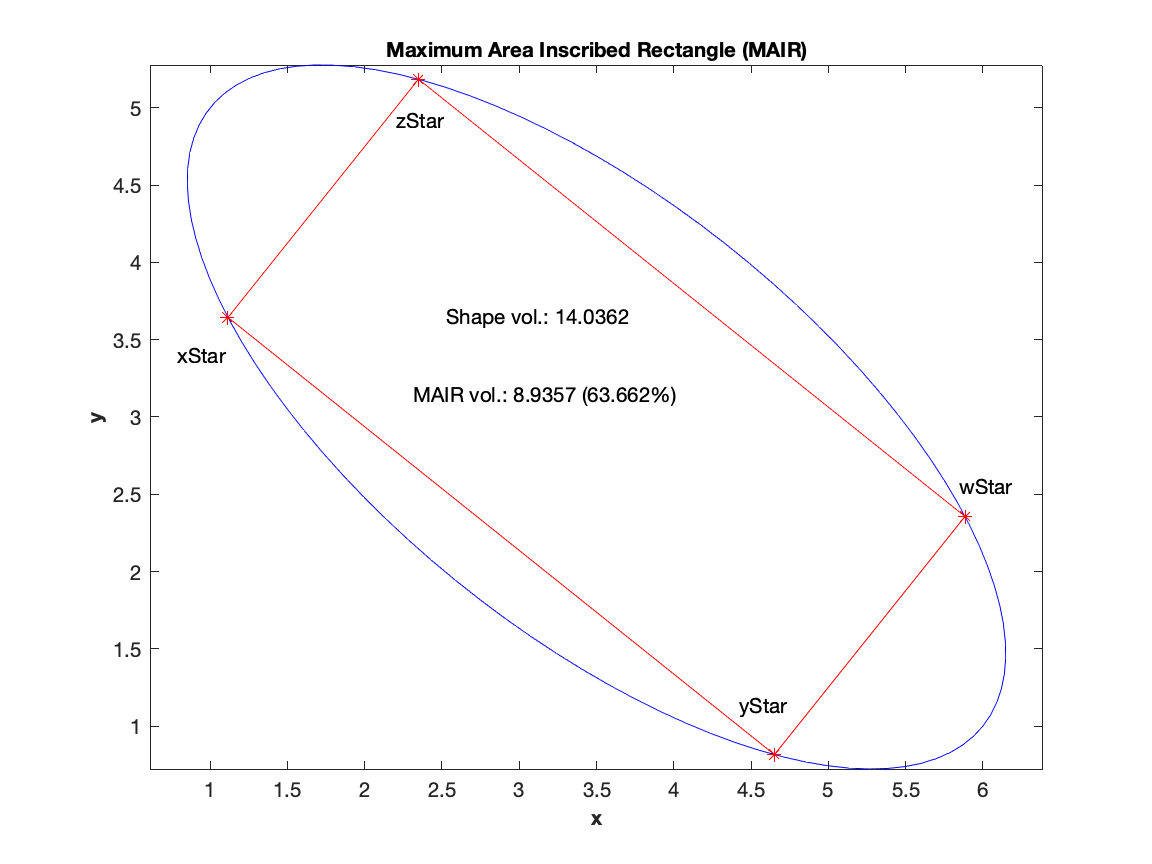}\par\end{centering} }
\subfloat[\label{fig:Ellipse-b}]{ \begin{centering}\includegraphics[width=0.35\columnwidth]{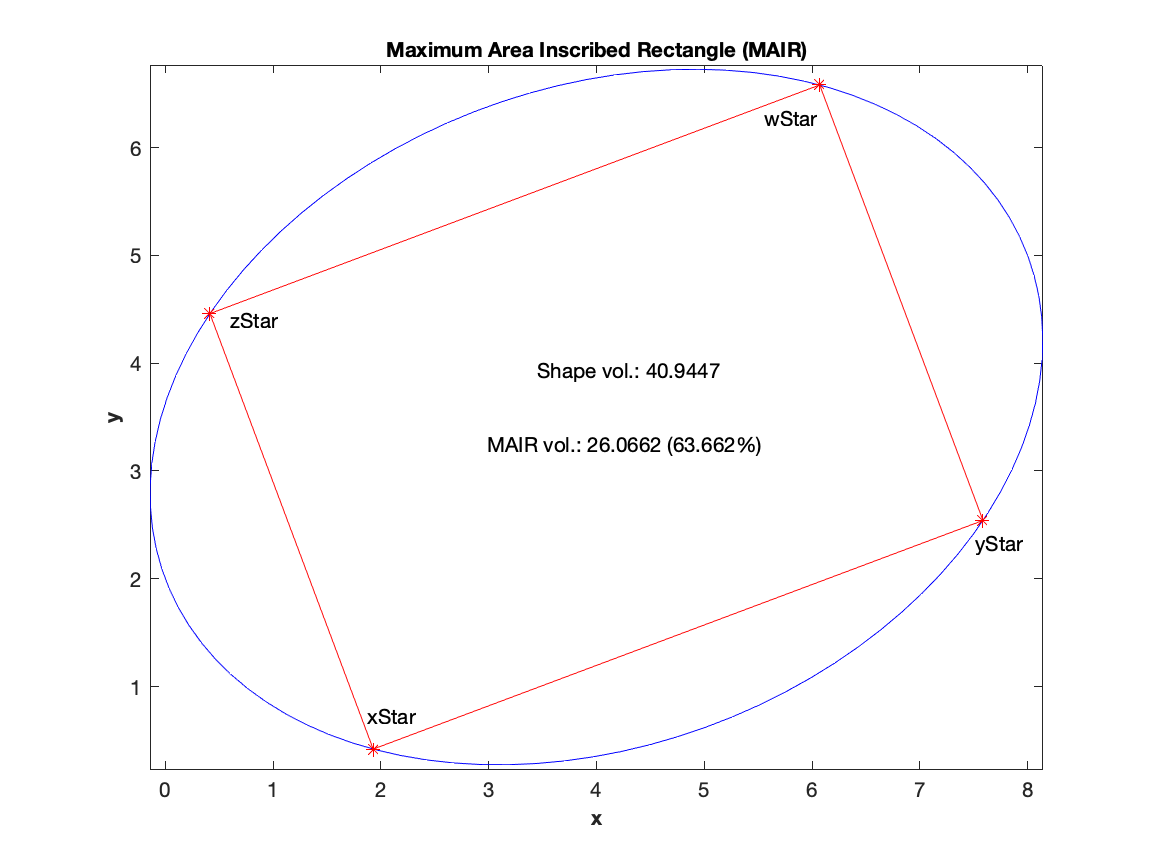}\par\end{centering} }
\par\end{centering}

\caption{\label{fig:Ellipse}The largest inscribed rectangles in two given ellipses. Notably, the fraction of the area of the MAIR to the area of the ellipse is the same as  that of circles ($2/\pi \simeq 0.63662$), which can be verified by the elementary calculus.}
\end{figure}

\begin{figure}[h]
\begin{centering}
\subfloat[\label{fig:intersection-a}]{ \begin{centering}\includegraphics[width=0.35\columnwidth]{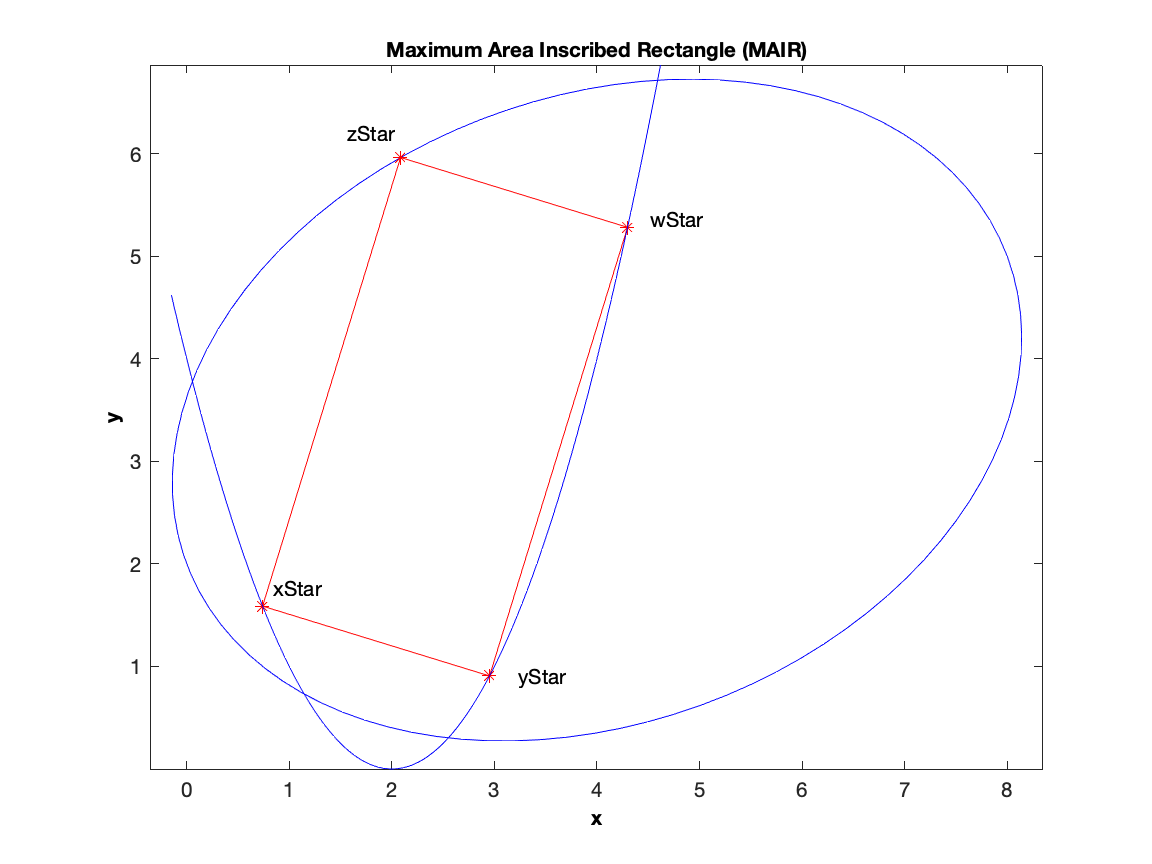}\par\end{centering} }
\subfloat[\label{fig:intersection-b}]{ \begin{centering}\includegraphics[width=0.35\columnwidth]{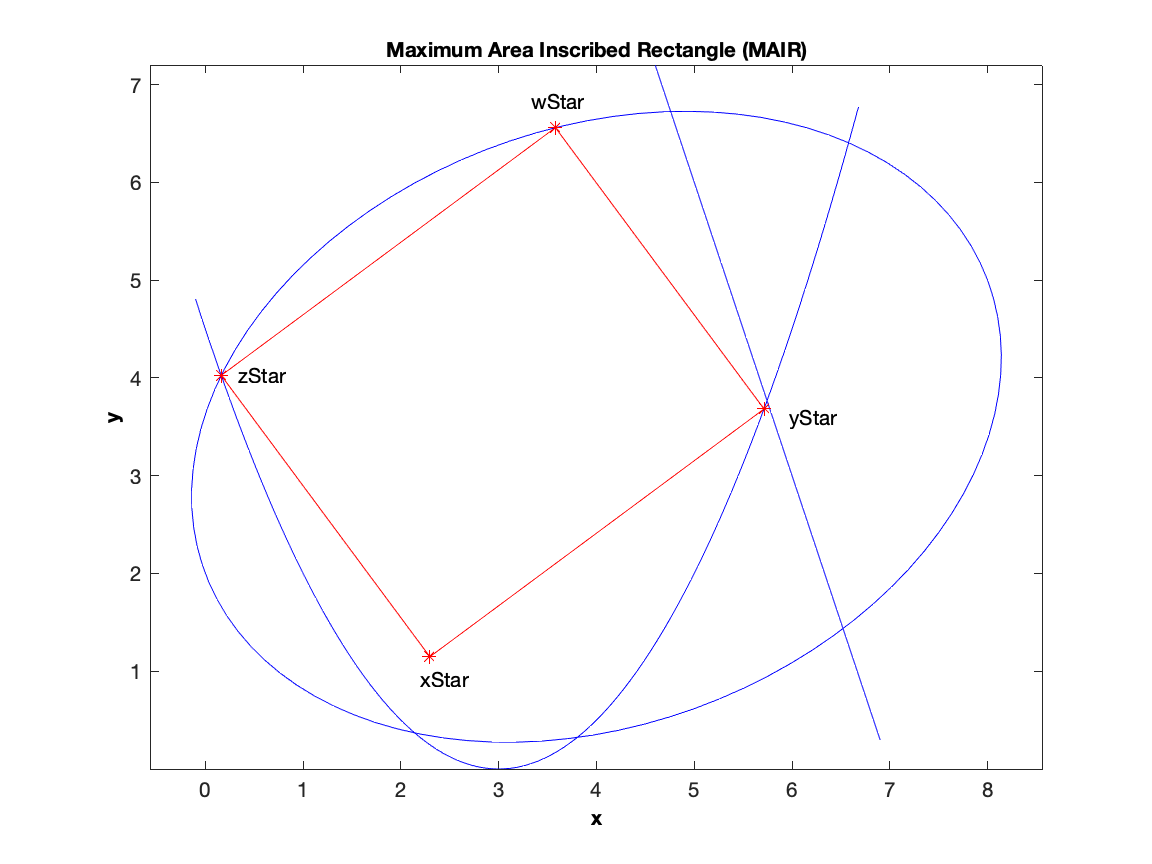}\par\end{centering} }
\par\end{centering}

\caption{\label{fig:intersection}The largest inscribed rectangle in the intersection of an ellipse and a parabola is presented in (\ref{fig:intersection-a}) and for the intersection of an ellipse, a parabola, and a half-space is shown in (\ref{fig:intersection-b}). We see that the conditions of the Theorem \ref{thm:optPropertiesPolygon} hold for these particular non-polygonal examples.}
\end{figure}

\section{Conclusions}
We have presented several optimization models for the problem of finding the maximum volume (axis-aligned) inscribed rectangle in a convex set defined by a finite number of convex inequalities. We presented efficient  
(1-$\varepsilon$)--approximation algorithms for the MVAIR, MAAIR, and MAIR problems. The running times of these algorithms only depend on the dimension and the number of convex inequalities that define the convex set and therefore is agnostic to the geometric structure of the convex set. 
We have also analyzed the optimal properties of the MAIR in convex polygons, centrally symmetric convex sets, and axially symmetric convex sets.
One future research direction is to explore the optimal properties of the MVIR problem in the higher dimension and  
to develop efficient algorithms for finding the MVIR. Another potential direction for future research would be the consideration of inscribing other geometric shapes in a convex set or approximating a convex set that contains ``holes'' with multiple inscribed rectangles.

\section*{Acknowledgment} The author  
is indebted to Shuzhong Zhang and Gilad Lerman for their various helpful comments and key observations which helped the derivation of some of the results.


\bibliographystyle{myIEEEtran}
\bibliography{IEEEabrv,LIR_Refs.bib}

\end{document}